\documentclass{vldb}
\usepackage{tikz}
\usepackage{pgfplots}
\usetikzlibrary{patterns}

\usepackage{threeparttable}
\usepackage{url}
\usepackage{amsmath}
\usepackage{amssymb}
\usepackage{amsfonts}
\usepackage{breakurl}
\usepackage{bbding}
\usepackage{tabularx}
\usepackage{multirow}
\usepackage{makecell}
\usepackage[T1]{fontenc}
\usepackage{floatrow}
\usepackage{graphicx}
\usepackage{balance}  
\usepackage{soul}
\usepackage{txfonts}
\usepackage{times}
\usepackage{epsfig}
\usepackage[linesnumbered,ruled,noend]{algorithm2e}
\usepackage[font=small,skip=0pt]{caption}
\usepackage[font=small,skip=0pt]{subcaption}
\usepackage{indentfirst}
\usepackage[noend]{algorithmic}
\newcommand{\nop}[1]{}
\newtheorem{definition}{\bf Definition}
\newtheorem{problem statement}{\bf Problem Statement}
\newtheorem{theorem}{\bf Theorem}
\newtheorem{lemma}{Lemma}
\newcommand{\Paragraph}[1]{~\vspace*{-0.9\baselineskip}\\{\bf #1}}

\makeatletter

\newcommand{\iRm}[1]{\expandafter\@slowromancap\romannumeral #1@}
\newcommand{\seqpar}[1]{\overset{#1}{\preccurlyeq}}

\renewcommand{\footnotesize}{\small}

\makeatother
\tikzstyle{vertex}=[circle, draw, scale=0.5, transform shape] 
\tikzstyle{hvertex}=[circle, draw, scale=0.43, transform shape]
\tikzstyle{minushvertex}=[circle, draw, scale=0.40, transform shape]
\tikzstyle{keyword}=[circle, draw, scale=0.40, transform shape]
\usepackage{verbatim}
\usetikzlibrary{shapes}
\usetikzlibrary{arrows}

\usepackage{pifont}
\newcommand{\cmark}{\ding{51}}%
\newcommand{\xmark}{\ding{55}}%

\begin{document}

\newcommand{\iL}{\mathbb{L}}
\newcommand{\iLa}{\mathbb{L}_{\alpha}}

\newcommand{\arxivshow}[1]{#1}
\newcommand{\submitshow}[1]{}
\newcommand{\arxivcompile}[1]{#1}

\newcommand{\picfolder}{./pics/}
\newcommand{\expfolder}{./exp/}
\newcommand{\stockfolder}{./exp/stock/}
\newcommand{\genefolder}{./exp/gene/}
\newcommand{\powerfolder}{./exp/power/}
\newcommand{\syntheticfolder}{./exp/synthetic/}
\newcommand{\stockafolder}{./exp/stock_app/}
\newcommand{\geneafolder}{./exp/gene_app/}
\newcommand{\powerafolder}{./exp/power_app/}
\newcommand{\syntheticafolder}{./exp/synthetic_app/}

\arxivcompile{
	\renewcommand{\picfolder}{./}
	\renewcommand{\expfolder}{./}
	\renewcommand{\stockfolder}{./}
	\renewcommand{\genefolder}{./}
	\renewcommand{\powerfolder}{./}
	\renewcommand{\syntheticfolder}{./}
	\renewcommand{\stockafolder}{./}
	\renewcommand{\geneafolder}{./}
	\renewcommand{\powerafolder}{./}
	\renewcommand{\syntheticafolder}{./}
}

\title{Computing Longest Increasing Subsequences over Sequential Data Streams}
%
%
%
%
%

\numberofauthors{3} 
%

\author{%
{\normalfont {Youhuan Li${^\dag}$}, Lei Zou{${^\dag}$},  Huaming Zhang{${^\ddag}$},  Dongyan Zhao{${^{\dag}}$}}%
\\
\fontsize{10}{10}\selectfont\itshape $~^{\dag}$Peking University, China;
\fontsize{10}{10}\selectfont\itshape $~^{\ddag}$University of Alabama in Huntsville,USA
\\
\fontsize{9}{9}\selectfont\ttfamily\upshape $~^{\dag}$$\{$liyouhuan,zoulei,zhaody$\}$@pku.edu.cn, $~^{\ddag}$hzhang@cs.uah.edu
\\}

\maketitle

\normalsize

\captionsetup[figure]{font=small,skip=0pt}

\newcommand{\citeAPPexp}{A }
\newcommand{\citeAPPproof}{B }
\newcommand{\citeAPPenum}{C }
\newcommand{\citeAPPdelete}{D }
\newcommand{\citeAPPgap}{E }
\newcommand{\citeAPPweight}{F }
\newcommand{\citeAPPsorted}{G }
\newcommand{\citeAPPconstraints}{H }


\begin{abstract}

In this paper, we propose a data structure, a quadruple neighbor list (QN-list, for short), to support real time queries of all \underline{l}ongest \underline{i}ncreasing \underline{s}ubsequence (LIS) and LIS with constraints over sequential data streams. The QN-List built by our algorithm requires $O(w)$ space, where $w$ is the time window size. The running time for building the initial QN-List takes $O(w \log w)$ time.
Applying the QN-List, insertion of the new item takes $O(\log w)$ time and deletion of the first item takes $O(w)$ time. To the best of our knowledge, this is the first work to support both LIS enumeration and LIS with constraints computation by using a single uniform data structure for real time sequential data streams. Our method outperforms the state-of-the-art methods in both time and space cost, not only theoretically, but also empirically.


\nop{
Using a time window $W$ of size $w$, the items of a sequential data stream within the time window induce a finite sequence. As the window slides, the finite sequence gets updated by deleting the first item and appending a new item to the end. This online feature of the forever changing sequences imposes additional challenges to design various query algorithms for the sequences. In this paper, we propose a \emph{unified index}, an orthogonal list-based index, to support real time queries of all LIS and LIS with constraints over sequential data streams. The index built by our algorithm requires $O(w)$ space, where $w$ is the time window size. The running time for building the initial index takes $O(wlogw)$ time.
Applying the index, deletion of the first item takes $O(w)$ time and insertion of the new item takes $O(logw)$ time. So a complete update to the index when time window slides only takes linear time, which guarantees that the index is scalable to high speed sequential data streams.

To the best of our knowledge, this is the first work to support both LIS enumeration and constrained LIS computation by using a single unified index for real time sequential data streams. Our method outperforms the state-of-the-art methods in both time and space cost, not only theoretically, but also empirically.
}

\end{abstract}


\vspace{-0.1in}
\section{Introduction}\label{sec:introduction}

Sequential data is a time series consisting of a sequence of data points, which are obtained by successive measurements made over a period of time. Lots of technical issues have been studied over sequential data, such as (approximate) pattern-matching query \cite{DBLP:conf/sigmod/FaloutsosRM94,DBLP:conf/icde/LianCY08}, clustering \cite{DBLP:journals/pr/Liao05}. Among these, computing the Longest Increasing Subsequence (LIS) over sequential data is a classical problem. Given a sequence $\alpha$, the LIS problem is to find a longest subsequence of a given sequence where the elements in the subsequence are in the increasing order. LIS is formally defined as follows. 

\begin{definition}\textbf{(Longest Increasing Subsequence).}\label{def:lis}  Let $\alpha=\{a_1$, $a_2$, $\cdots$, $a_n\}$ be a sequence, an increasing\footnote{Increasing subsequence in this paper is not required to be strictly monotone increasing and all items in $\alpha$ can also be arbitrary numerical value.}
subsequence $s$ of $\alpha$ is a subsequence of $\alpha$ whose elements are sorted in order from the smallest to the biggest. An increasing subsequence $s$ of $\alpha$ is called a \underline{L}ongest \underline{I}ncreasing \underline{S}ubsequence (LIS) if there is no other increasing subsequence $s^{\prime}$ with $|s|<|s^\prime|$. A sequence $\alpha$ may contain multiple LIS, all of which have the same length. We denote the set of LIS of $\alpha$ by $LIS(\alpha)$.
\end{definition}

Besides the static model (i.e., computing LIS over a given sequence $\alpha$), recently, computing LIS has been considered in the streaming model \cite{liswAlbert2004,lissetChen2007}. Formally, given an infinite time-evolving sequence $\alpha_{\infty}$ = $\{a_1,...,a_{\infty}\}$ ($a_i \in \mathbb{R}$), we continuously compute LIS over the subsequence induced by the time window $\{a_{i-(w-1)}$, $a_{i-{(w-2)}}$,..., $a_{i}\}$. The size of the time window is the number of the items it spans in the data stream. Consider the sequence $\alpha=\{a_1=3,a_2=9,a_3=6,a_4=2,a_5=8,a_6=5, a_7=7\}$ under window $W$ in Figure \ref{fig:timewindow}. There are four LIS in $\alpha$: $\{3,6,7\}, \{3,6,8\}$, \{$2$,$5$,$7$\} and \{$3$, $5$, $7$\}. Besides LIS enumeration, we introduce two important features of LIS, i.e., \emph{gap} and \emph{weight} and compute LIS with various constraints, where ``gap'' measures the value difference between the tail and the head item of LIS and ``weight'' measures the sum of all items in LIS (formally defined in Definitions \ref{def:weight}-\ref{def:height}). Figure \ref{fig:timewindow} shows LIS with various specified constraints. In the following, we demonstrate the usefulness of LIS in different applications. 

%
%
%
%
%


\vspace{-0.15in}
\begin{figure}[h!]
\centering
\includegraphics[width=0.95\textwidth]{\picfolder 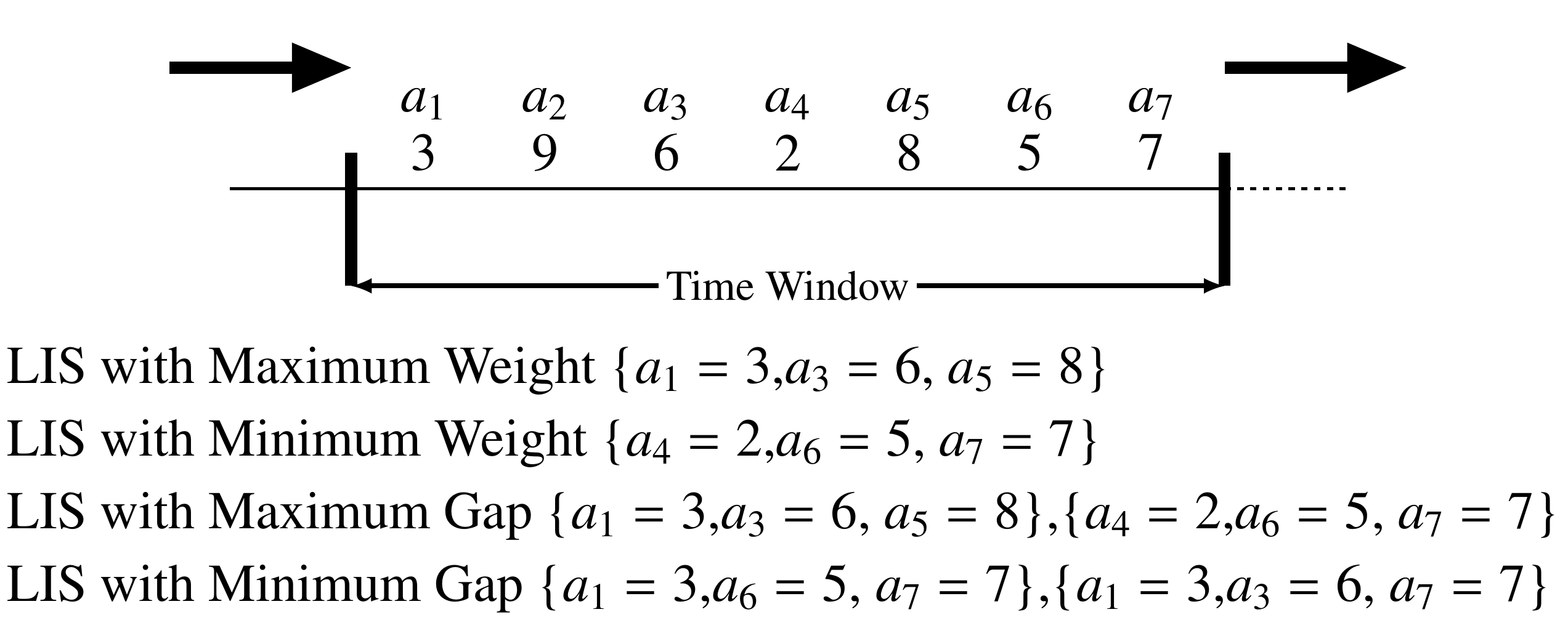}
\caption{Computing LIS with constraints in data stream model}
\vspace{-0.10in}
\label{fig:timewindow}
\end{figure}
\vspace{-0.05in}

\textbf{{Example 1: Realtime Stock Price Trend Detection.}} 
LIS is a classical measure for sequence sortedness and trend analysis \cite{SortedGopalan2007}. As we know, a company's stock price forms a time-evolving sequence and the real-time measuring the stock trend is of great significance to the stock analysis.  Given a sequence $\alpha$ of the stock prices within a period, an LIS of $\alpha$ measures an uptrend of the prices. We can see that price sequence with a long LIS always shows obvious upward tendency for the stock price even if there are some price fluctuations. Note that we do not require that the price increasing is contiguous without break, since stock price fluctuation within a couple of days does not impact the overall long term tendency within this period. 

\begin{figure}[!h]
\centering
\includegraphics[width=0.70\linewidth]{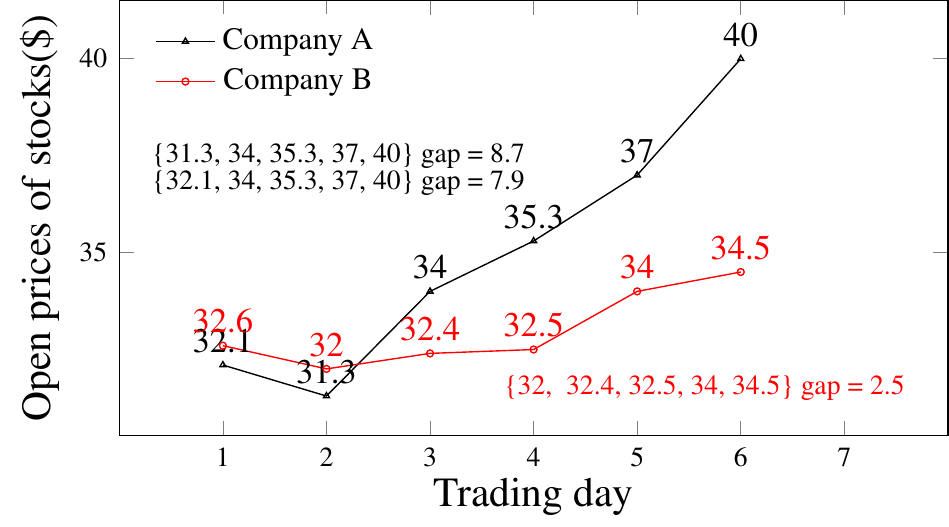}
\caption{LIS with different gaps of stock price sequence}
\vspace{-0.1in}
\label{fig:stockapp}
\end{figure}
\vspace{-0.1in}

Although the LIS length can be used to measure the uptrend stability, LIS with different gaps indicate different growth intensity. For example, Figure \ref{fig:stockapp} presents the stock prices sequences of two company: $A$ and $B$. Although both sequences of $A$ and $B$ have the same LIS length ($5$), growth intensity of $A$'s stock obvious dominates that of $B$, which is easily observed from the different gaps in LIS in $A$ and $B$. Therefore, besides LIS length, \emph{gap} is another feature of LIS that weights the growth intensity. We consider that the computation of LIS with extreme gap that is more likely chosen as measurement of growth intensity than a random LIS. Furthermore, this paper also considers other constraints for LIS, such as \emph{weight} (see Definition \ref{def:weight}) and study how to compute LIS with constraints directly rather than using  post-processing technique.

\nop{
Although VARIANT \cite{variant2009} and MHLIS \cite{minheight2009} propose LIS with constraints, they only work for the static model. Obviously, in real-time stock price analysis, we should consider the data stream scenario, such as sliding window model. Neither VARIANT nor MHLIS can be used in this context except for re-computing LIS with constraints from scratch in each window, which costs

$O(w\log w)$ time where $w$ is the time window size. However, our method only needs $O(w)$ time in each time window. Experiments also confirm that our method outperforms VARIANT and MHLIS significantly. 
} 

\nop{
\vspace{0.1in}
Second, none of existing methods except \cite{lissetChen2007} addresses LIS enumeration problem, namely, finding all LIS of an arbitrary sequence, no matter in the static  model or the stream model. Existing solutions except \cite{lissetChen2007} aim to compute the length of LIS and output a single LIS rather than enumerating all LIS. However, many applications require enumerating all LIS. The following biological sequence query example illustrates LIS enumeration. We first discuss it in the static model and then extend it to the data stream scenario.  
} 

\textbf{{Example 2:  Biological Sequence Query.}} LIS is also used in biological sequence matching \cite{liswAlbert2004,blastalignment}. For example, Zhang \cite{blastalignment} designed a two-step algorithm (BLAST+LIS) to locate a transcript or protein sequence in the human genome map. The BLAST (Basic Local Alignment Search Tool) \cite{blast:algorithm} algorithm is to identify high-scoring segment pairs (HSPs) between query transcript sequence $Q$ and a long genomic sequence $L$. Figure \ref{fig:alignment} visualizes the outputs of BLAST. The segments with the same color (number) denote the HSPs. For example, segment 2 (the red one) has two matches in the genomic sequence $L$, denoted as $2_1$ and $2_2$. To obtain a global alignment, the matches of segments 1, 2, 3  in the genomic sequence $L$ should coincide with the segment order in query sequence $Q$, which constitutes exactly the LIS (in $L$) that are listed in Figure \ref{fig:alignment}. For example, LIS $\{1,2_{1},3_{1}\}$ represents a global alignment of $Q$ over sequence $L$. Actually, there are three different LIS in $L$ as shown in Figure \ref{fig:alignment}, which correspond to three different alignments between query transcript/protein $Q$ and genomic sequence $L$. Obviously, outputting only a single LIS may miss some important findings. Therefore, we should study LIS enumeration problem. 

\nop{
local similarity

enables a scientist to identify library sequences that resemble the query sequence above a certain threshold. However, the BLAST is a local similarity search program whose output often contains many redundant high-scoring segment pairs (HSPs) that do not have global alignment information. For example, given a query transcript sequence $Q$ over a long genomic sequence $L$, Figure \ref{fig:alignment} visualizes the outputs of BLAST. The segments with the same color (number) denote the HSPs. For example, segment 2 (the red one) has two matches in the genomic sequence $L$, denoted as $2_1$ and $2_2$. To obtain a global alignment, the matches of segments 1, 2, 3  in the genomic sequence $L$ should coincide with the segment order in query sequence $Q$, which constitutes exactly the LIS (in $L$) that are listed in Figure \ref{fig:alignment}. However, given a query sequence $Q$, there are three different LIS in $L$ as shown in Figure \ref{fig:alignment}, which correspond to three different alignments between query transcript/protein $Q$ and genomic sequence $L$. Obviously, outputting only a single LIS may miss some important findings. Therefore, we should study LIS enumeration problem. }

We extend the above LIS enumeration application into the sliding window model \cite{streamwindow2008Towards}. In practice, the range of the whole alignment result of $Q$ over $L$ should not be too long. Thus, we can introduce a threshold length $|w|$ to discover all LIS that span no more than $|w|$ items, i.e, all LIS in each time window with size $|w|$. This is analogous to our problem definition in this paper. 

\begin{figure}[!h]
\centering
\includegraphics[width=0.80\linewidth]{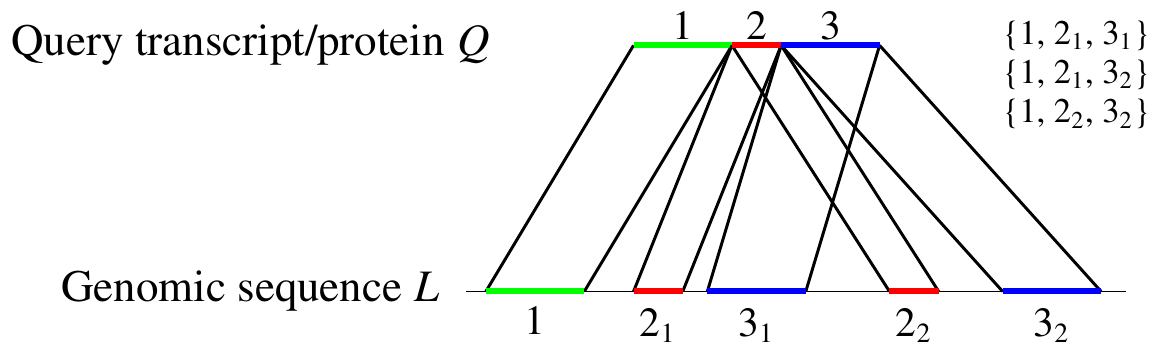}
\caption{Biological Sequence Alignment }
\vspace{-0.05in}
\label{fig:alignment}
\end{figure}

\nop{
LISSET \cite{lissetChen2007} is the only work computing LIS enumeration over sequential data streams. LISSET requires $O(w^2)$ space while our proposed solution only uses $O(w)$ space,
where $w$ is the size of the input sequence in each time window. Besides, in the context of sliding window stream model, LISSET takes $O(w)$ time for the insertion of an item,  while our method costs $O(\log w)$ time. 
} 

Although LIS has received considerable attention from the theoretical computer science community \cite{lissetChen2007, variant2009,kim1990finding, minheight2009}, none of the existing approaches support both LIS enumeration and constrained LIS enumeration simultaneously. For example, the method presented in \cite{lissetChen2007} supports LIS enumeration, but fails to compute constrained LIS. In \cite{variant2009} and \cite{minheight2009}, the method can be used to compute constrained LIS, but not to enumerate all LIS. More importantly, many works are based on \emph{static sequences} rather than data streams. Techniques developed in these works cannot handle updates which are essential in the context of data streams. To the best of our knowledge, there are only three research articles that addressed the problem of computing LIS over data stream model \cite{liswAlbert2004,lissetChen2007, deorowicz2013cover}. None of them computes constrained LIS. Literature review and the comparative studies of our method against other related work are given in Section \ref{sec:relatedwork} and Section \ref{sec:compare}, respectively. 
\nop{Furthermore, although the result in \cite{lissetChen2007} can be extended to enumerate all LIS in each time window, it requires $O(w^2)$ space to enumerate LIS. Our approach only requires $O(w)$ space, where $w$ is the time window size.}

\vspace{0.05in}
\subsection{Our Contributions}
Observed from the above examples, we propose a novel solution in this paper that studies both LIS enumeration and computing LIS with constraints with \emph{a uniform method} \emph{under the data stream model}. We propose a novel data structure to efficiently support both LIS enumeration and LIS with constraints. Furthermore, we design an efficient update algorithm for the maintenance of our data structure so that our approach can be applied to the data stream model. Theoretical analysis of our algorithm proves that our method outperforms the state-of-the-arts work (see Section \ref{sec:theory} for details). We prove that the space complexity of our data structure is $O(w)$, while the algorithm proposed in  \cite{lissetChen2007} needs a space of size $O(w^2)$. Time complexities of our data structure construction and update algorithms are also better than \cite{lissetChen2007}. For example, \cite{lissetChen2007} needs $O(w^2)$ time for the data structure construction, while our method needs $O(w \log w)$ time. Besides, we prove that both our LIS enumeration and  LIS with constraints query algorithms are \emph{optimal output-sensitive} algorithms\footnote{\small{The algorithm time complexity is linear to the corresponding output size.}}. Comprehensive comparative study of our results against previous results is given in Section \ref{sec:compare}. We use real and synthetic datasets to experimentally evaluate our approach against the state-of-the-arts work. Experimental results also confirm that our algorithms outperform existing algorithms. Experimental codes and datasets are available at Github \cite{lisgit}.


We summarize our major contributions in the following:
\vspace{-0.1in}

\begin{enumerate}
\setlength{\itemsep}{0.1em}
\setlength{\parskip}{0pt}
  
\item We are the first to consider the computation of both LIS with constraints and LIS enumeration in the data stream model.
\item We introduce a novel data structure to handle both LIS enumeration and computation of LIS with constraints uniformly.
\item Our data structure is scalable under data stream model because of the linear update algorithm and linear space cost.
\item Extensive experiments confirm the superiority of our method.  
\end{enumerate}
\vspace{-0.05in}

\nop{
Extensive experiments over both large real and synthetic datasets confirm that our method outperforms the state-of-the-arts significantly. 
}

\vspace{-0.1in}
\section{Related Work}\label{sec:relatedwork}

LIS-related problems have received considerable attention in the literature. We give a briefly review of the related work from the perspectives of the \emph{solution} and \emph{problem definition}, respectively.


\vspace{-0.05in}
\subsection{Solution Perspective} \label{sec:solutionRW}
Generally, existing LIS computation approaches can be divided into following three categories: 

\emph{1. Dynamic Programming-based}. Dynamic programming is a classical method to compute the length of LIS. Given a sequence $\alpha$, assuming that $\alpha_i$ denotes the prefix sequence consisting of the first $i$ items of $\alpha$, then the dynamic programming-based method is to compute the LIS of $\alpha_{i+1}$ after computing the LIS of $\alpha_{i}$. However, dynamic programming-based method costs $O(w^2)$ time where $n$ denotes the length of the sequence $\alpha$. Dynamic programming-based method can be easily extended to enumerate all LIS in a sequence which costs $O(w^2)$ space.

\emph{2. Young's tableau-based}. \cite{schensted1961longest} proposes a Young's tableau-based solution to compute LIS in $O(w \log w)$ time. The width of the first row of Young's tableau built over a sequence $\alpha$ is exactly the length of LIS in $\alpha$.  Albert et al.\cite{liswAlbert2004} followed the Young's tableau-based work to compute the LIS length in sliding window. They maintained the first row of Young's tableau, called principle row, when window slides. For a sequence $\alpha$ in a window, there are $n= |\alpha|$ suffix subsequences and the prime idea in \cite{liswAlbert2004} is to compress all principle rows of these suffix subsequence into an array, which can be updated in $O(w)$ time when update happens. Besides, they can output an LIS with a tree data structure which costs $O(w^2)$ space.

\emph{3. Partition-based}. There are also some work computing LIS by partitioning items in the sequence \cite{lissetChen2007, variant2009,deorowicz2013cover, minheight2009}. They classify items into $l$ partitions: $P_1$,$P_2$...,$P_l$, where $l$ is the length of LIS of the sequence. For each item $a$ in $P_k$ ($k=1,...,l$), the maximum length of the increasing subsequence ending with $a$ is exactly $k$. Thus, when partition is built, we can start from items in $P_l$ and then scan items in $P_{l-k}$ ($1\leq k < l$) to construct an LIS. The partition is called different names in different approaches, such as \emph{greedy-cover} in \cite{variant2009, deorowicz2013cover}, \emph{antichain} in \cite{lissetChen2007}. Note that \cite{variant2009} and \cite{minheight2009} conduct the partition over a static sequence to efficiently compute LIS with constraints. \cite{deorowicz2013cover} use partition-based method as subprogram to find out the largest LIS length among $n-w$ windows where $w$ is the size of the sliding window over a sequence $\alpha$ of size $n$. Their core idea is to avoid constructing partition on the windows whose LIS length is less than those previously found. In fact, they re-compute the greedy-cover in each of the windows that are not filtered from scratch. None of the partition-based solutions address the data structure maintenance issues expect for \cite{lissetChen2007}. \cite{lissetChen2007} is the only one to study the LIS enumeration in streaming model. Both of their insertion and deletion algorithms cost $O(w)$ time \cite{lissetChen2007}. Besides, to support update, they assign each item with $O(w)$ pointers and thus their method costs $O(w^2)$ space. 

Actually, our approach belongs to the partition-based solution, where each horizontal list(see Definition \ref{def:hlist}) is essentially a \emph{partition}. However, because of introducing up/down neighbors in QN-list (see Definition \ref{def:neighbors} and \ref{def:orthogonal}), our data structure costs only $O(w)$ space. Besides, the insertion and deletion time of our method is $O(\log w)$ and $O(w)$, respectively, which makes it suitable in the streaming context. 
Furthermore, our data structure supports both LIS enumeration and LIS with various constraints.



\vspace{-0.1in}
\subsection{Problem Perspective} \label{sec:problemRW}
We briefly position our problem in existing work on LIS computation in \emph{computing task} and \emph{computing model}. Note that LIS can also be used to compute LCS (longest common subsequence) between two sequences \cite{HS77}, but that is not our focus in this paper. First, there are three categories of LIS computing tasks. The first is to compute the length of LIS and output a single LIS (not enumerate all) in sequence $\alpha$ \cite{liswAlbert2004,deorowicz2013cover,lengthFredman1975,lics-Liben-Nowell2006,schensted1961longest}. The second is LIS enumeration, which finds all LIS in a sequence $\alpha$ \cite{EnumLISBespamyatnikh2000,lissetChen2007}. \cite{EnumLISBespamyatnikh2000} computes LIS enumeration only on the sequence that is required to be a permutation of \{$1$,$2$,...,$n$\} rather than a general sequence (such as \{$3$, $9$, $6$, $2$, $8$, $5$, $7$\} in the running example). The last computing task studies LIS with constraints, such as gap and weight \cite{variant2009,minheight2009}. On the other hand, there are two computing models for LIS. One is the static model assuming that the sequence $\alpha$ is given without changes. For example, \cite{variant2009,schensted1961longest,robinson1938,minheight2009} are based on the static model. These methods cannot be applied to the streaming context directly except re-computing LIS from scratch in each time window.  The other model is the data stream model, which has been considered in some recent work\cite{liswAlbert2004,lissetChen2007}.
 
Table \ref{tab:fillgaprelatedwork} illustrates the existing works from two perspectives: \emph{computing task} and \emph{computing model}. There are two observations from the table. First, there is no existing uniform solution for all LIS-related problems, such as LIS length, LIS enumeration and LIS with constraints. Note that any algorithm for computing LIS enumeration and LIS with constraints can be applied to computing LIS length directly. Thus, we only consider LIS enumeration and LIS with constraints in the later discussion. Second, no algorithm supports computing LIS with constraints in the streaming context.  Therefore, the major contribution of our work lies in that we propose a uniform solution (the same data structure and computing framework) for all LIS-related issues in the streaming context. Table \ref{tab:fillgaprelatedwork} properly positions our method with regard to existing works. 

None of the existing work can be easily extended to support all LIS-related problems in the data steam model except for LISSET \cite{lissetChen2007}, which is originally proposed to address LIS enumeration in the sliding window model. Also, LISSET can compute LIS with constraints using post-process technique (denoted as LISSET-post in Figure \ref{fig:exp:stock}). So, we compare our method with LISSET not only theoretically, but also empirically in Section \ref{sec:compare}. LISSET requires $O(w^2)$ space while our method only uses $O(w)$ space, where $w$ is the size of the input sequence. Experiments show that our method outperforms LISSET significantly, especially computing LIS with constraints (see Figures \ref{fig:maxwtime}-\ref{fig:minhtime}). 

\vspace{-0.05in}
 \begin{table}[!h]
 \centering
 \small
     \caption{Our Method VS. Existing Works on Computing LIS(s)} 
     \label{tab:fillgaprelatedwork}
 
     \begin{small}
     \resizebox{1.0\textwidth}{!}
     {
   	\begin{tabular}{|l|c|c|}
     \hline
     \bfseries{ Computing Task}  			   &
     \bfseries{ Static only}     &
     \bfseries{ Stream} 
     \\ \hline
     LIS length(outputting a single LIS)  & \cite{EnumLISBespamyatnikh2000}\cite{variant2009}\cite{schensted1961longest}\cite{robinson1938}\cite{minheight2009}
     
     & \cite{liswAlbert2004}\cite{lissetChen2007}\cite{deorowicz2013cover}, Our Method
     \\ \hline
     LIS Enumeration & \cite{EnumLISBespamyatnikh2000}\footnotemark[3] & \cite{lissetChen2007}, Our Method
     \\ \hline
     LIS with constraints & \cite{variant2009}\cite{minheight2009} & Our Method  
     \\ \hline
 
   	\end{tabular}
     }
     \end{small}
     \vspace{-0.1in}
 \end{table}
\vspace{-0.08in}
 
 \footnotetext[3]{\cite{EnumLISBespamyatnikh2000} computes LIS enumeration only on the sequence that is required to be a permutation of \{$1$,$2$,...,$n$\}.}
 
 \setcounter{footnote}{3}

\vspace{-0.1in}
\section{Problem Formulation}\label{sec:problemdef}


Given a sequence $\alpha=\{a_1$, $a_2$, $\cdots$, $a_n\}$, the set of  increasing subsequences of $\alpha$ is denoted as $IS(\alpha)$. For a sequence $s$, the head and tail item of $s$ is denoted as $s^h$ and $s^t$, respectively. We use $|s|$ to denote the length of $s$.

Consider an infinite time-evolving sequence $\alpha_{\infty}$ = $\{a_1,...,a_{\infty}\}$ ($a_i \in \mathbb{R}$). In the sequence $\alpha_{\infty}$, each $a_i$ has a unique position $i$ and $a_i$ occurs at a corresponding time point $t_i$, where $t_i < t_j$ when $0<i<j$. We exploit the \emph{tuple-basis} sliding window model \cite{streamwindow2008Towards} in this work. There is an internal \emph{position} to tuples based on their arrival order to the system, ensuring that an input tuple is processed as far as possible before another input tuple with a higher \emph{position}. A sliding window $W$ contains a consecutive block of items in $\{a_1,\cdots,a_{\infty}\}$, and $W$ slides a single unit of position per move towards $a_{\infty}$ continually. We denote the size of the window $W$ by $w$, which is the number of items within the window. During the time $[t_i, t_{i+1})$, items of $\alpha$ within the sliding time window $W$ induce the sequence $\{a_{i-(w-1)}$,$a_{i-{(w-2)}}$,...,$a_{i}\}$, which will be denoted by $\alpha(W,i)$. Note that, in the sliding window model, as the time window continually shifts towards $a_{\infty}$, at a pace of one unit per move, the sequence formed and the corresponding set of all its LIS will also change accordingly. In the remainder of the paper, all LIS-related problems considered are in the data stream model with sliding windows.

\vspace{-0.05in}
\begin{definition}\textbf{(LIS-enumeration)}. Given a time-evolving sequence $\alpha_{\infty}=\{ a_1,...,a_{\infty}\}$ and a sliding time window $W$ of size $w$, \emph{LIS-enumeration} is to report $LIS(\alpha(W,i))$ (i.e., all LIS within the sliding time $W$) continually as the window $W$ slides. All LIS in the same time window have the same length. 

\end{definition}
\vspace{-0.1in}

As mentioned in Introduction, some applications are interested in computing LIS with constraints instead of simply enumerating all of them. Hence, we study the following constraints over the LIS's \emph{weight} (Definition \ref{def:weight}) and \emph{gap} (Definition \ref{def:height}), after which we define several problems computing LIS with various constraints (Definition \ref{def:computeconstraints}) \footnote{So far, eight kinds of constraints for LIS were proposed in the literature \cite{variant2009,minheight2009,fastyang2008}. Due to the space limit, we only study four of them (i.e., max/min weight/gap) in this paper. However, our method can also easily support the other four constraints, which are provided in Appendix \citeAPPconstraints 
\submitshow{of the full version of this paper \cite{li2016lis}}
.
} 
. 

\vspace{-0.05in}
\begin{definition}\textbf{(Weight)}.\label{def:weight} Let $\alpha$ be a sequence, $s$ be an LIS in $LIS(\alpha)$. The \emph{weight} of $s$ is defined as $\sum\nolimits_{a_i  \in s} {a_i }$, i.e., the sum of all the items in $s$, we denote it by $weight(s)$.
\end{definition}

\begin{definition}\textbf{(Gap)}.\label{def:height} Let $\alpha$ be a sequence, $s$ be an LIS in $LIS(\alpha)$. The \emph{gap} of $s$ is defined as $gap(s)=s^{t}-s^{h}$, i.e., the difference between the tail $s^t$ and the head $s^h$ of $s$.
\end{definition}

\begin{figure*}[t] \RawFloats

\setlength{\belowcaptionskip}{-2pt}
\begin{minipage}[t]{0.4\linewidth}
\centering
\subcaptionbox{Horizontal lists\label{fig:horizontal}}
{
	\resizebox{0.45\linewidth}{!}
	{
		\includegraphics[width = 1in]{\picfolder 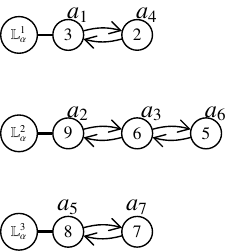}
	}
}
\subcaptionbox{QN-List $\mathbb{L}_{\alpha}$ \label{fig:orthogonal}}
{
	\resizebox{0.45\linewidth}{!}
	{
		\includegraphics{\picfolder 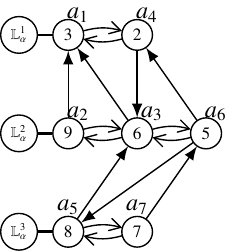}
	}
}
\caption
{
 Horizontal lists and QN-List of  running example $\alpha$
}
\end{minipage}%
\hspace{0.1in}
\begin{minipage}[t]{0.35\linewidth}
\centering
\resizebox{\linewidth}{!}
{
	\includegraphics{\picfolder 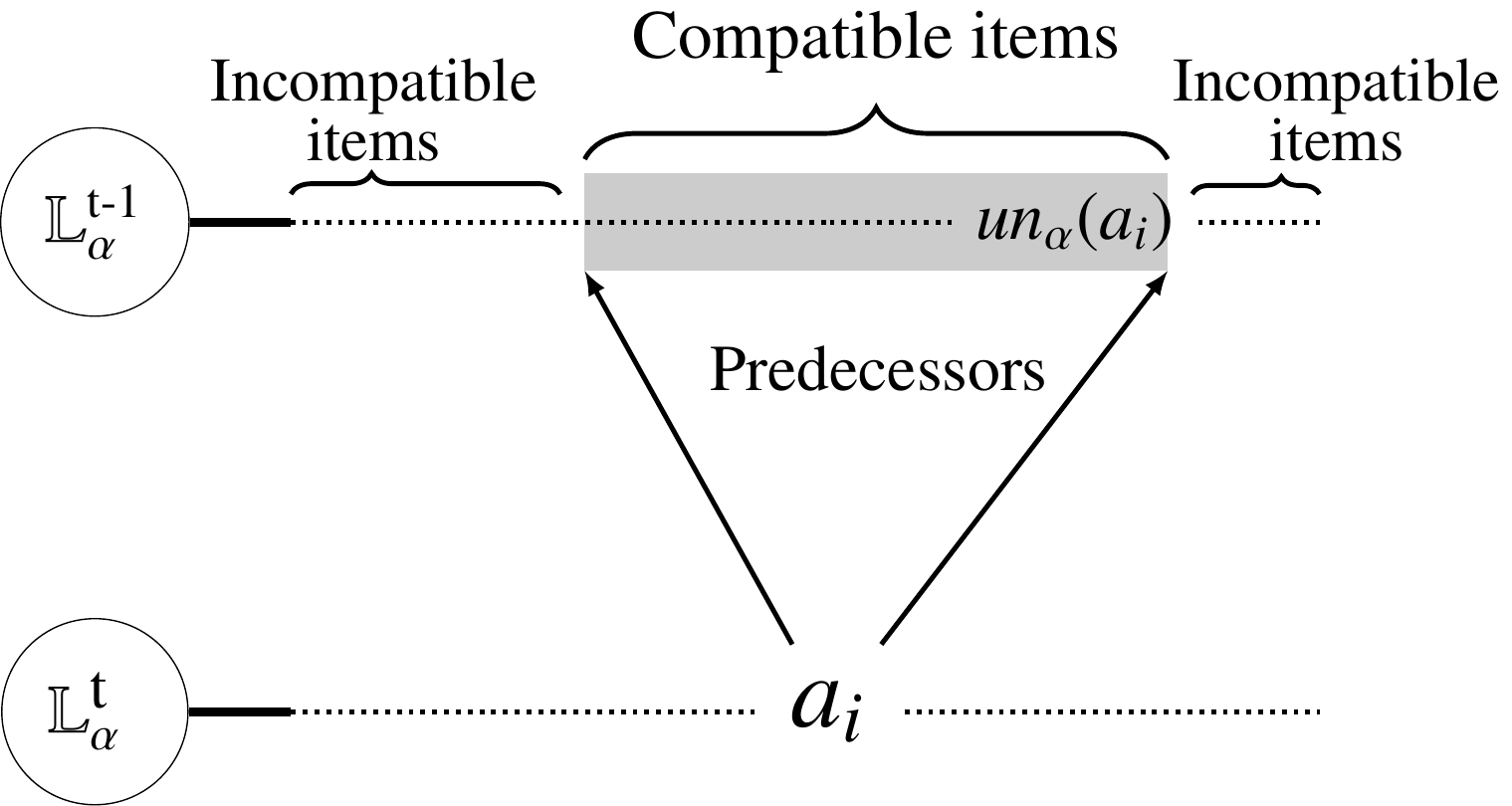}
}
\caption{Sketch of predecessors of $a_i$}
\label{fig:consecutive}
\end{minipage}%
\hspace{0.1in}
\begin{minipage}[t]{0.20\linewidth}
\centering
\resizebox{\linewidth}{!}
{
	\includegraphics{\picfolder 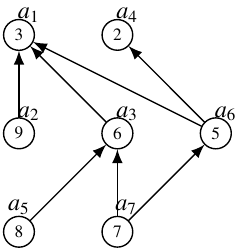}
}
\caption{DAG of running example $\alpha $}
\label{fig:lisdag}
\end{minipage}%

\end{figure*}

\begin{definition}\textbf{(Computing LIS with Constraint)}. \label{def:computeconstraints}
Given a time-evolving sequence $\alpha_{\infty}=\{ a_1,...,a_{\infty}\}$ and a sliding window $W$, each of the following problems is to report all the LIS subject to its own specified constraint within a time window continually as the window slides. For $s \in LIS(\alpha(W,t_i))$:

$s$ is an \textbf{LIS with Maximum Weight} if 
\[ \forall s^{\prime} \in LIS(\alpha(W,t_i)), weight(s) \geq weight(s^{\prime}) \]

\vspace{-0.05in}
$s$ is an \textbf{LIS with Minimum Weight} if
\[ \forall s^{\prime} \in LIS(\alpha(W,t_i)), weight(s) \leq weight(s^{\prime}) \]

\vspace{-0.05in}
$s$ is an \textbf{LIS with Maximum Gap} if
\[ \forall s^{\prime} \in LIS(\alpha(W,t_i)), gap(s) \geq gap(s^{\prime}) \]

\vspace{-0.05in}
$s$ is an \textbf{LIS with Minimum Gap} if
\[ \forall s^{\prime} \in LIS(\alpha(W,t_i)), gap(s) \leq gap(s^{\prime}) \]
\end{definition}
\vspace{-0.1in}

A running example that is used throughout the paper is given in Figure \ref{fig:timewindow}, which shows a time-evolving sequence $\alpha_{\infty}$ and its first time window $W$. 

\nop{The induced sequence within the time window is $\alpha =\{ a_1=3, a_2=9, a_3=6, a_4=2, a_5=8, a_6=5, a_7=7 \}$. There are four LIS in $\alpha$: $\{3, 6, 7\}$, $\{3, 6, 8\}$, $\{2, 5, 7\}$ and $\{3, 5, 7\}$. The LIS with various specified constraints are also presented in Figure \ref{fig:timewindow}.
}

\vspace{-0.1in}
\section{Quadruple Neighbor List \large $\mathbb{L}_{\alpha}$}\label{sec:construction}
In this section, we propose a data structure, a \emph{q}uadruple \emph{n}eighbor list (QN-list for short), denoted as $\mathbb{L}_{\alpha}$, for a sequence $\alpha=\{a_1$,$a_2$,...,$a_w\}$, which is induced from $\alpha_{\infty}$ by a time window $W$ of size $w$. Some important properties and the construction of $\mathbb{L}_{\alpha}$ are discussed in Section \ref{sec:indexprop} and Section \ref{sec:indexconstruct}, respectively. In Section \ref{sec:lisenumnew}, we present an efficient algorithm over $\mathbb{L}_{\alpha}$ to enumerate all LIS in $\alpha$. In the following two sections, we will discuss how to update the QN-List efficiently in data stream scenario (Section \ref{sec:maintenance}) and compute LIS with constraints (Section \ref{sec:computation}).

\vspace{-0.04in}
\subsection{{\large $\mathbb{L}_{\alpha}$}---Background and Definition}
\label{sec:indexdef}




For the easy of the presentation, we introduce some concepts of LIS before we formally define the quadruple neighbor list (QN-List, for short). Note that two concepts (\emph{rising length} and \emph{horizontal list}) are analogous to the counterpart in the existing work. We explicitly state the connection between them as follows. 

\vspace{-0.06in}
\begin{definition}(\textbf{Compatible pair})
\label{def:compatiblepair}
Let $\alpha=$ $\{a_1$ $,a_2,...,a_w\}$ be a sequence.
  $a_i$ is \emph{compatible} with $a_j$ if $i < j$ and $a_i \leq a_j$ in $\alpha$. We denote it by $a_i$ $\seqpar{\alpha} a_j$.
\end{definition}

\vspace{-0.1in}
\begin{definition}\textbf{(Rising Length)} \cite{lissetChen2007} \footnote{\emph{Rising length} in this paper is the same as \emph{height} defined in \cite{lissetChen2007}. We don't use \emph{height} here to avoid confusion because \emph{height} is also defined as the difference between the head item and tail item of an LIS in \cite{minheight2009}.}\label{def:risinglength}
Given a sequence $\alpha=$ $\{a_1$ $,a_2,...,a_w\}$ and $a_i$ $\in$ $\alpha$, we use $IS_{\alpha}(a_i)$ to denote the set of all \emph{increasing subsequences} of $\alpha$ that ends with $a_i$.

The \emph{rising length} $RL_{\alpha}(a_i)$ of $a_i$ is defined as the maximum length of subsequences in $IS_{\alpha}(a_i)$, namely,
\[
RL_{\alpha}(a_i) =
\max{
        \left\{\;
        \left| s \right| \mid s \in IS_{\alpha}(a_i)
        \right\}
    }
\]
\end{definition}

For example, consider the sequence $\alpha =$ $\{a_1=3,a_2=9,a_3=6,a_4=2,a_5=8,a_6=5,a_7=7\}$ in Figure \ref{fig:timewindow}. Consider $a_5=8$. There are four increasing subsequences\{$a_1=3$, $a_5=8$\},  \{$a_3=6$, $a_5=8$\}, \{$a_4=2$, $a_5=8$\},  \{$a_1=3$, $a_3=6$, $a_5=8$\} that end with $a_5$\footnote{Strictly speaking, \{$a_5$\} is also an increasing subsequence with length 1.}. The maximum length of these increasing subsequences is $3$. Hence, $RL_{\alpha}(a_5)=3$.

\begin{definition}(\textbf{Predecessor}). \label{def:pred}
Given a sequence $\alpha$ and $a_i$ $\in \alpha$, for some item $a_j$, $a_j$ is a \emph{predecessor} of $a_i$ if 
\[a_j \seqpar{\alpha} a_i \; AND \; RL_{\alpha}(a_j) = RL_{\alpha}(a_i) - 1\]
and the set of predecessors of $a_i$ is denoted as $Pred_{\alpha}(a_i)$.
\end{definition}
\vspace{-0.05in}

In the running example in Figure \ref{fig:timewindow}, $a_3$ is a predecessor of $a_5$ since $a_3 \seqpar{\alpha} a_5$ and $RL_{\alpha}(a_3)(=2)=RL_{\alpha}(a_5)(=3)-1$. Analogously, $a_1$ is also a predecessor of $a_3$.

With the above definitions, we introduce four neighbours for each item $a_i$ as follows:


\begin{definition}(\textbf{Neighbors of an item}). \label{def:neighbors} Given a sequence $\alpha$ and $a_i$ $\in \alpha$, $a_i$ has up to four neighbors.
\begin{enumerate}
\setlength\itemsep{0.0em}
\item 
	\textbf{left neighbor} $ln_{\alpha}(a_i)$: $ln_{\alpha}(a_i) = a_j$ if $a_j$ is the \emph{nearest} item \emph{before} $a_i$ such that $RL_{\alpha}(a_i)$ $=RL_{\alpha}(a_j)$.
\item 
	\textbf{right neighbor} $rn_{\alpha}(a_i)$: $rn_{\alpha}(a_i) = a_j$ if $a_j$ is the \emph{nearest} item \emph{after} $a_i$ such that $RL_{\alpha}(a_i)$ $=$ $RL_{\alpha}(a_j)$.
\item 
	\textbf{up neighbor} $un_{\alpha}(a_i)$: $un_{\alpha}(a_i)=a_j$ if $a_j$ is the \emph{nearest} item \emph{before} $a_i$ such that $RL_{\alpha}(a_j)=RL_{\alpha}(a_i)-1$.
\item 
	\textbf{down neighbor} $dn_{\alpha}(a_i)$: $dn_{\alpha}(a_i)=a_j$ if $a_j$ is the \emph{nearest} item \emph{before} $a_i$ such that $RL_{\alpha}(a_j)$ $=RL_{\alpha}(a_i)+1$. 
\end{enumerate}
\end{definition}
\vspace{-0.05in}

Apparently, if $a_i$ $= ln_{\alpha}(a_j)$ then $a_j$ $=rn_{\alpha}(a_i)$. Besides, we know that left neighbor(Also right neighbor) of item $a_i$ has the same rising length as $a_i$ and naturally, items linked according to their left and right neighbor relationship forms a \emph{horizontal list}, which is formally defined in Definition \ref{def:hlist}. The horizontal lists of $\alpha$ is presented in Figure \ref{fig:horizontal}.

\vspace{-0.1in}
\begin{definition}(\textbf{Horizontal list}). \label{def:hlist}
Given a sequence $\alpha$, consider the subsequence consisting of all items whose rising lengths are $k$: $s_k = $ \{$a_{i_1}$, $a_{i_2}$,...,$a_{i_k}$\}, $i_1$ $< i_2$,...,$< i_k$. We know that for $1 \leq k^{\prime}$ $< k$, $a_{i_{k^{\prime}}}=$  $ln_{\alpha}(a_{i_{k{^\prime}+1}})$ and $a_{i_{k^{\prime}+1}}=$ $rn_{\alpha}(a_{i_{k'}})$. We define the list formed by linking items in $s_k$ together with left and right neighbor relationships as a horizontal list, denoted as $\mathbb{L}_{\alpha}^{k}$.
\end{definition}
\vspace{-0.08in}

\nop{
Let us recall the sequence $\alpha$ in Figure \ref{fig:timewindow}. The horizontal lists of $\alpha$ is presented in Figure \ref{fig:horizontal}, where the curve arrows indicate the left and right neighbor relationship.}

Recall the \emph{partition-based solutions} mentioned in Section \ref{sec:relatedwork}. Each horizontal list is essentially a \emph{partition}, which is the same as a \emph{greedy-cover} in \cite{deorowicz2013cover} and  \emph{antichain} in \cite{lissetChen2007}. Based on the horizontal list, we define our data structure QN-list (Definition \ref{def:orthogonal}) as follows.

\nop{
Apparently, horizontal lists are essentially partitions over the sequence that group items according to the rising length of them. In fact, there are some previous work also partition the sequence in this way, such as the greedy-cover in \cite{deorowicz2013cover}.}


\begin{figure*}[t!]
\newcommand{\isubfiglis}[4]
{
    \begin{subfigure}[h]{#1}
    \centering
    \resizebox{\linewidth}{1.36in}
    {
        \includegraphics{#2}
    }
    \caption{#3}
    \label{#4}
    \end{subfigure}
}
\centering
\newcommand{\isubfigwidth}{0.116\textwidth}
\newcommand{\isubfigdist}{0.001\textwidth}
    \isubfiglis{\isubfigwidth}{\picfolder 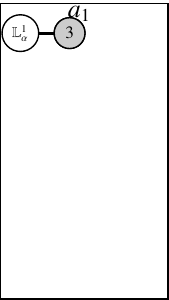}{\underline{3}}{fig:ipica}
    \isubfiglis{\isubfigwidth}{\picfolder 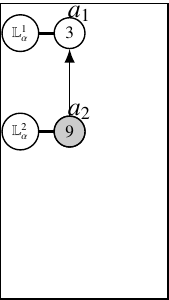}{3,\underline{9}}{fig:ipicb}
    \isubfiglis{\isubfigwidth}{\picfolder 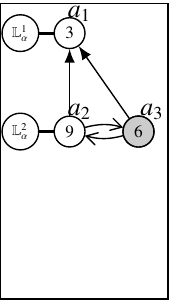}{3,9,\underline{6}}{fig:ipicc}
    \isubfiglis{\isubfigwidth}{\picfolder 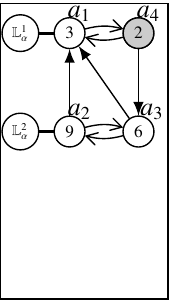}{3,9,6,\underline{2}}{fig:ipicd}
    \isubfiglis{\isubfigwidth}{\picfolder 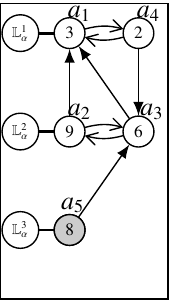}{3,9,6,2,\underline{8}}{fig:ipice}
    \isubfiglis{0.17\textwidth}{\picfolder 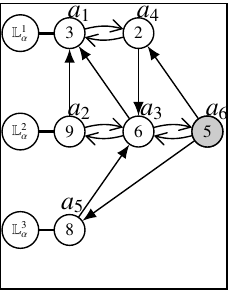}{3,9,6,2,8,\underline{5}}{fig:ipicf}
    \isubfiglis{0.17\textwidth}{\picfolder 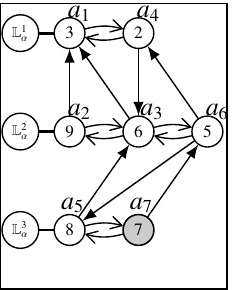}{3,9,6,2,8,5,\underline{7}}{fig:ipicg}
 \vspace{-0.1in}
\caption{Example of the quadruple neighbor list construction on sequence: \{$a_1=3$, $a_2=9$, $a_3=6$, $a_4=2$, $a_5=8$, $a_6=5$, $a_7=7$\}.}
\label{fig:ris}
\end{figure*} 

\begin{definition}(\textbf{Quadruple Neighbor List (QN-List)}).\label{def:orthogonal}
Given a sequence $\alpha=\{a_1,...,a_w\}$, the quadruple neighbor list over $\alpha$ (denoted as $\mathbb{L}_{\alpha}$) is a data structure containing all horizontal lists (See Definition \ref{def:hlist}) of $\alpha$ and each item $a_i$ in $\mathbb{L}_{\alpha}$ is also linked directly to its up neighbor and down neighbor. In essence, $\mathbb{L}_{\alpha}$ is constructed by linking all items in $\alpha$ with their four kinds of neighbor relationship. Specifically, $|\mathbb{L}_{\alpha}|$ denotes the number of horizontal lists in $\mathbb{L}_{\alpha}$.

\end{definition}
\vspace{-0.06in}

Figure \ref{fig:orthogonal} presents the QN-List $\mathbb{L}_{\alpha}$ of running example sequence $\alpha$ (in Figure \ref{fig:timewindow}) and the horizontal curve arrows indicate the left and right neighbor relationship while the vertical straight arrows indicate the up and down neighbor relationship.

\vspace{-0.06in}
\begin{theorem} \label{theorem:indexspace}
 Given a sequence $\alpha=\{a_1,...,a_w\}$, the data structure $\mathbb{L}_{\alpha}$ defined in Definition \ref{def:orthogonal} uses $O(w)$ space \footnote{Due to space limits, all proofs for theorems and lemmas are given in Appendix \citeAPPproof
\submitshow{
of the full version of this paper \cite{li2016lis}
}.
}. 
\end{theorem}

\subsection{{\large $\mathbb{L}_{\alpha}$}---Properties} \label{sec:indexprop}
Next, we discuss some properties of the QN-List $\mathbb{L}_{\alpha}$. These properties will be used in the maintenance algorithm in Section \ref{sec:maintenance} and various $\mathbb{L}_{\alpha}$-based algorithms in Section \ref{sec:computation}.

\vspace{-0.05in}

\begin{lemma} \label{lem:consecutive}
 Let $\alpha=$ $\{a_1$ $,a_2,...,a_w\}$ be a sequence. Consider two items $a_i$ and $a_j$ in a horizontal list $\mathbb{L}_{\alpha}^t$ (see Definition \ref{def:hlist}).  
\vspace{-0.05in}
\begin{enumerate}
\setlength\itemsep{0.0em}
\item 
	If $t = 1$, $a_i$ has no predecessor. If $t > 1$ then $a_i$ has at least one predecessor and all predecessors of $a_i$ are located in $\mathbb{L}_{\alpha}^{t-1}$.
\item \label{item:decreasing}
	If $rn_{\alpha}(a_j) = a_i$, then $i > j$ and $a_i < a_j$. If $ln_{\alpha}(a_j) = a_i$, then $ i < j$ and $a_i > a_j$. Items in a horizontal list $\mathbb{L}_{\alpha}^t$ ($t=1,\cdots,m$) are monotonically decreasing while their subscripts (i.e., their original position in $\alpha$) are monotonically increasing from the left to the right. And no item is compatible with any other item in the same list.
\item \label{item:consecutive}
	$\forall a_i \in \alpha$, all predecessors of $a_i$ form a nonempty consecutive block in $\mathbb{L}_{\alpha}^{t-1}$ ($t > 1$).
\item
	$un_\alpha(a_i)$ is the rightmost predecessor of $a_i$ in $\mathbb{L}_{\alpha}^{t-1}$ ($t > 1$).
\end{enumerate}
\end{lemma}
\vspace{-0.05in}

Figure \ref{fig:consecutive} shows that all predecessors of $a_i$ $\in \mathbb{L}_{\alpha}^{t}$ form a consecutive block from $un_{\alpha}(a_i)$ to the left in $\mathbb{L}_{\alpha}^{t-1}$, i.e., Lemma \ref{lem:consecutive}(3).


\vspace{-0.08in}
\begin{lemma}\label{lem:property}
Given sequence $\alpha$ and its $\mathbb{L}_{\alpha}$, $\forall$ $a_i \in \mathbb{L}_{\alpha}^{t}$ ($1\leq t \leq m$).
\begin{enumerate}
\setlength\itemsep{0.0em}
	\item \label{item:rlen}
		$RL_{\alpha}(a_i)$ = $t$  if and only if $a_i$ $\in$ $\mathbb{L}_{\alpha}^{t}$. In addition, the length of LIS in $\alpha$ is exactly the number of horizontal lists in $\mathbb{L}_{\alpha}$.
	\item \label{item:unrightmost}
		$un_\alpha(a_i)$(if exists) is the rightmost item in $\mathbb{L}_{\alpha}^{t-1}$ which is before $a_i$ in sequence $\alpha$.
	\item \label{item:dnrightmost}
		$dn_\alpha(a_i)$(if exists) is the rightmost item in $\mathbb{L}_{\alpha}^{t+1}$ which is before $a_i$ in sequence $\alpha$. Besides, $dn_{\alpha}(a_i)$ $> a_i$.
\end{enumerate}	
\end{lemma}

\vspace{-0.15in}
\begin{lemma} \label{lem:tailsorted}
Given sequence $\alpha$ and its $\mathbb{L}_{\alpha}$, for $1\leq i , j \leq |\mathbb{L}_{\alpha}|$
\[Tail(\mathbb{L}^i_{\alpha}) \leq Tail(\mathbb{L}^j_{\alpha}) \leftrightarrow i \leq j\]
where $Tail(\mathbb{L}^i_{\alpha})$ denotes the last item in list $\mathbb{L}^i_{\alpha}$.
\end{lemma}

\subsection{{\large $\mathbb{L}_{\alpha}$}---Construction} \label{sec:indexconstruct}
The construction of $\mathbb{L}_{\alpha}$ over sequence $\alpha$ lies in the determination of the four neighbors of each item in $\alpha$. We discuss the construction of $\mathbb{L}_{\alpha}$ as follows. Figure \ref{fig:ris} visualizes the steps of constructing $\mathbb{L}_{\alpha}$ for a given sequence $\alpha$.

 \textbf{Building QN-List $\mathbb{L}_{\alpha}$.} 
 \vspace{-0.1in}
\begin{enumerate}
\setlength\itemsep{0.0em}
  \item Initially, four neighbours of each item $a_i$ are set NULL;

  \item At step 1, $\mathbb{L}_{\alpha}^{1}$ is created in $\mathbb{L}_{\alpha}$ and $a_1$ is added into $\mathbb{L}_{\alpha}^{1}$ \footnote{We also record the position $i$ of each item $a_i$ in $\mathbb{L}_{\alpha}$ besides the item value.};
\vspace{-0.05in}
  \item At step 2, if $a_2 < a_1$, it means $RL_{\alpha}(a_2)=RL_{\alpha}(a_1)=1$. Thus, we append $a_2$ to $\mathbb{L}_{\alpha}^{1}$. Since $a_2$ comes after $a_1$ in sequence $\alpha$,  we set $rn_{\alpha}(a_1) = a_2$ and $ln_{\alpha}(a_2)=a_1$ respectively. 
  
  If $a_2 \geq a_1$, we can find an increasing subsequence $\{a_1,a_2\}$, i.e, $RL_{\alpha}(a_2)=2$. Thus, we create the second horizontal list $\mathbb{L}_{\alpha}^{2}$ and add $a_2$ to $\mathbb{L}_{\alpha}^{2}$. Furthermore, it is straightforward to know $a_1$ is the nearest predecessor of $a_2$; So, we set $un_{\alpha}(a_2)=a_1$;

  \item (By the induction method) At step $i$, assume that the first $i-1$ items have been correctly added into the QN-List (in essence, the QN-List over the subsequence of the first $(i-1)$ items of $\alpha$ is built), let's consider how to add the $i$-th item $a_i$ into the data structure. Let $m$ denote the number of horizontal lists in the current $\mathbb{L}_{\alpha}$. Before adding $a_i$ into $\mathbb{L}_{\alpha}$, let's first figure out the rising length of $a_i$. Consider a horizontal list $\mathbb{L}_{\alpha}^{t}$, we have the following two conclusions\footnote{Readers can skip the following paragraphs (a) and (b) if they only care about the construction steps.}: 
  \vspace{-0.05in}
  \begin{enumerate}

  \setlength\itemsep{0.0em}
  \item 
	  If $Tail(\mathbb{L}_{\alpha}^{t})$ $> a_i$, then $RL_{\alpha}(a_i)$ $\leq t$. Assume that $RL_{\alpha}(a_i)$ $> t$. It means that there exits at least one item $a_j$ ($\in$ $\mathbb{L}_{\alpha}^{t}$) such that $a_j $ $\seqpar{\alpha}$ $a_i$, i.e., $a_j$ is a predecessor (or recursive predecessor) of $a_i$. As we know $Tail(\mathbb{L}_{\alpha}^{t})$ is the minimum item in $\mathbb{L}_{\alpha}^{t}$ (see Lemma \ref{lem:property}).  $Tail(\mathbb{L}_{\alpha}^{t})$ $> a_i$ means that all items in $\mathbb{L}_{\alpha}^{t}$ are larger than $a_i$. That is contradicted to $a_j $ $\seqpar{\alpha}$ $a_i \wedge a_j \in \mathbb{L}_{\alpha}^{t}$. Thus,  $RL_{\alpha}(a_i)$ $\leq t$.
	  \nop{
	  We know that $Tail(\mathbb{L}_{\alpha}^{t})$ is the minimum item in $\mathbb{L}_{\alpha}^{t}$, namely, the minimum item before $a_i$ whose rising length is $t$ (see Lemma \ref{lem:property}). Thus, if $RL_{\alpha}(a_i)$ $> t$, it means that there exits at least one item $a_j$ ($\in$ $\mathbb{L}_{\alpha}^{t}$) such that $a_j $ $\seqpar{\alpha}$ $a_i$, i.e., $a_j$ is a predecessor (or recursive predecessor) of $a_i$. However, $Tail(\mathbb{L}_{\alpha}^{t})$ $> a_i$ means that all items in $\mathbb{L}_{\alpha}^{t}$ are larger than $a_i$. That is contradicted to  $a_j $ $\seqpar{\alpha}$ $a_i \wedge a_j \in \mathbb{L}_{\alpha}^{t}$. Thus,  $RL_{\alpha}(a_i)$ $\leq t$. }
	  
  \item
	  If $Tail(\mathbb{L}_{\alpha}^{t})$ $\leq a_i$, then $RL_{\alpha}(a_i) $ $> t$. Since $Tail(\mathbb{L}_{\alpha}^t)$ is before $a_i$ in $\alpha$ and $Tail(\mathbb{L}_{\alpha}^{t})$ $\leq a_i$, $Tail(\mathbb{L}_{\alpha}^t)$ is compatible $a_i$. Let us consider an increasing subseqeunce $s$ ending with $Tail(\mathbb{L}_{\alpha}^t)$, whose length is $t$ since $Tail(\mathbb{L}_{\alpha}^t)$'s rising length is $t$. Obviously, $s^{\prime}=s \oplus a_i$ is a length-(t+1) increasing subsequence ending with $a_i$. In other words, the rising length of $a_i$ is at least $t+1$, i.e, $RL_{\alpha}(a_i) $ $> t$.  
	 \end{enumerate}

	\vspace{-0.05in}
   Besides, we know that $Tail(\mathbb{L}_{\alpha}^t)$ $\geq Tail(\mathbb{L}_{\alpha}^{t'})$ if $t \geq $ $t^{\prime}$(see Lemma \ref{lem:tailsorted}). Thus, we need to find the first list $\mathbb{L}_{\alpha}^{t}$ whose tail $Tail(\mathbb{L}_{\alpha}^{t})$ is larger than $a_i$. Then, we append $a_i$ to the list. Since all tail items are increasing, we can perform the binary search (Lines 4-14 in Algorithm \ref{alg:insertelement}) that needs $O(\log m)$ time. If there is no such list, i.e., $Tail(\mathbb{L}_{\alpha}^{m}) \leq a_i$, we create a new empty list $Tail(\mathbb{L}_{\alpha}^{m+1})$ and insert $a_i$ into $Tail(\mathbb{L}_{\alpha}^{m+1})$.
 
 According to Lemma \ref{lem:consecutive}, it is easy to know $a_i$ can only be appended to the end of $\mathbb{L}_{\alpha}^{t}$, i.e., $rn_{\alpha}(Tail(\mathbb{L}_{\alpha}^t))$ $= a_i$ and $ln_{\alpha}(a_i)$ $= Tail(\mathbb{L}_{\alpha}^t)$. Besides, according to Lemma \ref{lem:property}(\ref{item:unrightmost}), we know that $un_{\alpha}(a_i)$ is the rightmost item in $\mathbb{L}_{\alpha}^{t-1}$ which is before $a_i$ in $\alpha$, then we set $un_{\alpha}(a_i)$ $= Tail(\mathbb{L}_{\alpha}^{t-1})$ (if exists). Analogously, we set $dn_{\alpha}(a_i)$ $= Tail(\mathbb{L}_{\alpha}^{t+1})$ (if exists).

So far, we correctly determine the four neighbors of $a_i$. We can repeat the above steps until all items are inserted to $\mathbb{L}_{\alpha}$.

\end{enumerate}
\vspace{-0.05in}

We divide the above building process into two pieces of pseudo codes. Algorithm \ref{alg:insertelement} presents pseudo codes for inserting one element into the current QN-List $\mathbb{L}_{\alpha}$, while Algorithm \ref{alg:buildingindex} loops on Algorithm \ref{alg:insertelement} to insert all items in $\alpha$ one by one to build the QN-List $\mathbb{L}_{\alpha}$. Initially, $\mathbb{L}_{\alpha} = \emptyset$. The QN-List $\mathbb{L}_{\alpha}$ obtained in Algorithm \ref{alg:buildingindex} will be called the \emph{corresponding data structure} of $\alpha$.

\vspace{-0.05in}
\begin{algorithm}[!h]
\small
\caption{Insert an element into $\mathbb{L}_{\alpha}$}
 \label{alg:insertelement}
\KwIn{$a_i$, an element to be inserted}
\KwOut{the updated QN-List $\mathbb{L}_{\alpha}$}

Let $m = |\mathbb{L}_{\alpha}|$  \\
Since the sequence \{$Tail(\mathbb{L}_{\alpha}^1)$, $Tail(\mathbb{L}_{\alpha}^2)$,...,$Tail(\mathbb{L}_{\alpha}^{m})$\} is increasing (Lemma \ref{lem:tailsorted}), we can conduct a binary search to determine minimum $k$ where $Tail(\mathbb{L}_{\alpha}^{k})$ $> a_i$.  \\
\nop{
$k = 1, min = 1, max = |\mathbb{L}_{\alpha}|$  \\
/* Binary search to find the exact $\mathbb{L}_{\alpha}^{k}$ to insert $a_i$ */ \\
\While{$min < max$}{
	\If{$Tail(\mathbb{L}_{\alpha}^{min}) > a_i$}{
		$k = min$ and BREAK
	}
	\If{$Tail(\mathbb{L}_{\alpha}^{max}) \leq a_i$}{
		$k = max + 1$ and BREAK
	}
	$mid = (max+1 + min)/2$ \\
	\If{$Tail(\mathbb{L}_{\alpha}^{mid}) \leq a_i$}{
		$min = mid + 1$ and CONTINUE)
	}
	\If{$Tail(\mathbb{L}_{\alpha}^{mid-1}) > a_i$}{
		$max = mid - 1$ and CONTINUE
	}
	$k = mid$ and BREAK \label{code:binaryEnd}
}\label{code:binaryBegin}
} 
\If{($\mathbb{L}_{\alpha}^{k}$ exists)}
{
    $a^* = Tail(\mathbb{L}_{\alpha}^{k})$;\\
    /*append $a_i$ to the list $\mathbb{L}_{\alpha}^{k}$*/\\
    $rn_{\alpha}(a^*)=a_i$;   $ln_{\alpha}(a_i)=a^*$; \\
    If $k>1$, then let $un_{\alpha}(a^*)$ = $Tail(\mathbb{L}_{\alpha}^{k-1})$; \\
    If $k$<$|\mathbb{L}_{\alpha}|$, then let $dn_{\alpha}(a^*)$ = $Tail(\mathbb{L}_{\alpha}^{k+1})$; \\
}
\Else
{
Create list $\mathbb{L}_{\alpha}^{m+1}$ in $\mathbb{L}_{\alpha}$ and add $a_i$ into $\mathbb{L}_{\alpha}^{m+1}$ \\
}
RETURN $\mathbb{L}_{\alpha}$
\end{algorithm}

\vspace{-0.27in}
\begin{algorithm}[!h]
\small
\caption{Building $\mathbb{L}_{\alpha}$ for a sequence $\alpha=\{a_1,...,a_w\}$}
 \label{alg:buildingindex}
\KwIn{a sequence $\alpha=\{a_1,...,a_w\}$}
\KwOut{the corresponding data structure $\mathbb{L}_{\alpha}$ of $\alpha$}

\For {each item $a_i$ in $\alpha$}
{
Call Algorithm \ref{alg:insertelement} to insert $a_i$ into $\mathbb{L}_{\alpha}$.
}
RETURN $\mathbb{L}_{\alpha}$;
\end{algorithm}

\vspace{-0.2in}
\begin{theorem} \label{theorem:timeconstruct}
Let $\alpha = \{a_1,a_2,...,a_w\}$ be a sequence with $w$ items. Then we have the following:
\vspace{-0.06in}
\begin{enumerate}
\setlength\itemsep{0.0em}
\item   The time complexity of Algorithm \ref{alg:insertelement} is $O(\log w)$.
\item The time complexity of Algorithm \ref{alg:buildingindex}  is $O(w \log w)$.
\end{enumerate}
\end{theorem}

\subsection{LIS Enumeration} \label{sec:lisenumnew}
Let's discuss how to enumerate all LIS of sequence $\alpha$ based on the QN-List $\mathbb{L}_{\alpha}$. Consider an LIS of $\alpha$ : $s = $ \{$a_{i_1}$, $a_{i_2}$,...,$a_{i_m}$\}. According to Lemma \ref{lem:property}(\ref{item:rlen}), $a_{i_m}$ $\in \mathbb{L}_{\alpha}^{m}$. In fact, the last item of each LIS must be located at the last horizontal list of $\mathbb{L}_{\alpha}$ and we can enumerate all LIS of $\alpha$ by enumerating all $|\mathbb{L}_{\alpha}|$ long increasing subsequence ending with items in $\mathbb{L}_{\alpha}^{|\mathbb{L}_{\alpha}|}$. For convenience, we use $MIS_{\alpha}(a_i)$ to denote the set of all  $RL_{\alpha}(a_i)$ long increasing subsequences ending with $a_i$. Formally, $MIS_{\alpha}(a_i)$ is defined as follows:
\vspace{-0.05in}
\[
MIS_{\alpha}(a_i) = \{s\;|\;s \in IS_{\alpha}(a_i) \land |s| = RL_{\alpha}(a_i) \}
\]
\vspace{-0.15in}

Consider each item $a_i$ in the last list $\mathbb{L}_{\alpha}^{|\mathbb{L}_{\alpha}|}$. We can compute all LIS of $\alpha$ ending with $a_i$ by iteratively searching for predecessors of $a_i$ in the above list from the bottom to up until reaching the first list $\mathbb{L}_{\alpha}^1$. This is the basic idea of our LIS enumeration algorithm.

For brevity, we virtually create a \emph{directed acyclic graph} (DAG) to more intuitively discuss the LIS enumeration on $\mathbb{L}_{\alpha}$. The DAG is defined based on the predecessor relationships between items in $\alpha$. Each vertex in the DAG corresponds to an item in $\alpha$. A directed edge is inserted from $a_i$ to $a_j$ if $a_j$ is a predecessor of $a_i$ ($a_i$ and $a_j$ is also called parent and child respectively).

\vspace{-0.05in}
\begin{definition}\label{def:daggraph}(DAG $G(\alpha)$).
Given a sequence $\alpha$, the directed graph $G$ is denoted as $G(\alpha)$ $=(V, E)$, where the vertex set $V$ and the edge set $E$ are defined as follows: \vspace{-0.05in}
\[
\begin{array}{l}
 V = \{ a_i |a_i  \in \alpha \} ; \; \;
 E = \{ (a_i ,a_j )|a_j \; is\; a\;predecessor\;of\;a_i \}  
 \end{array}
\]
\end{definition}
\vspace{-0.05in}

The $G(\alpha)$ over the sequence $\alpha=\{$3, 9, 6, 2, 8, 5, 7$\}$ is presented in Figure \ref{fig:lisdag}. We can see that each path with length $|\mathbb{L}_{\alpha}|$ in $G(\alpha)$ corresponds to an LIS. For example, we can find a path $a_5=8 \to a_3=6 \to a_1=3$, which is the reverse order of LIS \{$3$,$6$,$8$\}. Thus, we can easily design a DFS-like traverse starting from items in $\mathbb{L}_{\alpha}^{|\mathbb{L}_{\alpha}|}$ to output all path with length $|\mathbb{L}_{\alpha}|$ in $G(\alpha)$. 

Note that we do not actually need to build the DAG in our algorithm since we can equivalently  conduct the DFS-like traverse on $\mathbb{L}_{\alpha}$.  Firstly, we can easily access all items in $\mathbb{L}_{\alpha}$ which are the starting vertexes of the traverse. Secondly, the key operation in the DFS-like traverse is to get all predecessors of a vertex. In fact, according to Lemma \ref{lem:consecutive} which is demonstrated in Figure \ref{fig:consecutive}, we can find all predecessors of $a_i$ by searching $\mathbb{L}_{\alpha}^{t-1}$ from $un_{\alpha}(a_i)$ to the left until meeting an item $a^*$ that is not compatible with $a_i$. All touched items ($a^*$ excluded) during the search are predecessors of $a_i$.

\nop{
All operations over the DAG can be easily mapped to the data structure $\mathbb{L}_{\alpha}$ directly.

Recall that predecessors of an item $a_i$ $\in \mathbb{L}_{\alpha}^{t}$ form a \emph{nonempty consecutive block} ending with $un_{\alpha}(a_i)$ in $\mathbb{L}_{\alpha}^{t-1}$ (Lemma \ref{lem:consecutive}), which is demonstrated in Figure \ref{fig:consecutive}. Thus, we can find all predecessors of $a_i$ by searching $\mathbb{L}_{\alpha}^{t-1}$ from $un_{\alpha}(a_i)$ to the left until meeting an item $a^*$ that is not compatible with $a_i$. All touched items ($a^*$ excluded) during the search are predecessors of $a_i$. Algorithm \ref{alg:findpreds} presents the pseudo codes for finding all predecessors of $a_i$.

\begin{algorithm}[h!]
\small
\caption{Return Predecessors of $a_i$ based on $\mathbb{L}_{\alpha}$}
 \label{alg:findpreds}
\KwIn{$\alpha$, $\mathbb{L}_{\alpha}$ and $a_i$}
\KwOut{Set of predecessors of $a_i$: $Pred_{\alpha}(a_i)$}
Initial $Pred_{\alpha}(a_i) = \emptyset$ and $a = un_{\alpha}(a_i) $\\
	\If{$a \seqpar{\alpha} a_i$}{
		$Pred_{\alpha}(a_i).push(a)$ \\
		$a = ln_{\alpha}(a)$
	}
RETURN $Pred_{\alpha}(a_i)$
\end{algorithm}
}

We construct LIS $s$ from each item $a_{i_m}$ in $\mathbb{L}_{\alpha}^{m}$ (i.e., the last list) as follows. $a_{i_m}$ is first pushed into the bottom of an initially empty stack. At each iteration, the up neighbor of the top item is pushed into the stack. The algorithm continues until it pushes an item in $\mathbb{L}_{\alpha}^1$ into the stack and output items in the stack since this is when the stack holds an LIS. Then the algorithm starts to pop top item from the stack and push another predecessor of the current top item into stack. It is easy to see that this algorithm is very similar to depth-first search (DFS) (where the function call stack is implicitly used as the stack) and more specifically, this algorithm outputs all LIS as follows: (1) every item in $\mathbb{L}_{\alpha}^{m}$ is pushed into stack; (2) at each iteration, every predecessor (which can be scanned on a horizontal list from the up neighbor to left until discovering an incompatible item) of the current topmost item in the stack is pushed in the stack; (3) the stack content is printed when it is full (i.e., an LIS is in it).


\vspace{-0.10in}
\begin{theorem}\label{theorem:allLIS} 
The time complexity of our LIS enumeration algorithm is $O(${OUTPUT}$)$, where \emph{OUTPUT} is the total size of all LIS.
\end{theorem}
\vspace{-0.05in}

Pseudo code for LIS enumeration is presented in Appendix \citeAPPenum
\submitshow{
of the full version of this paper \cite{li2016lis}
}.

\vspace{-0.05in}
\section{Maintenance}\label{sec:maintenance}
\vspace{-0.05in}

When time window slides, $a_1$ is deleted and a new item $a_{w+1}$ is appended to the end of $\alpha$. It is easy to see that the quadruple neighbor list maintenance  consists of two operations: deletion of the first item $a_1$ and insertion of $a_{w+1}$ to the end. Algorithm \ref{alg:insertelement} in Section \ref{sec:indexconstruct} takes care of the insertion already. Thus we only consider ``deletion'' in this section.  The sequence $\{a_2,\cdots,a_w\}$ formed by deleting $a_1$ from $\alpha$ is denoted as $\alpha^-$. We divide the discussion of the quadruple neighbor list maintenance into two parts: the horizontal update for updating left and right neighbors and the vertical update for up and down neighbors. 


\vspace{-0.05in}
\subsection{Horizontal Update}
This section studies the horizontal update. We first introduce ``k-hop up neighbor'' that will be used in latter discussions. 

\vspace{-0.05in}
\begin{definition}\textbf{(k-Hop Up Neighbor)}.
\label{def:khopup}
Let $\alpha=$ $\{a_1$ $,a_2,...,a_w\}$ be a sequence and $\mathbb{L}_{\alpha}$ be its corresponding quadruple neighbor list. For $\forall a_i \in \alpha$, the \emph{k-hop up neighbor} $un^k_{\alpha}(a_i)$ is defined as follows:
\vspace{-0.05in}
  $$
  un^k_{\alpha}(a_i)\ =\
  \begin{cases}
  a_i & k = 0 \\
  un_{\alpha}(un^{k-1}_{\alpha}(a_i)) & k \ge 1
  \end{cases}
  $$
\end{definition}

To better understand our method, we first illustrate the main idea and the algorithm's sketch using a running example. More analysis and algorithm details are given afterward. 

\Paragraph{Running example and intuition.} Figure \ref{fig:maintenance}(a) shows the corresponding QN-list $\mathbb{L}_{\alpha}$ for the sequence $\alpha$ in the running example. After deleting $a_1$, some items in $\mathbb{L}_{\alpha}^t$ ($1\leq t \leq m$) should be promoted to the above list $\mathbb{L}_{\alpha}^{t-1}$ and the others are still in $\mathbb{L}_{\alpha}^{t}$. The following Theorem \ref{lem:deletion} tells us how to distinguish them. In a nutshell, given an item $a \in \mathbb{L}_{\alpha}^t$ ($1<t \leq m$), if its $(t-1)$-hop up neighbor is $a_1$ (the item to be deleted), $a$ should be promoted to the above list; otherwise, $a$ is still in the same list. 

For example, Figure \ref{fig:maintenance}(a) and \ref{fig:maintenance}(b) show the QN-lists before and after deleting $a_1$. $\{a_2,a_3\}$ are in $\mathbb{L}_{\alpha}^{2}$ and their 1-hop up neighbors are $a_1$ (the item to be deleted), thus, they are promoted to the first list of $\mathbb{L}_{\alpha^-}$. Also, $\{a_4\}$ is in $\mathbb{L}_{\alpha}^{3}$, whose 2-hop up neighbor is also $a_1$. It is also promoted to $\mathbb{L}_{\alpha^-}^2$. More interesting, for each horizontal list $\mathbb{L}_{\alpha}^t$ ($1\leq t \leq m$), the items that need to be promoted are on the left part of  $\mathbb{L}_{\alpha}^t$, denoted as $Left(\mathbb{L}_{\alpha}^{t})$, which are the shaded ones in Figure \ref{fig:maintenance}(a). Note that $Left(\mathbb{L}_{\alpha}^{1})=\{a_1\}$. The right(remaining) part of $\mathbb{L}_{\alpha}^t$ is denoted as $Right(\mathbb{L}_{\alpha}^{t})$. The horizontal update is to couple $Left(\mathbb{L}_{\alpha}^{t+1})$ with $Right(\mathbb{L}_{\alpha}^{t})$ into a new horizontal list $\mathbb{L}_{\alpha^{-}}^t$. For example, $Left(\mathbb{L}_{\alpha}^{2})= \{a_2,a_3\}$ plus  $Right(\mathbb{L}_{\alpha}^{1})=\{a_4\}$ to form $\mathbb{L}_{\alpha^{-}}^1=\{a_2,a_3,a_4\}$, as shown in Figure \ref{fig:maintenance}(b). Furthermore, the red bold line in Figure \ref{fig:maintenance}(a) denotes the \emph{separatrix} between the left and the right part, which starts from $a_1$. Algorithm \ref{alg:division} studies how to find the separatrix to divide each horizontal list $\mathbb{L}_{\alpha^{-}}^t$ into two parts efficiently.

\Paragraph{Analysis and Algorithm.}
Lemma \ref{theorem:nocross} tells us that the up neighbour relations of the two items in the same list do not cross, which is used in the proof of Theorem \ref{lem:deletion}. 

\vspace{-0.1in}
\begin{lemma}
\label{theorem:nocross}
Let $\alpha=\{a_1,...,a_w\}$ be a sequence and $\mathbb{L}_{\alpha}$ be its corresponding quadruple neighbor list. Let $m$ be the number of horizontal lists in $\mathbb{L}_{\alpha}$. Let $a_i$ and $a_j$ be two items in $\mathbb{L}_{\alpha}^t, t \geq 1$. If $a_i$ is on the left of $a_j$, $un^{k}_{\alpha}(a_i)=un^{k}_{\alpha}(a_j)$ or $un^{k}_{\alpha}(a_i)$ is on the left of $un^{k}_{\alpha}(a_j)$, for every $0\leq k< t$.
\nop{
then $lm^{k}_{\alpha}(a_i)$, $lm^{k}_{\alpha}(a_j)$, $un^{k}_{\alpha}(a_i)$ and $un^{k}_{\alpha}(a_j)$ are all in $\mathbb{L}_{\alpha}^{t-k}$, for every $1\leq k< t$. Furthermore, if $a_i < a_j$, then $lm^{k}_{\alpha}(a_i) \leq lm^{k}_{\alpha}(a_j)$ and $un^{k}_{\alpha}(a_i) \leq un^{k}_{\alpha}(a_j)$.
}
\end{lemma}
\nop{
\begin{proof}
If $t = 1$, $un_{\alpha}^{0}(a_i) = a_i$ is certainly on the left of $un_{\alpha}^{0}(a_j)$. If $t > 1$, $un_{\alpha}(a_i)$ is before $a_i$ in $\alpha$ (Lemma \ref{lem:def:undn}(1)). Since $a_i$ is on the left of $a_j$, $a_i$ is certainly before $a_j$ in $\alpha$ (Lemma \ref{lem:order}(\ref{item:decreasing})), hence, $un_{\alpha}^{a_i}$ is also before $a_j$ in $\alpha$. While, $un_{\alpha}(a_j)$ is the rightmost item in $\mathbb{L}_{\alpha}^{t-1}$ who is before $a_j$ (Lemma \ref{lem:def:undn}(\ref{item:un})). Thus, $un_{\alpha}(a_i)$ is either $un_{\alpha}(a_j)$ or an item on the left of $un_{\alpha}(a_j)$. Recursively,  for every $0\leq k < t$, $un_{\alpha}^{k}$ is either $un_{\alpha}^{k}(a_j)$ or an item on the left of $un_{\alpha}^{k}(a_j)$
\end{proof}
} 

\vspace{-0.15in}
\begin{theorem}
\label{lem:deletion}
Given a sequence $\alpha=$ \{$a_1,a_2,\cdots,a_w$\} and $\mathbb{L}_{\alpha}$. Let $m = $ $|\mathbb{L}_{\alpha}|$. Let $\alpha^- = \{a_2,\cdots, a_w\}$ be obtained from $\alpha$ by deleting $a_1$. Then for any $a_i, 2\leq i\leq m \in \mathbb{L}_{\alpha}^ t, 1\leq t \leq m$, we have the following:

\vspace{-0.05in}
\begin{enumerate}
\setlength\itemsep{0.0em}

\item If $un^{t-1}_{\alpha}(a_i)$ is $a_1$, then $RL_{\alpha^-}(a_i) = RL_{\alpha}(a_i)-1$.
\item If $un^{t-1}_{\alpha}(a_i)$ is not $a_1$, then $RL_{\alpha^-}(a_i) = RL_{\alpha}(a_i)$.
\end{enumerate}
\end{theorem}
\vspace{-0,1in}

\underline{Naive method.}
With Theorem \ref{lem:deletion}, the straightforward method to update horizontal lists is to compute  $un_{\alpha}^{t-1}(a_i)$ for each $a_i$ in $\mathbb{L}_{\alpha}^{t}$. If $un_{\alpha}^{t-1}(a_i)$ is $a_1$, promote $a_i$ into $\mathbb{L}_{\alpha}^{t-1}$. After grouping items into the correct horizontal lists, we sort the items of each horizontal list in the decreasing order of their values. According to Theorem \ref{lem:deletion} and Lemma \ref{lem:consecutive}(\ref{item:decreasing}) (which states that the horizontal list is in decreasing order), we can easily know that the horizontal lists obtained by the above process is the same as re-building $\mathbb{L}_{\alpha^-}$ for sequence $\alpha^{-}$ (i.e., the sequence after deleting $a_1$).

\underline{Optimized method.}
For each item $a_i$ in $\mathbb{L}_{\alpha}^{t}$ ($1\leq t \leq m$) in the running example, we report its  $(t-1)$-hop up neighbor in Figure \ref{fig:division}.  The shaded vertices denote the items whose $(t-1)$-hop up neighbors are $a_1$ in $ \mathbb{L}_{\alpha}^1$; and the others are in the white vertices. Interestingly, the two categories of items of a list form two consecutive blocks. The shaded one is on the left and the other on the right.

Let us recall Lemma \ref{theorem:nocross}, which says that the up neighbour relations of the two items in the same list do not cross. In fact, after deleting $a_1$, for each $a_i$ $\in \mathbb{L}_{\alpha}^{t}$, if $un_{\alpha}^{t-1}(a_i)$ is $a_1$, then for any item $a_j$ at the left side of $a_i$ in $\mathbb{L}_{\alpha}^{t}$, $un_{\alpha}^{t-1}(a_i)$ is also $a_1$. While, if $un_{\alpha}^{t-1}(a_i)$ is not $a_1$, then for any item $a_k$ at the right side of $a_i$ in $\mathbb{L}_{\alpha}^{t}$, $un_{\alpha}^{t-1}(a_i)$ is not $a_1$. The two claims can be proven by Lemma \ref{theorem:nocross}. This is the reason why two categories of items form two consecutive blocks, as shown in Figure \ref{fig:division}.

After deleting $a_1$, we divide each list $\mathbb{L}_{\alpha}^{t}$ into two sublists: $Left(\mathbb{L}_{\alpha}^{t})$ and $Right(\mathbb{L}_{\alpha}^{t})$. For any item $a_j$ $\in Left(\mathbb{L}_{\alpha}^{t})$, $un_{\alpha}^{t-1}(a_j)$ is $a_1$ while for any item $a_k$ $\in Right(\mathbb{L}_{\alpha}^{t})$, $un_{\alpha}^{t-1}(a_k)$ is not $a_1$. Instead of computing the $(t-1)$-hop up neighbor of each item, we propose an efficient algorithm (Algorithm \ref{alg:division}) to divide each horizontal list $\mathbb{L}_{\alpha}^{t}$ into two sublists: $Left(\mathbb{L}_{\alpha}^{t})$ and $Right(\mathbb{L}_{\alpha}^{t})$.

\begin{figure}[!h]
\vspace{-0.1in}
\centering
\begin{subfigure}[t]{0.45\textwidth}
\centering
	\resizebox{\linewidth}{!}
	{
		\includegraphics[width=0.60\textwidth]{\picfolder 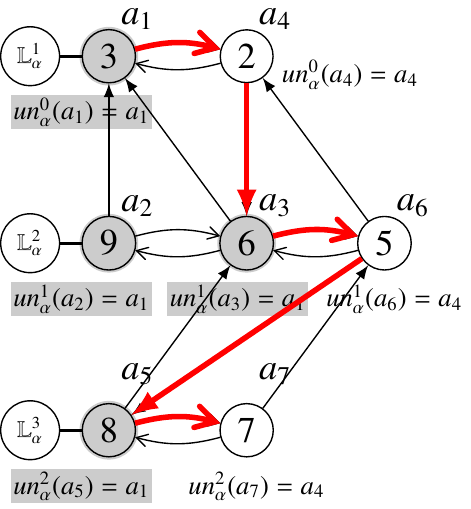}
	}
	\caption{Division.}
	\vspace{-0.1in}
	\label{fig:division}
\end{subfigure}
\begin{subfigure}[t]{0.45\textwidth}
\centering
	\resizebox{\linewidth}{!}
	{
		\includegraphics[width=0.60\textwidth]{\picfolder 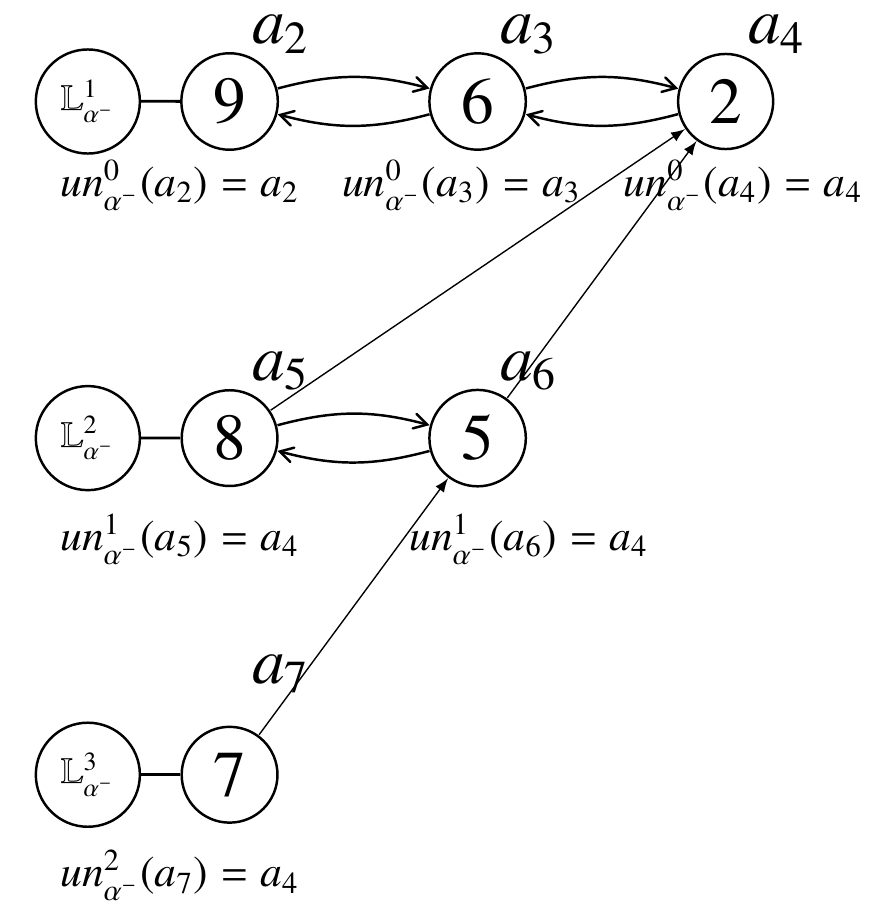}
	}
	\caption{after Deletion.}
	\vspace{-0.1in}
	\label{fig:afterupdate}
\end{subfigure}
\caption{Maintenance}
\label{fig:maintenance}
\end{figure}
\vspace{-0.15in}

Let's consider the division of each horizontal list of $\mathbb{L}_{\alpha}$. In fact, in our division algorithm, the division of $\mathbb{L}_{\alpha}^{t}$ depends on that of $\mathbb{L}_{\alpha}^{t-1}$. We first divide $\mathbb{L}_{\alpha}^{1}$. Apparently, $Left(\mathbb{L}_{\alpha}^{1})$ $= \{a_1\}$ and $Right(\mathbb{L}_{\alpha}^{1})$ $= \{\mathbb{L}_{\alpha}^{1}\} -{\{a_1\}}$. Recursively, assuming that we have finished the division of $\mathbb{L}_{\alpha}^{t}$, $1 \leq t < m$, there are three cases to divide $\mathbb{L}_{\alpha}^{t+1}$. Note that for each item $a_i \in Left(\mathbb{L}_{\alpha}^{t})$, $un_{\alpha}^{t-1}(a_i)=a_1$; while for each item $a_i \in Right(\mathbb{L}_{\alpha}^{t})$, $un_{\alpha}^{t-1}(a_i) \neq a_1$.

\vspace{-0.05in}
\begin{enumerate}
\setlength\itemsep{0.0em}
\item If $Right(\mathbb{L}_{\alpha}^{t})$  $= NULL$, for any item $a_j$ $\in \mathbb{L}_{\alpha}^{t+1}$, we have $un_{\alpha}(a_j)$ $\in Left(\mathbb{L}_{\alpha}^{t})$, thus, $un_{\alpha}^{t}(a_j)$ is exactly $a_1$. Thus, all below lists are set to be the left part. Specifically, for any $t^{\prime}>t$, we set $Left(\mathbb{L}_{\alpha}^{t^{\prime}})$ $= \mathbb{L}_{\alpha}^{t^{\prime}}$ and $Right(\mathbb{L}_{\alpha}^{t^\prime})$ $=NULL$.
\item If $Right(\mathbb{L}_{\alpha}^{t})$ $\neq NULL$ and the head item of $Right(\mathbb{L}_{\alpha}^{t})$ is $a_k$:
\vspace{-0,06in}
	\begin{enumerate}
	 \item if $dn_{\alpha}(a_k)$ does not exist, namely, $\mathbb{L}_{\alpha}^{t+1}$ is empty at the time when $a_k$ is inserted into $\mathbb{L}_{\alpha}^{t}$, then all items in $\mathbb{L}_{\alpha}^{t+1}$ come after $a_k$ and their up neighbors are either $a_k$ or item at the right side of $a_k$, thus, the $t$-hop up neighbor of each item in $\mathbb{L}_{\alpha}^{t+1}$ cannot be $a_1$. Actually, all below lists are set to be the right part. Specifically, for any $t^{\prime}>t$, we set $Left(\mathbb{L}_{\alpha}^{t^{\prime}})$ $=NULL$ and $Right(\mathbb{L}_{\alpha}^{t^\prime})$ $= \mathbb{L}_{\alpha}^{t^{\prime}}$. 
	 \item if $dn_{\alpha}(a_k)$ exists, then $dn_{\alpha}(a_k)$ and items at its left side come before $a_k$ and their up neighbors can only be at the left side of $a_k$ (i.e., $Left(\mathbb{L}_{\alpha}^{t})$), thus, the $t$-hop up neighbor of $dn_{\alpha}(a_k)$ or items on the left of $dn_{\alpha}(a_k)$ must be $a_1$. Besides, items at the right side of $dn_{\alpha}(a_k)$ come after $a_k$, and their up neighbors is either $a_k$ or item at the right side of $a_k$, thus, the $t$-hop up neighbor of each item on the right of $dn_{\alpha}$ cannot be $a_1$. Generally, we set $Left(\mathbb{L}_{\alpha}^{t+1})$ as the induced sublist from the head of $\mathbb{L}_{\alpha}^{t+1}$ to $dn_{\alpha}(a_k)$(included) and set $Right(\mathbb{L}_{\alpha}^{t+1})$ as the remainder, namely, $Right(\mathbb{L}_{\alpha}^{t+1})=\mathbb{L}_{\alpha}^{t+1}$ $- Left(\mathbb{L}_{\alpha}^{t+1})$. We iterate the above process for the remaining lists.
	\end{enumerate}
\end{enumerate}
\vspace{-0.1in}

\nop{
We illustrate the above division approach by the running example (see Figure \ref{fig:division}). Initially,  $Left(\mathbb{L}_{\alpha}^1)$=$\{a_1\}$ and $Right(\mathbb{L}_{\alpha}^1)$=$\{a_4\}$. The head of $Right(\mathbb{L}_{\alpha}^1)$ is $a_4$. Thus, in the list $\mathbb{L}_{\alpha}^2$, the sublist from the head ($a_2$) to $dn_{\alpha}(a_4)=a_3$ forms $Left(\mathbb{L}_{\alpha}^2)=\{a_2,a_3\}$. The remainder forms $Right(\mathbb{L}_{\alpha}^2)=\{a_6\}$. Iteratively, we obtain $Left(\mathbb{L}_{\alpha}^3)=\{a_5\}$ and $Right(\mathbb{L}_{\alpha}^3)=\{a_7\}$. Actually, the red bold line in the Figure denotes the \emph{separatrix} between the left and the right part.
}

Finally, for any $1\leq t \leq m$, the left sublist $Left(\mathbb{L}_{\alpha}^t)$ should be promoted to the above list; and $Right(\mathbb{L}_{\alpha}^t)$ is still in the $t$-th list. Specifically, $\mathbb{L}_{\alpha^{-}}^t=Left(\mathbb{L}_{\alpha^{-}}^{t+1})+Right(\mathbb{L}_{\alpha^{-}}^{t})$, i.e., appending $Right(\mathbb{L}_{\alpha^{-}}^{t})$ to $Left(\mathbb{L}_{\alpha^{-}}^{t+1})$ to form $\mathbb{L}_{\alpha^{-}}^t$.  In the running example, we append $Right(\mathbb{L}_{\alpha}^{1})=\{a_2,a_3\}$ to $Left(\mathbb{L}_{\alpha}^{2})=\{a_4\}$ to form $\mathbb{L}_{\alpha^{-}}^{1}=\{a_2,a_3,a_4\}$, as shown in Figure \ref{fig:afterupdate}.

\vspace{-0.06in}
\begin{theorem}\label{theorem:updatedecreasing} 
The list formed by appending $Right(\mathbb{L}_{\alpha}^{t})$ to $Left(\mathbb{L}_{\alpha}^{t+1})$ are monotonic decreasing from the left to the right.
\end{theorem}
\nop{
\begin{proof}
Since $Left(\mathbb{L}_{\alpha}^{t+1})$ is sublist of $\mathbb{L}_{\alpha}^{t+1}$ that is strictly decreasing, $Left(\mathbb{L}_{\alpha}^{t+1})$ is strictly decreasing too. Similar, $Right(\mathbb{L}_{\alpha}^{t})$ is also monotonic decreasing. If $Left(\mathbb{L}_{\alpha}^{t+1})$ or $Right(\mathbb{L}_{\alpha}^{t})$ is $NULL$, this theorem holds certainly. Otherwise, let $a_j$ be the last item in $Left(\mathbb{L}_{\alpha}^{t+1})$ and $a_k$ be the first item $Right(\mathbb{L}_{\alpha}^{t})$. According to the way we divide horizontal lists of $\mathbb{L}_{\alpha}$, $a_j$ is the down neighbour of $a_k$. Thus, $a_k < dn_\alpha(a_k) = a_j$(Lemma \ref{alg:buildingindex}(\ref{item:dnlarger})). Therefore, the list formed by appending $Right(\mathbb{L}_{\alpha}^{t})$ to $Left(\mathbb{L}_{\alpha}^{t+1})$ is monotonic decreasing from the left to the right.
\end{proof}
}
\vspace{-0.08in}

According to Theorem \ref{lem:deletion} and Lemma \ref{lem:property}(1), we can prove that the list formed by appending $Right(\mathbb{L}_{\alpha}^{t})$ to $Left(\mathbb{L}_{\alpha}^{t+1})$,  denoted as $L$, contains the same set of items as $\mathbb{L}_{\alpha^-}^{t}$ does. Besides, according to Lemma \ref{lem:consecutive}(2) and Theorem \ref{theorem:updatedecreasing}, both $L$ and $\mathbb{L}_{\alpha^-}^{t}$ are monotonic decreasing, thus, we can know that $L$ is equivalent to $\mathbb{L}_{\alpha^-}^{t}$ and we can derive that the horizontal list adjustment method is correct.

\vspace{-0.05in}
\begin{algorithm}[!h]
\small
\caption{Divide each horizontal list after deletion}
\label{alg:division}
\KwIn{$\mathbb{L}_{\alpha}$: the quadruple neighbor list for $\alpha$.}
\KwIn{$a_1$: the item to be deleted.}
\KwOut{$Left(\mathbb{L}_{\alpha}^{t})$ and $Right(\mathbb{L}_{\alpha}^{t})$ for each $\mathbb{L}_{\alpha}^{t}$, $1\leq t \leq m$.}

$m=|\mathbb{L}_{\alpha}|$ \\
$Left(\mathbb{L}_{\alpha}^{1})=\{a_1\}$ \\
$Right(\mathbb{L}_{\alpha}^{1})=\mathbb{L}_{\alpha}^{1}/\{a_1\}$\\
\For{$t \gets 1$ to $m-1$}
{
    \If{$Right(\mathbb{L}_{\alpha}^{t}) = NULL$}
    {
     \For{$t' \gets (t+1)$ to $m$}
      {
	      $Left(\mathbb{L}_{\alpha}^{t'})=\mathbb{L}_{\alpha}^{t'}$ \\
	       $Right(\mathbb{L}_{\alpha}^{t'})=NULL$\\
      }
      RETURN
    }
    $a_k = Head(Right(\mathbb{L}_{\alpha}^{t}))$ \\
    \If{$dn_{\alpha}(a_k) = NULL$}{
      \For{$t' \gets (t+1)$ to $m$}
       {
	       $Left(\mathbb{L}_{\alpha}^{t'})=NULL$ \\
	       $Right(\mathbb{L}_{\alpha}^{t'})=\mathbb{L}_{\alpha}^{t'}$ \\
       }
       RETURN
    }

	Set $Left(\mathbb{L}_{\alpha}^{t+1})$ to be the part at the left side of $dn_{\alpha}(a_k)$ in $\mathbb{L}_{\alpha}^{t+1}$, including $dn_{\alpha}(a_k)$ itself. \\
	$Right(\mathbb{L}_{\alpha}^{t+1}) = \mathbb{L}_{\alpha}^{t+1} - Left(\mathbb{L}_{\alpha}^{t+1})$
}
RETURN
\end{algorithm}
\vspace{-0.15in}

\subsection{Vertical Update}

Besides adjusting the horizontal lists, we also need to update the vertical neighbor relationship in the quadruple neighbor list to finish the transformation from $\mathbb{L}_{\alpha}$ to $\mathbb{L}_{\alpha^-}$. Before presenting our method, we recall Lemma \ref{lem:property}(2), which says, for item $a_i \in \mathbb{L}_{\alpha}^{t}$, $un_\alpha(a_i)$(if exists) is the rightmost item in $\mathbb{L}_{\alpha}^{t-1}$ who is before $a_i$ in sequence $\alpha$; while, $dn_\alpha(a_i)$(if exists) is the rightmost item in $\mathbb{L}_{\alpha}^{t+1}$ who is before $a_i$ in sequence $\alpha$.

\Paragraph{Running example and intuition.} Let us recall Figure \ref{fig:maintenance}. After adjusting the horizontal lists, we need to handle updates of vertical neighbors. The following Lemma \ref{lem:vertical} tells us which vertical relations will remain when transforming  $\mathbb{L}_{\alpha}$ into  $\mathbb{L}_{\alpha^{-}}$. Generally, when we promote $Left(\mathbb{L}_{\alpha}^{t})$ to the above level, we need to change their up neighbors but not down neighbors. While, $Right(\mathbb{L}_{\alpha}^{t})$ is still in the same level after the horizontal update. We need to change their down neighbors but not up neighbors. 

For example, $Left(\mathbb{L}_{\alpha}^{3})=\{a_5\}$ is promoted to the $\mathbb{L}_{\alpha^{-}}^{2}$. In  $\mathbb{L}_{\alpha}$, $un_{\alpha}(a_5)$ is $a_3$, but we change it to $un_{\alpha^{-}}(a_5)=$ $a_4$, i.e.,  the rightmost item in $\mathbb{L}_{\alpha^{-}}^{1}$ who is before $a_5$ in sequence $\alpha^{-}$. Analogously,  $Right(\mathbb{L}_{\alpha}^{2})=\{a_6\}$ is still at the second level of $\mathbb{L}_{\alpha^{-}}$. $dn_{\alpha}(a_6)$ is $a_5$, but we change it to null (i.e., $dn_{\alpha^{-}}(a_6)=null$), since there is no item in $\mathbb{L}_{\alpha^{-}}^{3}$ who is before $a_6$. We give the formal analysis and algorithm description of the vertical update as follows. 


\Paragraph{Analysis and Algorithm.}




\vspace{-0.05in}
\begin{lemma}
\label{lem:dnleftunright}
  Given a sequence $\alpha$ and $\mathbb{L}_{\alpha}$, for any $1\leq t \leq m$: 
  \vspace{-0.05in}
  \begin{enumerate}
  \setlength\itemsep{0.0em}
  
  \item
  $\forall a_i \in$ $Left(\mathbb{L}_{\alpha}^{t})$, $dn_{\alpha}(a_i)$ (if exists) $\in$ $Left(\mathbb{L}_{\alpha}^{t+1})$.
  \item $\forall a_i \in$ $Right(\mathbb{L}_{\alpha}^{t+1})$, $un_{\alpha}(a_i)$ (if exists) $\in$ $Right(\mathbb{L}_{\alpha}^{t})$.
  \end{enumerate}
\end{lemma}
\nop{
\begin{proof}
\begin{enumerate}
	\item
	  Assuming that $a_h$ and $a_t$ are $Head(Right(\mathbb{L}_{\alpha}^{t}))$ and $Tail(Left(\mathbb{L}_{\alpha}^{t+1}))$ respectively. $dn_{\alpha}(a_i)$ is before $a_i$ in $\alpha$ (Lemma \ref{lem:def:undn}), while $a_i$ is before $a_h$ in $\alpha$ according to the horizontal adjustment, thus, $dn_{\alpha}(a_i)$ is before $a_h$ in $\alpha$. Since $dn_{\alpha}(a_h) = a_t$, $a_t$ is the rightmost item in $\mathbb{L}_{\alpha}^{t+1}$ who is before $a_h$ in $\alpha$. Thus, $dn_{\alpha}(a_i)$ is either $a_t$ or some item on the left of $a_t$. Since $a_t$ is the tail item of $Left(\mathbb{L}_{\alpha}^{t+1})$, $dn_{\alpha}(a_i) \in Left(\mathbb{L}_{\alpha}^{t+1})$.
	\item
     Assuming that $a_h$ and $a_t$ are $Head(Right(\mathbb{L}_{\alpha}^{t}))$  and $Tail(Left(\mathbb{L}_{\alpha}^{t+1}))$ respectively. $a_t$ is the rightmost item in $\mathbb{L}_{\alpha}^{t+1}$ that is before $a_h$ in $\alpha$ since $dn_{\alpha}(a_h) = a_t$ (Lemma \ref{lem:def:undn}). while $a_i$ is on the right of $a_t$ in $\mathbb{L}_{\alpha}^{t+1}$, then $a_i$ can't be an item before $a_h$ in $\alpha$ , namely, $a_h$ is before $a_i$ in $\alpha$. Since $un_{\alpha}(a_i)$ is the rightmost item in $\mathbb{L}_{\alpha}^{t}$ that is before $a_i$ in $\alpha$, then $un_{\alpha}(a_i)$ can't be on the left of $a_h$. Hence,  $un_{\alpha}(a_i) \in$ $Right(\mathbb{L}_{\alpha}^{t})$.
\end{enumerate}
\end{proof}
}

\begin{lemma} \label{lem:vertical}
Let $\alpha= \{a_1,a_2,\cdots,a_w)$ be a sequence. Let $\mathbb{L}_{\alpha}$ be its corresponding quadruple neighbor list and $m$ be the total number of horizontal lists in $\mathbb{L}_{\alpha}$. Let $\alpha^- = \{a_2,\cdots, a_w\}$ be obtained from $\alpha$ by deleting $a_1$. Consider an item $a_i \in \mathbb{L}_{\alpha^{-}}^{t}$, where $1\leq t \leq m$. According to the horizontal list adjustment, there are two cases for $a_i$: $a_i$ is from $Left(\mathbb{L}_{\alpha}^{t+1})$ or $a_i$ is from $Right(\mathbb{L}_{\alpha}^{t})$. Then, the following claims hold:
\vspace{-0.05in}
	\begin{enumerate}
	\setlength\itemsep{0.0em}
	
	\item Assuming $a_i$ is from $Left(\mathbb{L}_{\alpha}^{t+1})$
	\vspace{-0.1in}
		\begin{enumerate}
		\setlength\itemsep{0.0em}
		
		\item \label{item:dnremain}
			$dn_{\alpha^-}(a_i) = dn_{\alpha}(a_i)$ (i.e., the down neighbor do not change).
		\item \label{item:unnotx}
			Let $x$ be the rightmost item of $Left(\mathbb{L}_{\alpha}^{t})$.
		If $un_{\alpha}(a_i) \neq x$, then $un_{\alpha^-}(a_i) = un_{\alpha}(a_i)$ (i.e., the up neighbor remains).
		
		
		\end{enumerate}
	\vspace{-0.05in}
		\item
			Assuming $a_i$ is from $Right(\mathbb{L}_{\alpha}^{t})$
		\vspace{-0.1in}
		\begin{enumerate}
		\item \label{item:unremain}
			$un_{\alpha^-}(a_i) = un_{\alpha}(a_i)$ (i.e., the up neighbor do not change).
		\item \label{item:dnnoty}
			Let $y$ be the rightmost item of $Left(\mathbb{L}_{\alpha}^{t+1})$.		
		If $dn_{\alpha}(a_i) \neq y$, $dn_{\alpha^-}(a_i) = dn_{\alpha}(a_i)$ (i.e., the down neighbor remains)
		\end{enumerate}
	
	\end{enumerate}

\end{lemma}

With Lemma \ref{lem:vertical}, for an item $a_i \in \mathbb{L}_{\alpha^-}^{t}$, there are only two cases that we need to update the vertical neighbor relations of $a_i$.

\vspace{-0.1in}
\begin{enumerate}
\setlength\itemsep{0.0em}

\item
Case 1: $a_i$ is from $Left(\mathbb{L}_{\alpha}^{t+1})$. Let $x$ be the rightmost item of $Left(\mathbb{L}_{\alpha}^{t})$. We need to update the \emph{up neighbor} of $a_i$ in  $\mathbb{L}_{\alpha^-}$ if $un_{\alpha}(a_i)=x$. Figure \ref{fig:updateun} demonstrates this case.
\item
Case 2: $a_i$ is from $Right(\mathbb{L}_{\alpha}^{t})$. Let $y$ be the rightmost item of $Left(\mathbb{L}_{\alpha}^{t+1})$. We need to update the \emph{down neighbor} of $a_i$ in  $\mathbb{L}_{\alpha^-}$ if $dn_{\alpha}(a_i)=y$. Figure \ref{fig:updatedn} demonstrates this case.
\end{enumerate}
\vspace{-0.05in}

We illustrate the detailed process as follows.

\vspace{-0.05in}
\begin{figure}[!h]
\centering
\includegraphics[width=0.9\textwidth]{\picfolder 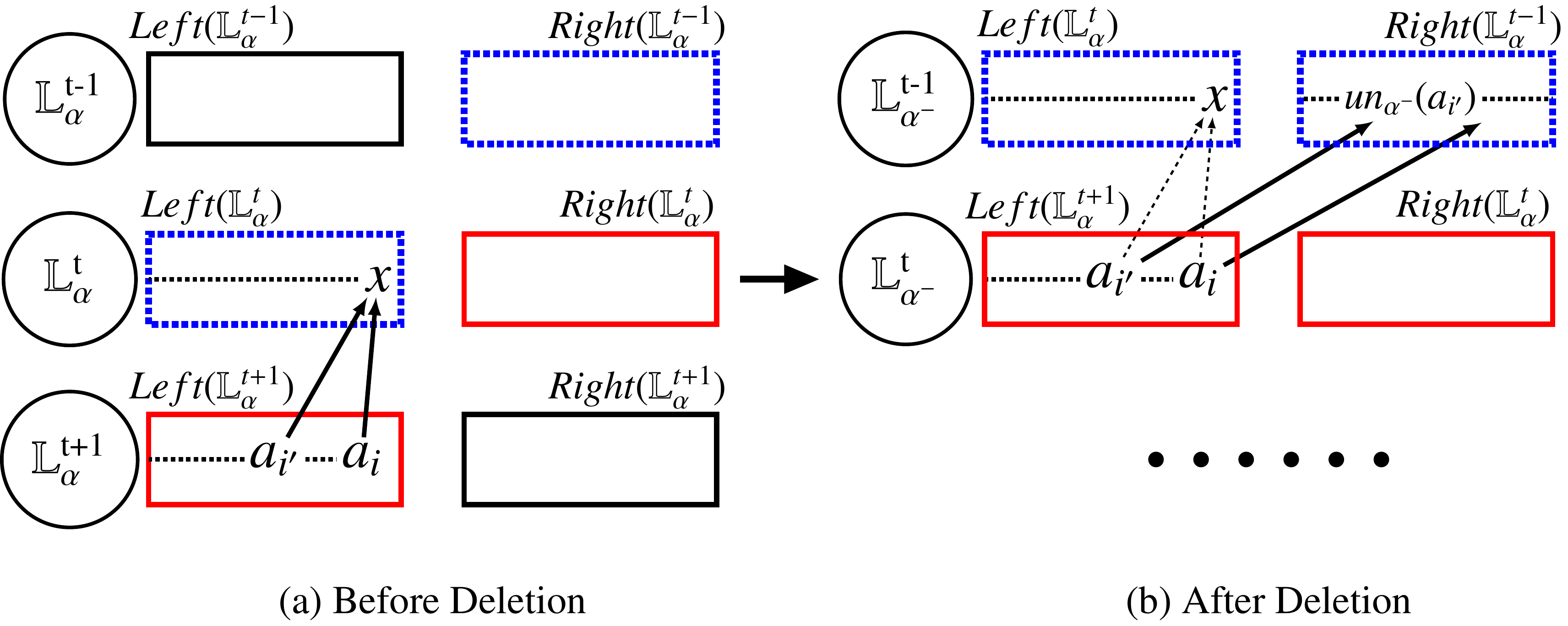}
\caption{Case 1: updating up neighbors}
\vspace{-0.1in}
\label{fig:updateun}
\end{figure}
\vspace{-0.1in}

\Paragraph{Case 1:} Consider all items in $Left(\mathbb{L}_{\alpha}^{t+1})$. According to the horizontal adjustment, $Left(\mathbb{L}_{\alpha}^{t+1})$ will be promoted into the list $\mathbb{L}_{\alpha^{-}}^{t}$.

Let $a_i$ be the rightmost item of $Left(\mathbb{L}_{\alpha}^{t+1})$ and $x=Tail(Left(\mathbb{L}_{\alpha}^{t}))$, namely, $x$ is the rightmost item in $Left(\mathbb{L}_{\alpha}^{t})$. According to Lemma \ref{lem:vertical}(1.b), if $un_{\alpha}(a_i)\neq x$, then $un_{\alpha}(a_i)$ $=un_{\alpha^-}(a_i)$. It is easy to prove that: If $un_{\alpha}(a_i)\neq x$ then $un_{\alpha}(a_j)\neq x$, where $a_j$ is on the left of $a_i$ in $Left(\mathbb{L}_{\alpha}^{t+1})$. In other words, all items in $Left(\mathbb{L}_{\alpha}^{t+1})$ do not change the vertical relations (see Lines \ref{code:unnochangeBegin}-\ref{code:unnochangeEnd} in Algorithm \ref{alg:updateun}).

Now, we consider the case that $un_{\alpha}(a_i)= x$ (Lines \ref{code:unchangeBegin}-\ref{code:unnochangeEnd} in Algorithm \ref{alg:updateun}). Then can scan $Left(\mathbb{L}_{\alpha}^{t+1})$ from $a_i$ to the left until finding the leftmost item $a_{i^{\prime}}$, where $un_{\alpha}(a_{i^{\prime}})$ is also $x$. The up neighbors of the items in the consecutive block from $a_{i^{\prime}}$ to $a_i$ (included both) are all $x$ in $\mathbb{L}_{\alpha}$ (note that $x$ is the rightmost item in $Left(\mathbb{L}_{\alpha}^{t})$ ), as shown in Figure \ref{fig:updateun}(a). These items' up neighbors need to be adjusted in  $\mathbb{L}_{\alpha^{-}}$. We work as follows: First, we adjust the up neighbor of $a_{i^{\prime}}$ in $\mathbb{L}_{\alpha^{-}}$. Initially, we set $a^{*}=un_{\alpha}(a_{i^{\prime}})=x$. Then, we move $a^{*}$ to the right step by step in $\mathbb{L}_{\alpha^{-}}^{t-1}$ until finding the rightmost item whose position is before $a_{i^{\prime}}$ in sequence $\alpha^{-}$. Finally, we set $un_{\alpha^{-}}(a_{i^{\prime}})=a^{*}$ (see Lines \ref{code:assginun} in Algorithm \ref{alg:updateun}).

In the running example, when deleting $a_1$ in Figure \ref{fig:division}, $Left(\mathbb{L}_{\alpha}^{3})$ $ = \{a_5\}$, and $un_{\alpha}(a_5)$ is exactly the tail item $a_3$ of $Left(\mathbb{L}_{\alpha}^{2})$, since $\mathbb{L}_{\alpha^-}^{1}$ is $\{a_2=9, a_3=6, a_4=2\}$, formed by appending $Right(\mathbb{L}_{\alpha}^{1})$($\{a_2=9, a_3=6\}$) to $Left(\mathbb{L}_{\alpha}^{2})$ ($\{a_4=2\}$), and $a_4$ is the rightmost item in $\mathbb{L}_{\alpha^-}^{1}$ who is before $a_5$ in $\alpha^-$, then we set $un_{\alpha^-}(a_5)$ as $a_4 = 2$, as shown in Figure \ref{fig:afterupdate}.

Iteratively, we consider the items on the right of $a_{i^{\prime}}$. Actually, the adjustment of the next item's up neighbor can begin from the current position of $a^{*}$ (Line \ref{code:unfromastar}). It is straightforward to know the time complexity of Algorithm \ref{alg:updateun} is $O(|\mathbb{L}_{\alpha^{-}}^{t-1}|)$, since each item in $\mathbb{L}_{\alpha^{-}}^{t-1}$ is scanned at most one time.

\Paragraph{Case 2:} Consider all items in $Right(\mathbb{L}_{\alpha}^{t})$. According to the horizontal adjustment, the down neighbors of items in $Right(\mathbb{L}_{\alpha}^{t})$ are the tail item (i.e., the rightmost item) of $Left(\mathbb{L}_{\alpha}^{t+1})$ or items in $Right(\mathbb{L}_{\alpha}^{t+1})$.

Actually, Case 2 is symmetric to Case 1. We highlight some important steps as follows. Let $a_i$ be the leftmost item in $Right(\mathbb{L}_{\alpha}^{t})$ and $y=Tail(Left(\mathbb{L}_{\alpha}^{t+1}))$, namely, $y$ is the rightmost item in $Left(\mathbb{L}_{\alpha}^{t+1})$. Obviously, $dn_{\alpha}(a_i)=y$, since the left-right division algorithm (Algorithm  \ref{alg:division}) guarantees that. Then we scan $Right(\mathbb{L}_{\alpha}^{t})$ from $a_i$ to the right until finding the rightmost item $a_{i^{\prime}}$, where $dn_{\alpha}(a_{i^{\prime}})$ is $y$. The up neighbors of the items in the consecutive block from $a_{i}$ to $a_{i^{\prime}}$ (included both) are all $y$ (see Figure \ref{fig:updatedn}(a)). Items on the right of $a_{i^{\prime}}$ need no changes in their down neighbors, since their down neighbors in $\mathbb{L}_{\alpha}$ are not $y$ (see Lemma \ref{lem:vertical}(2.b)).

\vspace{-0.1in}
\begin{figure}[!h]
\centering
\includegraphics[width=0.9\textwidth]{\picfolder 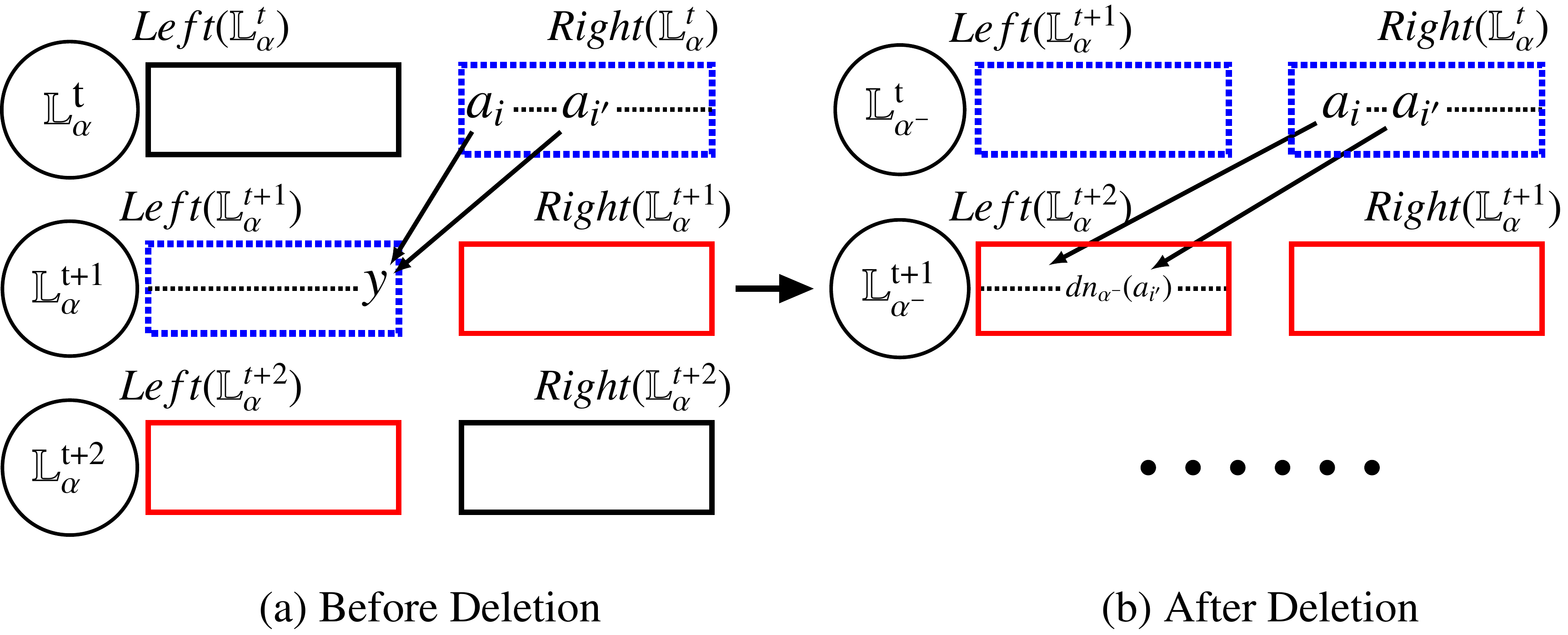}
\caption{Case 2: updating down neighbors}
\vspace{-0.1in}
\label{fig:updatedn}
\end{figure}
\vspace{-0.15in}

\vspace{-0.1in}
\begin{algorithm}[!h]
\small
\caption{Update up neighbors of items in $Left$($\mathbb{L}_{\alpha}^{t+1}$)}
\label{alg:updateun}
\KwIn{$Left$($\mathbb{L}_{\alpha}^{t+1}$))}
\KwOut{the updated $Left$($\mathbb{L}_{\alpha}^{t+1}$)}

\If{$Left(\mathbb{L}_{\alpha}^{t+1}) = NULL$}{
	RETURN
}

\If{$t = 1$}{
	\For{$a_i \in Left(\mathbb{L}_{\alpha}^{t+1})$}{
		$un_{\alpha^-}(a_i) = NULL$
	}
	RETURN
}
Let $a_i = Tail(Left(\mathbb{L}_{\alpha}^{t+1}))$ and
$x = Tail(Left(\mathbb{L}_{\alpha}^{t}))$ \\
\label{code:unnochangeBegin}
\If{$un_{\alpha}(a_i)$ is not $x$}{
	RETURN \label{code:unnochangeEnd}
}
\label{code:unchangeBegin}
Scan $Left(\mathbb{L}_{\alpha}^{t+1})$ from right to left and find the leftmost item whose up neighbor is $x$, denoted as $a_{i'}$ \\
$a^* = x$ \\
\While{$a_{i'} \geq a_i$}
{
	\While{$rn_{\alpha^-}(a^*)$ is before $a_{i'}$}{
		$a^* = rn_{\alpha^-}(a^*)$ \\ \label{code:unfromastar}
	}
	$un_{\alpha^-}(a_{i'}) = a^*$ \\ \label{code:assginun}
	$a_{i'} = rn_{\alpha^-}(a_{i'})$ \\ \label{code:unchangeEnd}
}
RETURN
\end{algorithm}
\vspace{-0.25in}
\begin{algorithm}[h!]
\small
\caption{Update down neighbors of items in $Right$($\mathbb{L}_{\alpha}^{t}$)}
\label{alg:updatedn}
\KwIn{$Right$($\mathbb{L}_{\alpha}^{t}$))}
\KwOut{the updated $Right$($\mathbb{L}_{\alpha}^{t}$)}

\If{$t \geq m-1$ OR $Right(\mathbb{L}_{\alpha}^{t}) = NULL$}{
	RETURN
}
Let $a_i = Head(Right(\mathbb{L}_{\alpha}^{t}))$ and $y = dn_{\alpha}(a_i)$ \\
Scan $Right(\mathbb{L}_{\alpha}^{t})$ from left to right and find the rightmost item whose down neighbor is $y$, denoted as $a_{i'}$ \\
$a^* = Tail(Left(\mathbb{L}_{\alpha}^{t+1}))$ \\
\While{$a_{i'} \leq a_i$}
{
	\If{$a^* = NULL$ OR $a^*$ is before $a_{i'}$}{
		$dn_{\alpha^-}(a_{i'}) = a^*$ \\ \label{code:assigndn}
		$a_{i'} = ln_{\alpha^-}(a_{i'})$ \\
	}\Else{
		$a^* = ln_{\alpha^-}(a^*)$ \\ \label{code:dnfromastar}
	}
}
RETURN
\end{algorithm}
\vspace{-0.1in}

We only consider the consecutive block from $a_{i}$ to $a_{i^{\prime}}$ (see Figure \ref{fig:updatedn}) as follows. First, we adjust the down neighbor of $a_{i^{\prime}}$ in $\mathbb{L}_{\alpha^{-}}$. Initially, we set $a^{*}=Tail(Left(\mathbb{L}_{\alpha}^{t+1}))$, i.e., the rightmost item of $Left(\mathbb{L}_{\alpha}^{t+1})$.  Then, we move $a^{*}$ to the left step by step in $\mathbb{L}_{\alpha^{-}}^{t+1}$ until finding the rightmost item whose position is before $a_{i^{\prime}}$. Finally, we set $dn_{\alpha^{-}}(a_{i^{\prime}})=a^{*}$ (see Lines \ref{code:assigndn} in Algorithm \ref{alg:updatedn}).

In the running example, when deleting $a_1 = 3$, $Right(\mathbb{L}_{\alpha}^{1})$ is $\{a_4=2\}$ whose head item is $a_4$. And $dn_{\alpha}(a_4)$ is $a_3$ $= 6$ that is the tail item of $Left(\mathbb{L}_{\alpha}^{2})$. Then, initially, we set $dn_{\alpha^-}(a_4)$ as the tail item of $Left(\mathbb{L}_{\alpha}^{3})$, namely, $dn_{\alpha^-}(a_4)=a_5 $ and scan $\mathbb{L}_{\alpha^-}^{2}$ from the right to the left until finding a rightmost item who is before $a_4$ in $\alpha^-$. Since there is no such item in $\mathbb{L}_{\alpha^-}^{2}$, we set $dn_{\alpha^-}(a_4)$ as $NULL$.

Iteratively, we consider the items on the left of $a_{i^{\prime}}$. Actually, the adjustment of the down neighbor can begin from the current position of $a^{*}$ (Line \ref{code:dnfromastar} in Algorithm \ref{alg:updatedn}). It is straightforward to know the time complexity of Algorithm \ref{alg:updatedn} is $O(|\mathbb{L}_{\alpha^{-}}^{t+1}|)$, since each item in $\mathbb{L}_{\alpha^{-}}^{t+1}$ is scanned at most twice.

\vspace{-0.05in}
\subsection{Putting It All Together}

Finally, we can see that solution to handle the deletion of the head item $a_1$ in sequence $\alpha$ consists two main phrase. The first phrase is to divides each list $\mathbb{L}_{\alpha}^{t}$ ($1\leq t \leq m$) using Algorithm \ref{alg:division} and then finishes the horizontal update by appending $Right(\mathbb{L}_{\alpha}^{t})$ to $Left(\mathbb{L}_{\alpha}^{t+1})$ . In the second phrase, we can call Algorithms \ref{alg:updateun} and \ref{alg:updatedn} for vertical update. Pseudo codes of algorithm handling deletion are presented in Appendix \citeAPPdelete
\submitshow{
of the full version of this paper \cite{li2016lis}
}.
\vspace{-0.1in}
\begin{theorem} \label{theo:deletetime}
The time complexity of our deletion algorithm is $O(w)$, where $w$ denotes the time window size.
\end{theorem}
\vspace{-0.06in}
\nop{
\begin{proof}
We can see that the time complexity of Algorithm \ref{alg:division} is $O(|LIS|)$ since division of each horizontal list costs $O(1)$ and there are $|LIS|$ horizontal lists in total. Besides, during the up neighbors update(Lines \ref{code:unupdateBegin}-\ref{code:unupdateEnd}), each horizontal list is scanned at most twice which means each item in $\alpha$ will be scanned at most twice. Similarly, during the down neighbors update(Lines \ref{code:dnupdateBegin}-\ref{code:dnupdateEnd}), each item in $\alpha$ is also scanned at most twice. We have $|\alpha| = w$, thus, the time complexity of Algorithm \ref{alg:deletion} is $O(|LIS|+w)$, namely, $O(w)$ since $|LIS|$ $\leq w$.
\end{proof}
}

\vspace{-0.1in}
\section{Computing LIS with constraints}\label{sec:computation}

As noted earlier in Section \ref{sec:introduction}, some applications are more interested in computing LIS with certain constraints. In this section, we consider four kinds of constraints (maximum/minimum weight/gap) that are defined in Section \ref{sec:problemdef}.

In Section \ref{sec:lisenumnew}, we define the DAG (Definition \ref{def:daggraph}) based on the \emph{predecessor} (Definition \ref{def:pred}). Each length-$m$ path in DAG denotes a LIS. Considering the equivalence between DAG and $\mathbb{L}_{\alpha}$, we illustrate our algorithm using DAG for the ease of the presentation. These algorithm steps can be easily mapped to those in $\mathbb{L}_{\alpha}$. According to  Lemma \ref{lem:consecutive}(2), items in $\mathbb{L}_{\alpha}^{t}$ ($1\leq t \leq m$) decrease from the left to the right. Thus, the leftmost length-$m$ path in DAG denotes the LIS with the maximum weight; while, the rightmost length-$m$ path denotes the LIS with the minimum weight. Formally, we define the \emph{leftmost child} as follows.

\vspace{-0.08in}
\begin{definition}\textbf{(Leftmost child)}. Given an item $a_i \in \mathbb{L}_{\alpha}^{t}$ ($1 \leq t \leq m$), the \emph{leftmost child} of $a_i$, denoted as $lm_{\alpha}(a_i)$, is the \emph{leftmost} predecessor (see Definition \ref{def:pred}) of $a_i$ in $\mathbb{L}_{\alpha}^{t-1}$.
\end{definition}
\vspace{-0.1in}

Recall Figure \ref{fig:lisdag}. $a_3$ is the leftmost child of $a_7$, denoted  as $a_3=lm_{\alpha}(a_7)$. Similar to the recursive definition of $k$-hop up neighbor $un^{k}_{\alpha}(a_i)$ for $a_i$, we recursively define  $lm^{k}_{\alpha}(a_i) = $ $lm_{\alpha}(lm_{\alpha}^{k-1}(a_i))$ $(k \geq 1)$ for any $k < t$, where $lm^{0}_{\alpha}(a_i)$ $= a_i$. Obviously, given an item $a_i \in \mathbb{L}_{\alpha}^m$ (i.e., the last list), $(lm_{\alpha}^{m-1}(a_i), lm_{\alpha}^{m-2}(a_i),...,lm_{\alpha}^{0}(a_i)=a_i)$ forms the leftmost path ending with $a_i$ in the DAG. It is easy to know that the leftmost path $(lm_{\alpha}^{m-1}(a_i), lm_{\alpha}^{m-2}(a_i),...,lm_{\alpha}^{0}(a_i)=a_i)$ in the DAG is the LIS ending with $a_i$ with maximum weight and minimum gap; while the rightmost path $(un_{\alpha}^{m-1}(a_i), un_{\alpha}^{m-2}(a_i),...,un_{\alpha}^{0}(a_i)=a_i)$ in the DAG is the LIS ending with $a_i$ with minimum weight and maximum gap. Formally, we have the following theorem that is the central to our constraint-based LIS computation.

\vspace{-0.08in}
\begin{theorem}
\label{theorem:constrainedmining}
Given a sequence $\alpha=\{a_1,...,a_w\}$ and $\mathbb{L}_{\alpha}$. Let $m=$ $|\mathbb{L}_{\alpha}|$ and DAG $G_{\alpha}$ be the corresponding DAG created  from $\mathbb{L}_{\alpha}$.
\begin{enumerate}
\setlength\itemsep{0.0em}
\vspace{-0.1in}
\item Given $a_i$, $a_j$ $\in \mathbb{L}_{\alpha}^t$ where $a_i < a_j$. then for every $1\leq k< t$,  $lm^{k}_{\alpha}(a_i)$, $lm^{k}_{\alpha}(a_j)$, $un^{k}_{\alpha}(a_i)$ and $un^{k}_{\alpha}(a_j)$ are all in $\mathbb{L}_{\alpha}^{t-k}$, and $lm^{k}_{\alpha}(a_i) \leq lm^{k}_{\alpha}(a_j)$, $un^{k}_{\alpha}(a_i) \leq un^{k}_{\alpha}(a_j)$.

\item \label{item:pmost:compare}
Given $a_i$ $\in \mathbb{L}_{\alpha}^t$. Consider an LIS $\beta$ ending with $a_i$: $\beta =\{a_{i_{t-1}}, \cdots,a_{i_{k}}, \cdots, a_{i_{0}}=a_i \}$. Then $\forall$ $k \in [0,t-1]$, $lm^{k}_{\alpha}(a_i)$, $a_{i_{k}}$, and $un^{k}_{\alpha}(a_i)$ are all in $\mathbb{L}_{\alpha}^{t-k}$. Also $lm^{k}_{\alpha}(a_i) \geq a_{i_{k}} \geq un^{k}_{\alpha}(a_i)$.

\item 
Given $a_i$ $\in \mathbb{L}_{\alpha}^{m}$ (the last list). Among all LIS ending with $a_i$, $\{$$lm^{m-1}_{\alpha}(a_i)$, $\cdots$,$lm^{0}_{\alpha}(a_i)$$\}$ has maximum weight and minimum gap,  while $\{$$un^{m-1}_{\alpha}(a_i)$, $\cdots$,$un^{0}_{\alpha}(a_i)$$\}$ has the minimum weight and maximum gap.

\item Let $a^{m}_h$ and $a^{m}_t$ be the head and tail of $\mathbb{L}_{\alpha}^{m}$ respectively. Then
the LIS  $\{$$lm^{m-1}_{\alpha}(a^{m}_h)$, $\cdots$,$lm^{0}_{\alpha}(a^{m}_h)$$\}$ has the maximum weight. The LIS $\{un^{m-1}_{\alpha}(a^{m}_t), \cdots,un^{0}_{\alpha}(a^{m}_t)\}$ has the minimum weight.

\end{enumerate}
\end{theorem}

\vspace{-0.1in}
\Paragraph{LIS with maximum/minimum weight}

Based on Theorem \ref{theorem:constrainedmining}(4), we can design algorithms 
\arxivshow{
	(Algorithms \ref{alg:lismaxweightnew} and \ref{alg:lisminweightnew} in Appendix \ref{sec:appendix:extremeweight})
}
to compute the unique LIS with maximum weight and the unique LIS with minimum weight, respectively 
\submitshow{
	(Refer to Appendix \citeAPPweight in the full version of this paper \cite{li2016lis} for pseudo codes)
}.
Generally speaking, it searches for the leftmost path and the rightmost path in the DAG. It is straightforward to know that both algorithms cost $O(w)$ time.

\Paragraph{LIS with maximum/minimum gap}

For the maximum gap (minimum gap, respectively) problem, there may be numerous LIS with maximum gap (minimum gap, respectively). In the running example, LIS with the maximum gap are \{$a_1=3$,$a_3=6$,$a_5=8$\} and \{$a_4=2$,$a_6=5$,$a_7=7$\}; while LIS with the minimum gap are \{$a_1=3$,$a_6=5$,$a_7=7$\} and \{$a_1=3$,$a_3=6$,$a_7=7$\}, as shown in Figure \ref{fig:timewindow}. According to Theorem \ref{theorem:constrainedmining}, it is obviously that if $\beta=\{a_{i_{m-1}},a_{i_{m-2}},...,a_{i_{0}}\}$ is a LIS with the maximum gap, $a_{i_{m-1}}=un_{\alpha}^{m-1}(a_{i_{0}})$; while, if $\beta$ is a LIS with the minimum gap, $a_{i_{m-1}}=lm_{\alpha}^{m-1}(a_{i_{0}})$. Note that, there may be multiple LIS sharing the same head and tail items. For example,  \{$a_1=3$,$a_6=5$,$a_7=7$\} and \{$a_1=3$,$a_3=6$,$a_7=7$\} are both LIS with minimum gap, but they share the same head and tail items. Since computing LIS with the minimum gap is analogous to LIS with the maximum gap, we only consider LIS with the maximum gap as follows.

For the maximum gap problem, one could compute \{$un^{m-1}_{\alpha}(a_i),$ $\cdots,un^{0}_{\alpha}(a_i)$\} for all $a_i \in \mathbb{L}_{\alpha}^{m}$ first,  and then figure out the maximum gap 
$\theta_{max}$ 
\arxivshow{((Line \ref{code:maximum} in Algorithm \ref{alg:lismaxheightnew})}
.
We design a sweeping algorithm 
\arxivshow{(Lines \ref{code:rmBegin}-\ref{code:rmEnd} in Algorithm \ref{alg:lismaxheightnew})}
with $O(w)$ time to compute $un^{t-1}_{\alpha}(a_i)$ for each item $a_i \in \mathbb{L}_{\alpha}^t$, $1\leq t \leq m$. We will return back to the sweeping algorithm at the end of this subsection. Here, we assume that $un^{m-1}_{\alpha}(a_i)$ has been computed for each item $a_i \in \mathbb{L}_{\alpha}$.

Assume that $a_i - un^{m-1}_{\alpha}(a_i) = \theta_{max}$ for some $a_i \in \mathbb{L}_{\alpha}^{m}$. We need to enumerate all LIS starting with $un^{m-1}(a_i)$ ending with $a_i$. We only need to slightly modify the LIS enumeration algorithm 
\arxivshow{(Algorithm \ref{alg:lisenum})}
as follows: Initially, we push $a_i$ into the stack. If $a_j \in \mathbb{L}_{\alpha}^{t}$ is a predecessor of top element in the stack, we push $a_j$ into the stack \emph{if and only if}  $un_{\alpha}^{t-1}(a_j)=$$un_{\alpha}^{m-1}(a_i)$ 
\arxivshow{(Line \ref{code:gapenum} in Algorithm \ref{alg:lismaxheightnew})}
.

To compute $un_{\alpha}^{t-1}(a_i)$ for each $a_i \in \mathbb{L}_{\alpha}^{t}$, $1\leq t \leq m$, we design a sweeping algorithm 
\arxivshow{(Lines \ref{code:rmBegin}-\ref{code:rmEnd} in Algorithm \ref{alg:lismaxheightnew})}
 from $\mathbb{L}_{\alpha}^{2}$ to $\mathbb{L}_{\alpha}^{m}$ and once we figure out $un_{\alpha}^{t-1}(a_j)$ for each $a_j$ in $\mathbb{L}_{\alpha}^{t}$, then for any $a_i \in$ $\mathbb{L}_{\alpha}^{t+1}$, $un_{\alpha}^{t}(a_i)$ is $un_{\alpha}^{t-1}(un_{\alpha}(a_i))$. Apparently, this sweeping algorithm takes $\Theta (w)$ time.

\vspace{-0.1in}
\begin{theorem} 
\label{timecom:maxh}
The time complexity of our algorithm for LIS with maximum gap is $O(w+$\emph{OUTPUT}$)$, where $w$ denotes the window size and \emph{OUTPUT} is the total length of all LIS with maximal gap.
\end{theorem}
\nop{
\begin{proof} The sweeping steps from $\mathbb{L}_{\alpha}^{2}$ to $\mathbb{L}_{\alpha}^{m}$ need to access each item at most one time. It takes $O(w)$ time. The output cost is at most one time of the output size. Therefore, the total cost is  $O(w+$\emph{OUTUT}$)$.
\end{proof}
} 
\vspace{-0.06in}

Computing LIS with minimum gap is analogous to that of LIS with maximum gap by computing $lm^{t-1}_{\alpha}(a_i)$ instead of $un^{t-1}_{\alpha}(a_i)$. We also design a sweeping algorithm
\arxivshow{
	(Lines \ref{code:lmBegin}-\ref{code:lmEnd} in Algorithm \ref{alg:mingap} in Appendix \ref{sec:appendix:extremegap})
}
from $\mathbb{L}_{\alpha}^{2}$ to $\mathbb{L}_{\alpha}^{m}$ to figure out $lm^{t-1}_{\alpha}(a_i)$ for each $a_i \in \mathbb{L}_{\alpha}^{t}$, $1\leq t \leq m$ 
\submitshow{
	(Pseudo codes are presented in the Appendix \citeAPPgap of the full version of this paper \cite{li2016lis})
}.

\vspace{-0.15in}
\section{Comparative Study}
\label{sec:compare}


To the best of our knowledge, there is no existing work that studies both LIS enumeration and LIS with constraints in the data stream model. In this section, we compare our method with five related algorithms, four of which are state-of-the-art LIS algorithms, i.e., {LISSET} \cite{lissetChen2007}, {MHLIS} \cite{minheight2009}, {VARIANT}  \cite{variant2009} and {LISone} \cite{liswAlbert2004}, and the last one is the classical dynamic program (DP) algorithm. None of them covers either the same computing model or the same computing task with our approach. Table \ref{tab:comparative} summarizes the differences between our approach with other comparative ones.

  
\textbf{LISSET} \cite{lissetChen2007} is the only one which proposed LIS enumeration in the context of ``stream model''. It enumerates all LIS in each sliding window but it fails to compute LIS with different constraints, such as LIS with extreme gaps and LIS with extreme weights. To enable the comparison in constraint-based LIS, we first compute all LIS followed by filtering using constraints to figure out constraint-based LIS, which is denoted as ``LISSET-Post'' in our experiments.



\textbf{MHLIS} \cite{minheight2009} is to find LIS with the minimum gap but it does not work in the context of data stream model. The data structure in MHLIS does not consider the maintenance issue. To enable the comparison, we implement two streaming version of MHLIS: \textbf{MHLIS+Rebuild} and \textbf{MHLIS+Ins/Del} where MHLIS+Rebuild is to re-compute LIS from scratch in each time window and MHLIS+Ins/Del is to apply our update method in MHLIS.


  
 A family of algorithms was proposed in \cite{variant2009} including LIS of minimal/maximal weight/gap (denoted as \textbf{VARIANT}). Since these algorithms are not intended for the streaming model, for the comparison, we implement two stream version of VARIANT: \textbf{VARIANT+Rebuild} and \textbf{VARIANT+Ins/Del} where VARIANT+Rebuild is to re-compute LIS from scratch in each time window and VARIANT+Ins/Del is to apply our update method in VARIANT.

We include the classical algorithm computing LIS based on dynamic programming (denoted as \textbf{DP}) in the comparative study. The standard DP LIS algorithm only computes the length of LIS and output a single LIS (not enumeration). To enumerate all LIS, we save all predecessors of each item when determining the maximum length of the increasing subsequence ending with it. 

\textbf{LISone}\cite{liswAlbert2004} computed LIS length and output an LIS in the sliding model. They  maintained the first row of Young's Tableaux when update happened. The length of the first row is exactly the LIS length of the sequence in the window. 

\vspace{-0.1in}
\begin{table}[!h]
\def\arraystretch{1}
\setlength{\tabcolsep}{0pt}
\centering

    \caption{Compaison between our method and the comparative ones} 
    \label{tab:comparative}

    \begin{small}
    \resizebox{!}{0.15\textwidth}
    {
    \begin{tabular}{|p{1.8cm}<{\centering}|p{1.4cm}<{\centering}|p{1.4cm}<{\centering}|p{2.2cm}<{\centering}|p{1.5cm}<{\centering}|p{1.0cm}<{\centering}|}
      \hline 
      \multirow{3}{*}{\textbf{Methods} } &{\textbf{Stream Model} }& \textbf{LIS  Enumeration} & \textbf{LIS with extreme weight} & \textbf{LIS with extreme gap} &\textbf{LIS length} \\ 
	  \hline
      Our Method & \cmark & \cmark & \cmark & \cmark & \cmark \\ 
      \hline
      LISSET & \cmark & \cmark & \xmark & \xmark & \cmark \\ 
      \hline
      MHLIS & \xmark & \xmark & \xmark & \cmark & \cmark \\ 
      \hline
      VARIANT & \xmark & \xmark & \cmark & \cmark & \cmark \\ 
      \hline
      DP & \xmark & \cmark & \xmark & \xmark & \cmark \\ 
      \hline
      LISone & \cmark & \xmark & \xmark & \xmark & \cmark \\ 
      \hline
    \end{tabular}
    }
    \end{small}
    \vspace{-0.1in}
\end{table}
\vspace{-0.1in}

\nop{straightforward way to compute LIS, which costs $O(w^2)$ time and $O(w)$ space where $w$ denotes the length of the sequence. DP computes only the length of LIS and it can figure out an LIS with an auxiliary array. To enumerate all LIS, we can just save all predecessors of each item when determine the maximum length of the increasing subsequence ending with it, which does not result in a higher time complexity while improves the space cost from $O(w)$ to $O(w^2)$.
}


\begin{table*}[t] \RawFloats

\setlength{\belowcaptionskip}{-2pt}
\begin{minipage}[t]{0.66\linewidth}
\def\arraystretch{1.2}
\setlength{\tabcolsep}{0.15pt}
\centering
\small
    \begin{small}
    \resizebox{\linewidth}{!}
    {
  	\begin{tabular}{|l|c|c|c|c|c|c|}
    \hline
    \multirow{2}{*}{\bfseries{Method}}  			   &
    \multirow{2}{*}{\bfseries{ \makecell{LIS \\ Enumeration}}}     &
    \multirow{2}{*}{\bfseries{\makecell{LIS with \\ max Weight}}} &
    \multirow{2}{*}{\bfseries{\makecell{LIS with \\ min Weight}}} &
    \multirow{2}{*}{\bfseries{\makecell{LIS with \\ max Gap}}} &
    \multirow{2}{*}{\bfseries{\makecell{LIS with \\ min Gap}}}  &
    \multirow{2}{*}{\bfseries{\makecell{LIS \\ Length}}}  
    \\

     &  &  &  &  & & 
    \\ \hline

    Our Method & $O(OUTPUT)$ & $O(OUTPUT)$ & $O(OUTPUT)$ & $O(w+OUTPUT)$ & $O(w+OUTPUT)$ & $O(OUTPUT)$ 
    \\ \hline
    LISSET & $O(OUTPUT)$ & -- & --  & -- & -- & $O(OUTPUT)$ 
    \\ \hline
    MHLIS  & -- & -- & --  & -- & $O(w+OUTPUT)$ & $O(OUTPUT)$ 
    \\ \hline
    VARIANT  & -- & $O(OUTPUT)$ & $O(OUTPUT)$  & $O(w+OUTPUT)$ & $O(w+OUTPUT)$ & $O(OUTPUT)$ 
    \\ \hline
    DP  & $O(OUTPUT)$ & -- & --  & -- & -- & $O(OUTPUT)$ 
    \\ \hline
    LISone  & -- & -- & --  & -- & -- & $O(OUTPUT)$ 
    \\ \hline
  	\end{tabular}
    }
    \vspace{0.05in}
    \caption{Theoretical Comparison on Online Query} 
    \label{tab:online:query}
    \end{small}
\end{minipage}%
\hspace{0.1in}
\begin{minipage}[t]{0.33\linewidth}
\def\arraystretch{1.3}
\setlength{\tabcolsep}{0.1pt}
\centering
\small
    \begin{small}
    \resizebox{0.95\linewidth}{!}
    {
    \begin{tabular}{|l|c|c|c|c|}
            \hline \multirow{2}{*}{\textbf{Methods} } &
            \multirow{2}{*}{\makecell{\textbf{Space} \\ \textbf{Complexity}} }  &
            \multicolumn{3}{c|}{\textbf{Time Complexity}} \\ \cline{3-5}
     & & \bfseries{Construction} & \bfseries{Insert} & \bfseries{Delete}  \\ \hline
      Our Method & $O(w)$ & $O(w \log w)$ & $O(\log w)$ & $O(w)$ \\ \hline
      LISSET & $O(w^2)$ & $O(w^2)$ & $O(w)$ & $O(w)$ \\ \hline
      MHLIS & $O(w)$ & $O(w \log w)$ & -- & -- \\ \hline
      VARIANT & $O(w)$ & $O(w \log w)$ & -- & -- \\ \hline
      DP & $O(w^2)$ & $O(w^2)$ & -- & -- \\ 
      \hline
	  LISone & $O(w^2)$ & $O(w^2)$ & $O(w)$ & $O(w)$ \\ 
	  \hline
    \end{tabular}
    }
    \end{small}
    \vspace{0.05in}
    \caption{Theoretical Comparison on Data Structure} 
    \label{tab:index:comp}
\end{minipage}%

\end{table*}


\vspace{-0.15in}
\subsection{Theoretical Analysis}\label{sec:theory}
\Paragraph{Data Structure Comparison.}
We compare the space, construction time and update time of our data structure against those of other works. The comparison results are presented in Table \ref{tab:index:comp}\footnote{The time complexities in Table \ref{tab:index:comp} are based on the worst case analysis. We also studies the time complexity of our method over sorted sequence in Appendix \citeAPPsorted
\submitshow{of the full version of this paper \cite{li2016lis}}
.
} 
. Note that the data structures in the comparative approaches cannot support all LIS-related problems, while our data structure can support both LIS enumeration and constraint-based LIS problems (Table \ref{tab:comparative}) in a uniform manner.

Since  MHLIS, VARIANT or DP does not address data structure maintenance issue, they cannot be used in the streaming model directly. To enable comparison of the three algorithms, we re-construct the data structure in each time window. In this case, the time complexity of the data structure maintenance in MHLIS, VARIANT and DP are the same with their construction time.

We assume that $w$ is the time window length. Table \ref{tab:index:comp} shows that our approach is better or not worse than any comparative work on any metric. Our data structure is better than LISSET on both space and the construction time complexity. Furthermore, the insertion time $O(\log w)$ in our method is also better than the time complexity $O(w)$ in LISSET. As mentioned earlier, none of MHLIS, VARIANT or DP addresses the data structure update issue. Thus, they need $O(w \log w)$ ($O(w^2)$ for DP) time to re-build data structure in each time window.  Obviously, ours is better than theirs.

\Paragraph{Online Query Algorithm Comparison.}
Table \ref{tab:online:query} shows online query time complexities of different approaches. As we know, the online query response time in the data stream model consists of both online query time and the data structure maintenance time. Since the data structure maintenance time has been presented in Table \ref{tab:index:comp}, we only show the online query algorithm' time complexities in Table \ref{tab:online:query}. We can see that, our online query time complexities are the same with the comparative ones. However, the data structure update time complexity in our method is better than others. Therefore, our overall query response time is  better than the comparative ones from the theoretical perspective.

\vspace{-0.1in}
\subsection {Experimental Evaluation}\label{sec:experimenteva}
\vspace{-0.05in}
\label{sec:experiment}
We evaluate our solution against the comparative approaches. All methods, including comparative methods, are implemented by C++ on Eclipse(4.5.0) and all codes are compiled by g++(5.2.0) under default settings. Each comparative method are implemented according the corresponding paper with our best effort. The experiments are conducted in Window 8.1 on a machine with an Intel(R) Core(TM) i7-4790 3.6GHz CPU and a 8G memory. All codes, including those for comparative methods are provided in Github \cite{lisgit}.

\Paragraph{Dataset.}
\label{sec:setup}
We use four datasets  in our experiments: real-world stock data, gene sequence datasets, power usage data and synthetic data. 
 The stock data is about the historical open prices of Microsoft Cooperation in the past two decades\footnote{\url{http://finance.yahoo.com/q/hp?s=MSFT&d=8&e=13&f=2015&g=d&a=2&b=13&c=1986&z=66&y=66}}, up to 7400 days.
 The gene datasets is a sequence of 4,525 matching positions, which are computed over the BLAST output of mRNA sequences\footnote{\url{ftp://ftp.ncbi.nih.gov/refseq/B_taurus/mRNA_Prot/}} against a gene dataset\footnote{\url{ftp://ftp.ncbi.nlm.nih.gov/genbank/}} according to the process in \cite{blastalignment}.
The power usage dataset\footnote{\url{http://www.cs.ucr.edu/~eamonn/discords/power_data.txt}} is a public power demand dataset used in \cite{Keogh2007powerdat}. It measured the power consumption for a Dutch research facility in 1997 which contains 35,040 power usage value.
The synthetic dataset \footnote{\url{https://archive.ics.uci.edu/ml/datasets/Pseudo+Periodic+Synthetic+Time+Series}} is a time series benchmark \cite{synthekeogh1999indexing} that contains one million data points(See \cite{synthekeogh1999indexing} for the details of data generation).
Due to the space limits, we only present the experimental results over stock dataset in this section and the counterparts over the other three datasets are available in Appendix \citeAPPexp
\submitshow{
of the full version of this paper \cite{li2016lis}
}.

\begin{figure}[h!]
\setlength{\belowcaptionskip}{-2pt}
\setlength{\abovecaptionskip}{-2pt}
\newcommand{\mywidth}{0.47\textwidth}
\captionsetup[figure]{font=small,skip=0pt}

    \centering
    \newcommand{\figf}{\stockfolder}
    \begin{subfigure}[t]{\mywidth}
	    \centering
	    \resizebox{\linewidth}{!}
	    {
	        \includegraphics{\figf 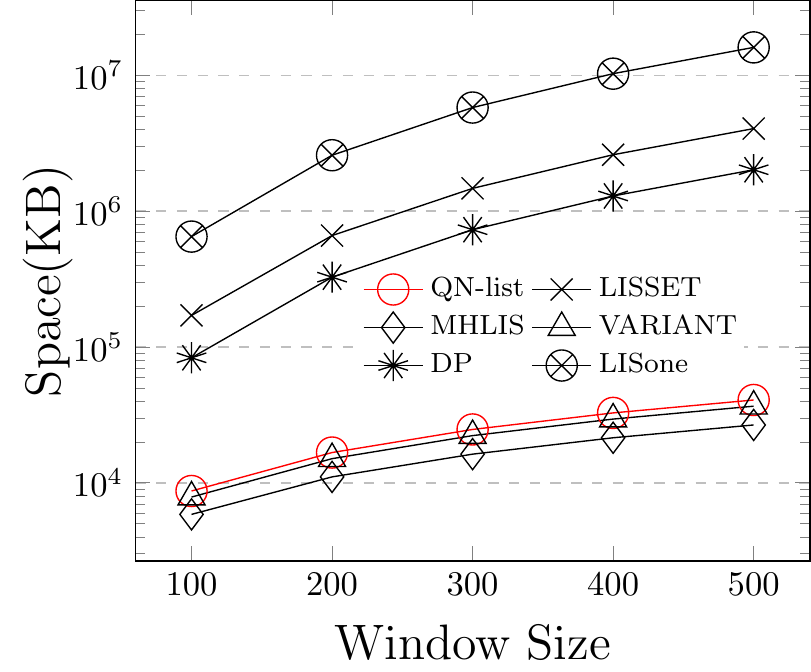}
	    }
	    \caption{Space cost}
	    \label{fig:index:space}
    \end{subfigure}
    \begin{subfigure}[t]{\mywidth}
        \centering
        \resizebox{\linewidth}{!}
        {
            \includegraphics{\figf 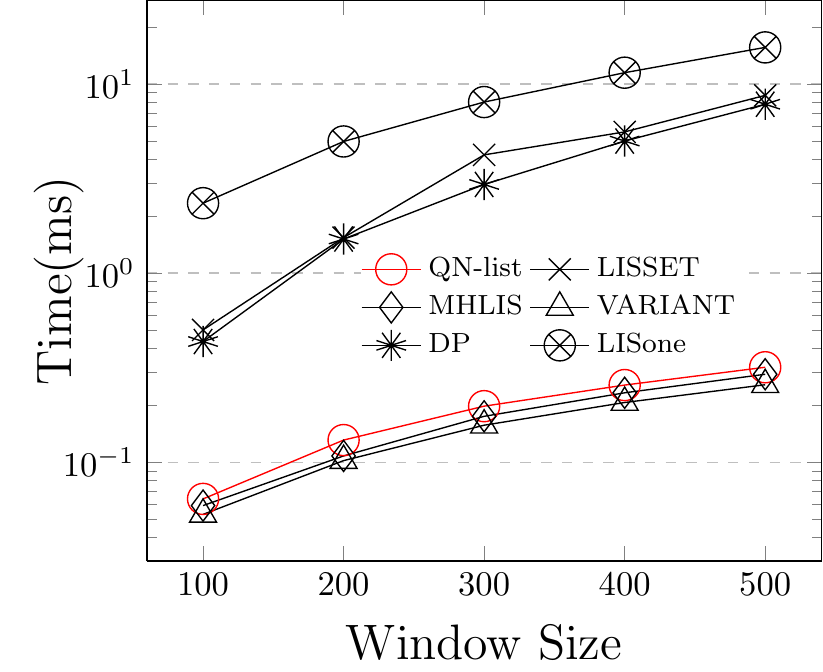}
        }
        \caption{Construction time cost}
        \label{fig:index:construct}
    \end{subfigure}
    \begin{subfigure}[t]{\mywidth}
        \centering
        \resizebox{\linewidth}{!}
        {
            \includegraphics{\figf 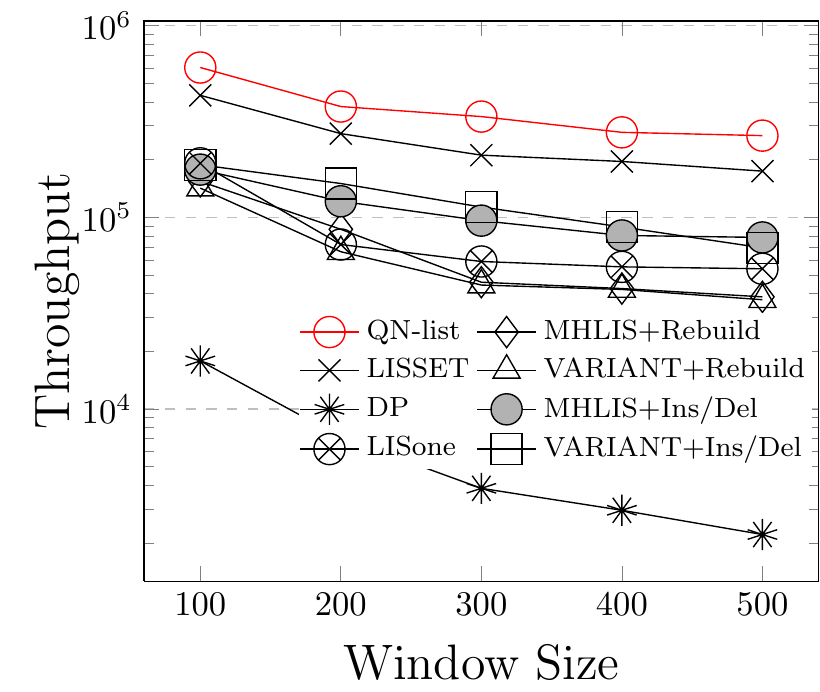}
        }
        \caption{Maintenance efficiency}
        \label{fig:index:update}
    \end{subfigure}
    \begin{subfigure}[t]{\mywidth}
        \centering
        \resizebox{\linewidth}{!}
        {
            \includegraphics{\figf 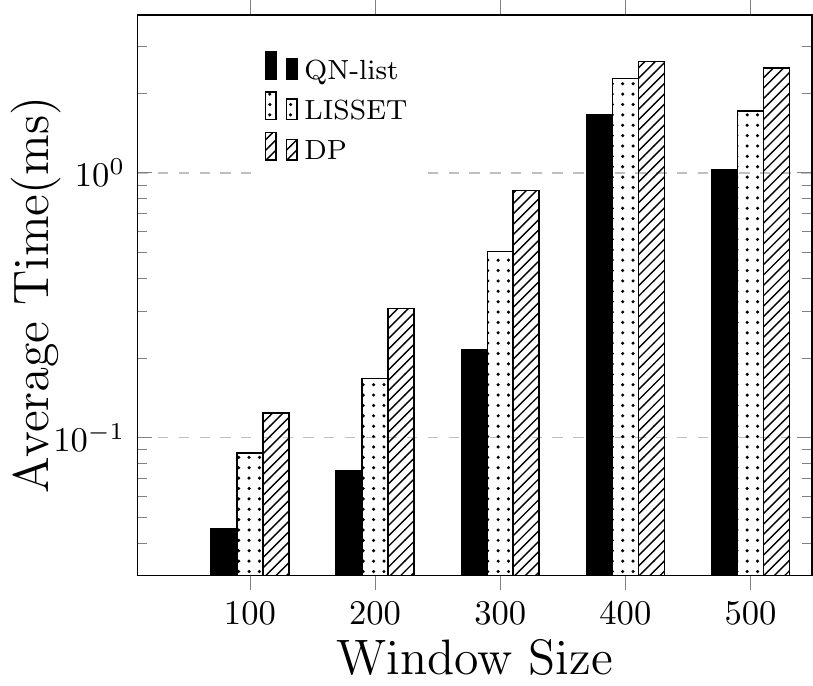}
        }
        \caption{LIS Enumeration}
        \label{fig:enumtime}
    \end{subfigure}
    \begin{subfigure}[t]{0.65\textwidth}
        \centering
        \resizebox{\linewidth}{!}
        {
            \includegraphics{\figf 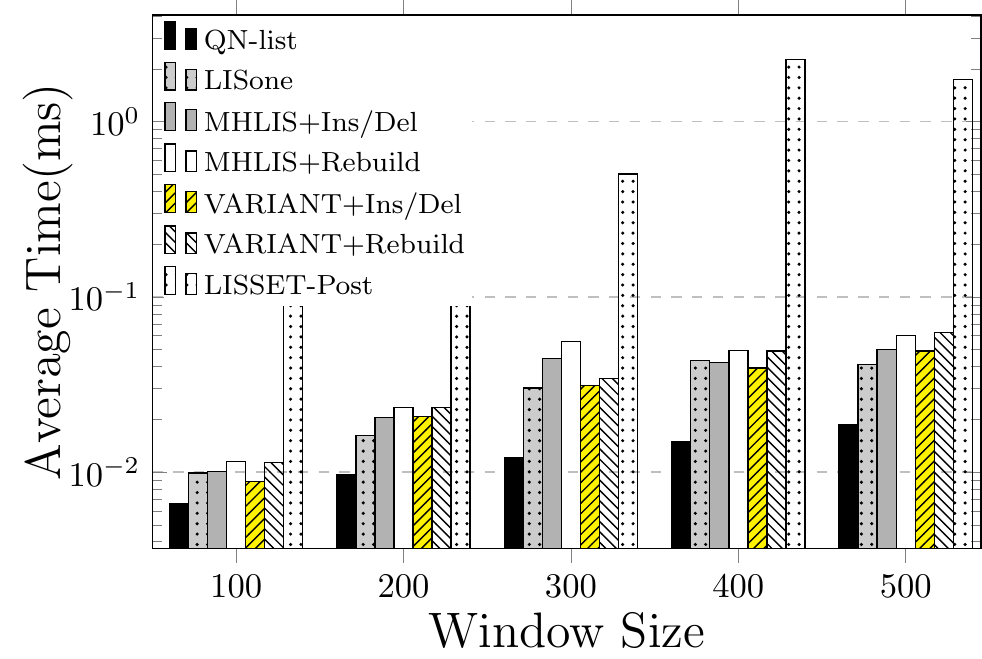}
        }
        \caption{One single LIS}
        \label{fig:lisone}
    \end{subfigure}
    \begin{subfigure}[t]{\mywidth}
	    \centering
	    \resizebox{\linewidth}{!}
	    {
	        \includegraphics{\figf 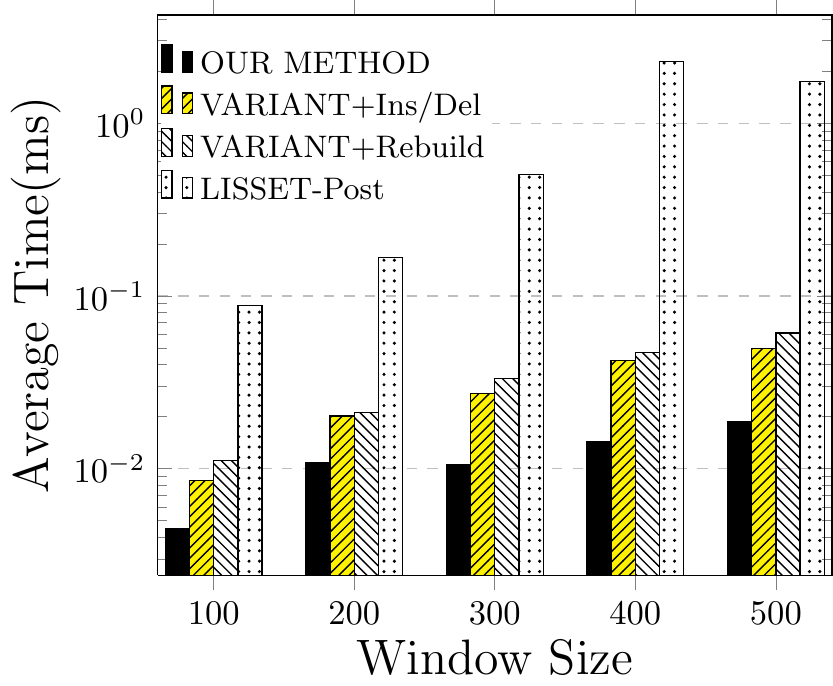}
	    }
	    \vspace{-0.1in}
	    \caption{LIS with maximum weight}
	    \label{fig:maxwtime}
    \end{subfigure}
    \begin{subfigure}[t]{\mywidth}
	    \centering
	    \resizebox{\linewidth}{!}
	    {
	        \includegraphics{\figf 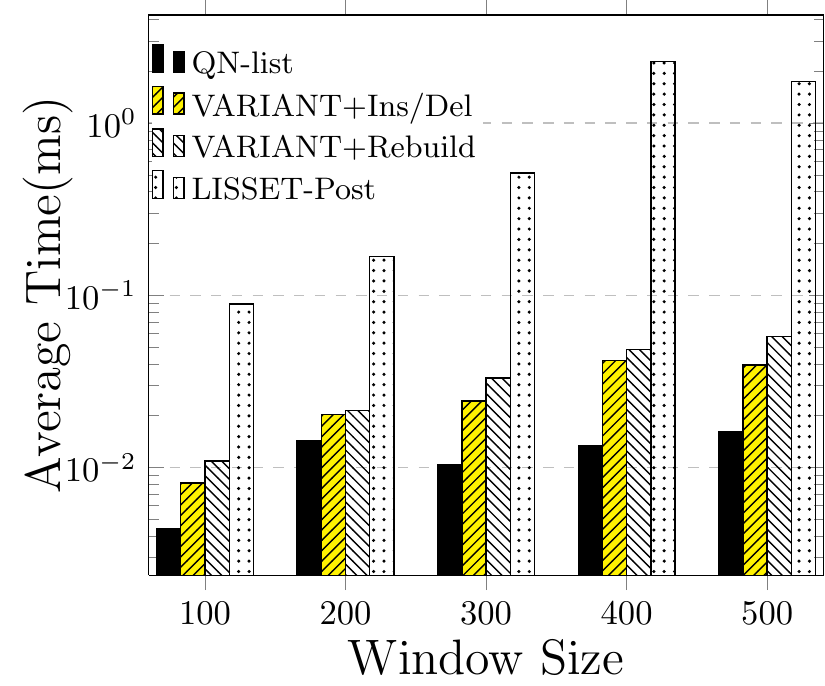}
	    }
	    \vspace{-0.1in}
	    \caption{LIS with minimum weight}
	    \label{fig:minwtime}
    \end{subfigure}
    \begin{subfigure}[t]{\mywidth}
	    \centering
	    \resizebox{\linewidth}{!}
	    {
	        \includegraphics{\figf 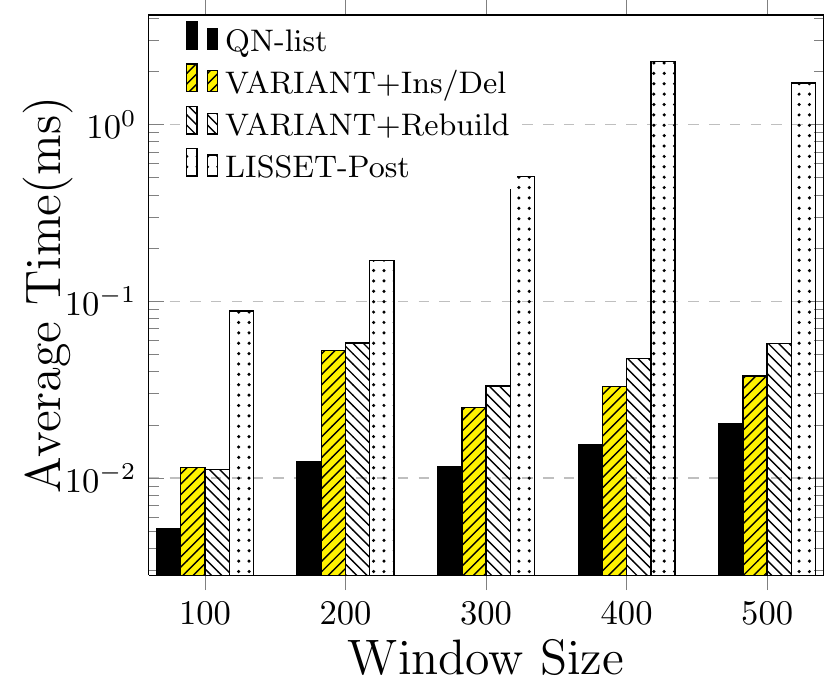}
	    }
	    \caption{LIS with maximum gap}
	    \label{fig:maxhtime}
    \end{subfigure}
    \begin{subfigure}[t]{\mywidth}
	    \centering
	    \resizebox{\linewidth}{!}
	    {
	        \includegraphics{\figf 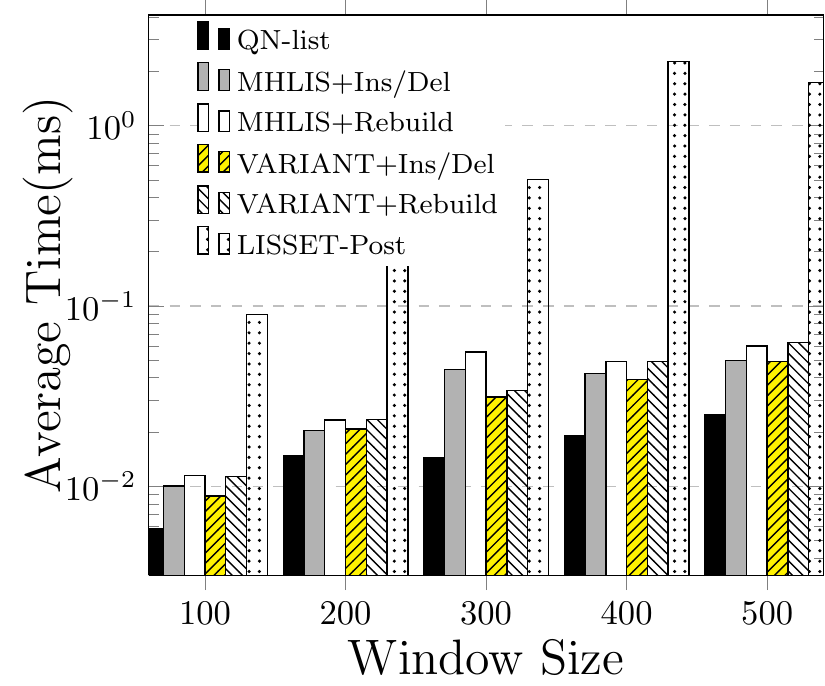}
	    }
	    \caption{LIS with minimum gap}
	    \label{fig:minhtime}
    \end{subfigure}
    \vspace{-0.0in}
\caption{Evaluation on stock data}
    \vspace{-0.1in}
\label{fig:exp:stock}
\end{figure}

\Paragraph{Data Structure Comparison.}
In this experiment, we compare the data structures of different approaches on space cost, construction time and update time.

The \textbf{space cost} of each method is presented in Figure \ref{fig:index:space}. Our method costs much less memory than LISSET, DP and LISone while slightly more than that of MHLIS and VARIANT, which results from the extra cost in our QN-List to support efficient maintenance and computing LIS with constraints. Note that none of the comparative methods can support both LIS enumeration and LIS with constraints; but our QN-List can support all these LIS-related problems in a uniform manner (see Table \ref{tab:comparative}).


We \textbf{construct} each data structure five times and present their average constuction time in Figure \ref{fig:index:construct}. Similarly, our method runs much faster than that of LISSET, DP and LISone, since our construction time is linear but LISSET , DP and LISone have the square time complexity (see Table \ref{tab:index:comp}). Our construction time is slightly slower than VARIANT and faster than MHLIS, since they have the same construction time complexity (Table \ref{tab:index:comp}).

None of MHLIS, VARIANT or DP addresses \textbf{maintenance} issue. To enable comparison, we implement two stream versions of MHLIS and VARIANT. The first is to rebuild the data structure in each time window(MHLIS+Rebuild, VARIANT+Rebuild). The second version is to apply our update idea into MHLIS and VARIANT (MHLIS+Ins/Del, VARIANT+ Ins/Del). The maintenance efficiency is measured by the throughput, i.e., the number of items to be handled in per second without answering any query. Figure \ref{fig:index:update} shows that our method is obviously faster than comparative approaches on data structure update performance.

\Paragraph{LIS Enumeration.}
We compare our method on LIS Enumeration with LISSET and DP, where LISSET is the only previous work that can be used to enumerate LIS under the sliding window model. We report the average query response time in Figure \ref{fig:enumtime}. In the context of data stream, the overall query response time includes two parts, i.e., the data structure update time and online query time. Our method is faster than both LISSET and DP, and with the increasing of time window size, the performance advantage is more obvious.

\Paragraph{LIS with Max/Min Weight.}
We compare our method with VARIANT on LIS with maximum/minimum weight. VARIANT \cite{variant2009} is the only previous work on LIS with maximum/minimum weight. Figures \ref{fig:minwtime} and \ref{fig:maxwtime} confirms the superiority of our method with regard to VARIANT(VARIANT+Rebuild and VARIANT+Ins/Del).

\Paragraph{LIS with Max/Min Gap.}
There are two previous proposals studying LIS with maximum/minimum gaps. VARIANT \cite{variant2009} computes the LIS with maximum and minimum gap while MHLIS \cite{minheight2009} only computes LIS with the minimum gap. The average running time in each window of different methods are  in Figures \ref{fig:maxhtime} and \ref{fig:minhtime}.
We can see that our method outperforms other methods significantly.

\Paragraph{LIS length (Output a single LIS).}
We compare our method with LISone \cite{liswAlbert2004} on outputting an LIS (The length comes out directly). Since other comparative methods can easily support outputting an LIS, we also add other comparative works into comparison. Figure \ref{fig:lisone} shows that our method is much more efficient than comparative methods on computing LIS length and output a single LIS.

\vspace{-0.10in}
\section {Conclusions}
\label{sec:con}
In this paper, we propose a uniform data structure to support enumerating
all LIS and LIS with specific constraints over sequential data stream.
The data structure built by our algorithm only takes linear space and can be updated only in linear time, which make our approach practical in handling high-speed sequential data streams. To the best of our knowledge, our work is the first to proposes a uniform solution (the same data structure and computing framework)
to address all LIS-related issues in the data stream scenario. Our method outperforms the state-of-the-art work not only theoretically, but also empirically in both time and space cost.

\vspace{-0.10in}

\section*{Acknowledgment}
\vspace{-0.05in}
This work was supported by The National Key Research and Development Program of China under grant 2016YFB1000603 and NSFC under grant 61622201, 61532010 and 61370055. Lei Zou is the corresponding author of this paper. 





\small
\bibliographystyle{abbrv}
\bibliography{Main}  

\clearpage
\appendix


\section{More Experiments}
\label{sec:power}
Experimental results on the gene data, power usage data and synthetic dataset are presented in Figure \ref{fig:exp:gene}--\ref{fig:exp:synthetic}.
\arxivshow{(This is a full version\cite{li2016lis}).}

\vspace{-0.05in}
\begin{figure}[t!]
    \centering
    \newcommand{\figf}{\genefolder}
    \begin{subfigure}[t]{0.45\textwidth}
        \centering
        \resizebox{\linewidth}{!}
        {
            \includegraphics{\figf 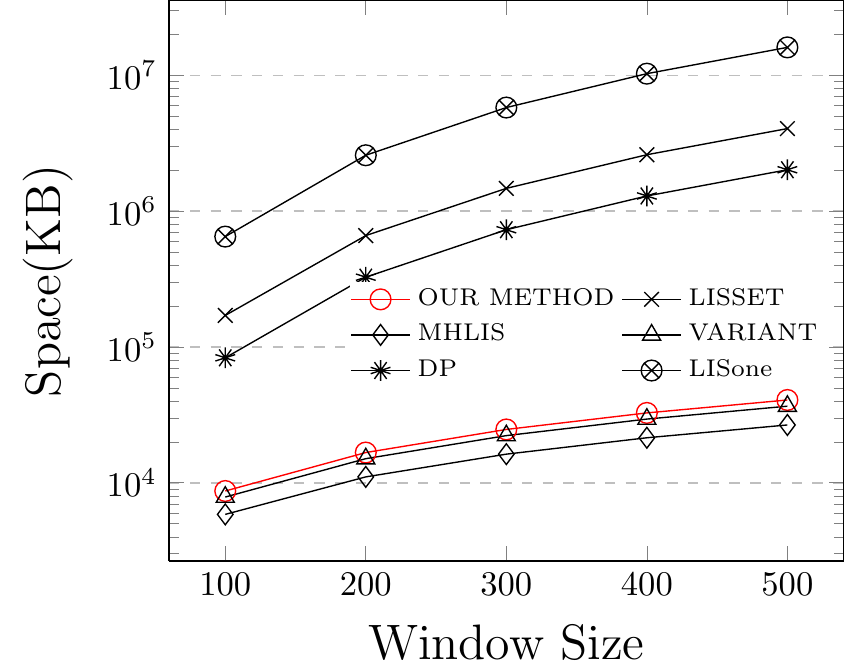}
        }
        \caption{Space cost}
        \label{fig:index:spacegene}
    \end{subfigure}
      	\hspace{0.1in} 
    \begin{subfigure}[t]{0.45\textwidth}
        \centering
        \resizebox{\linewidth}{!}
        {
            \includegraphics{\figf 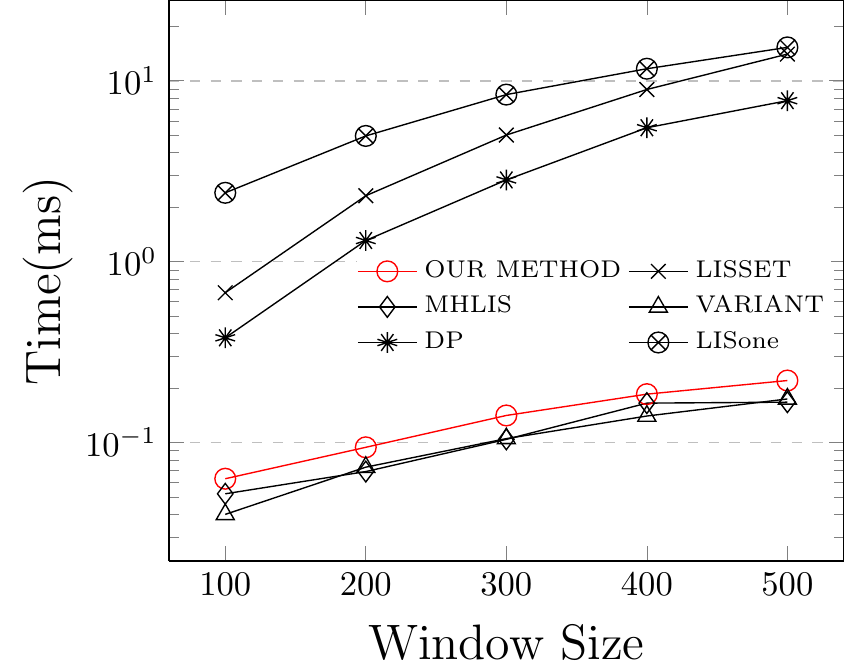}
        }
        \caption{Construction time}
        \label{fig:index:constructgene}
    \end{subfigure}
    \begin{subfigure}[t]{0.45\textwidth}
        \centering
        \resizebox{\linewidth}{!}
        {
            \includegraphics{\figf 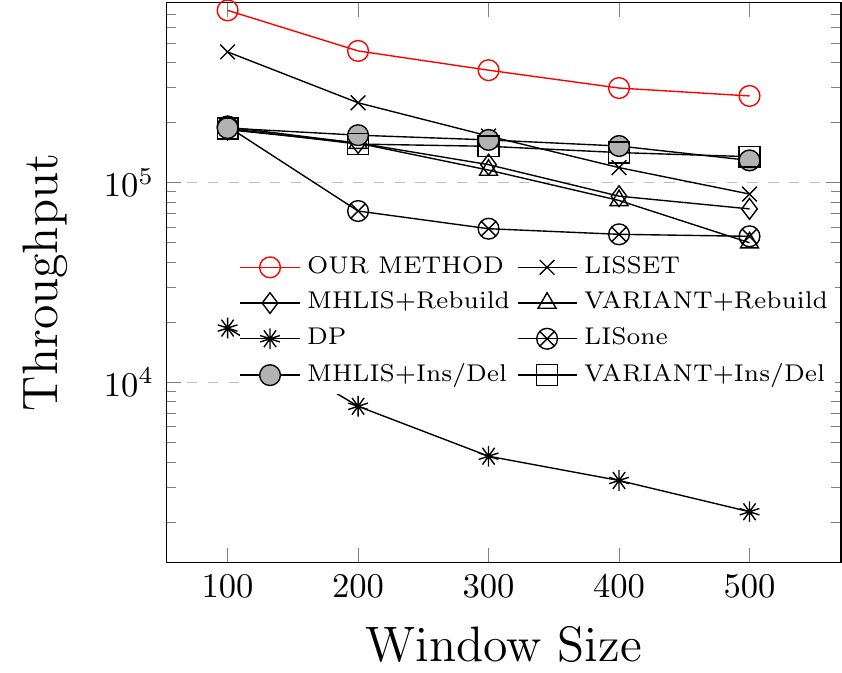}
        }
        \caption{Maintenance }
        \label{fig:index:updategene}
    \end{subfigure}
   	\hspace{0.1in} 
    \begin{subfigure}[t]{0.45\textwidth}
        \centering
        \resizebox{\linewidth}{!}
        {
            \includegraphics{\figf 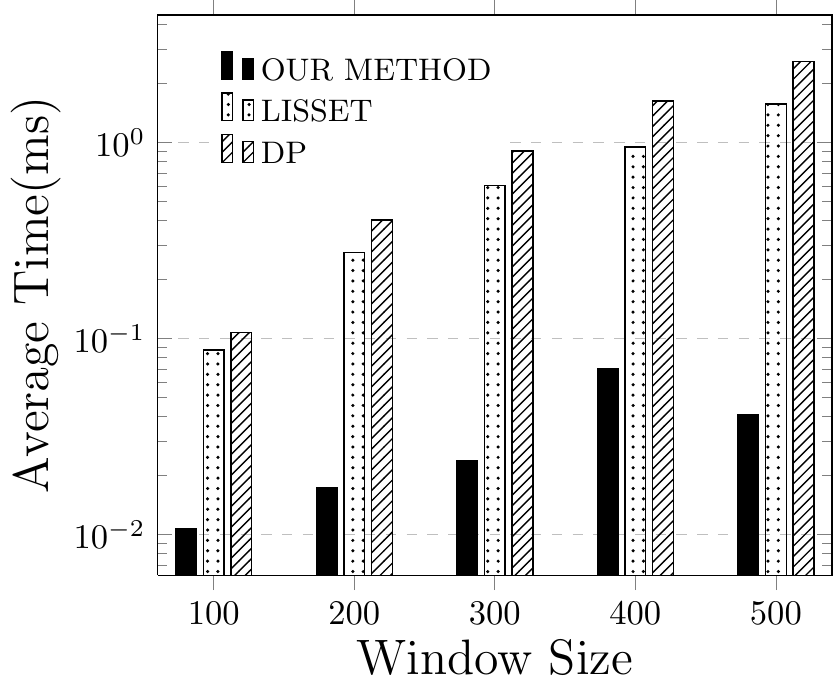}
        }
        \caption{LIS Enumeration}
        \vspace{-0.1in}
        \label{fig:enumtimegene}
    \end{subfigure}
    \begin{subfigure}[t]{0.65\textwidth}
        \centering
        \resizebox{\linewidth}{!}
        {
            \includegraphics{\figf 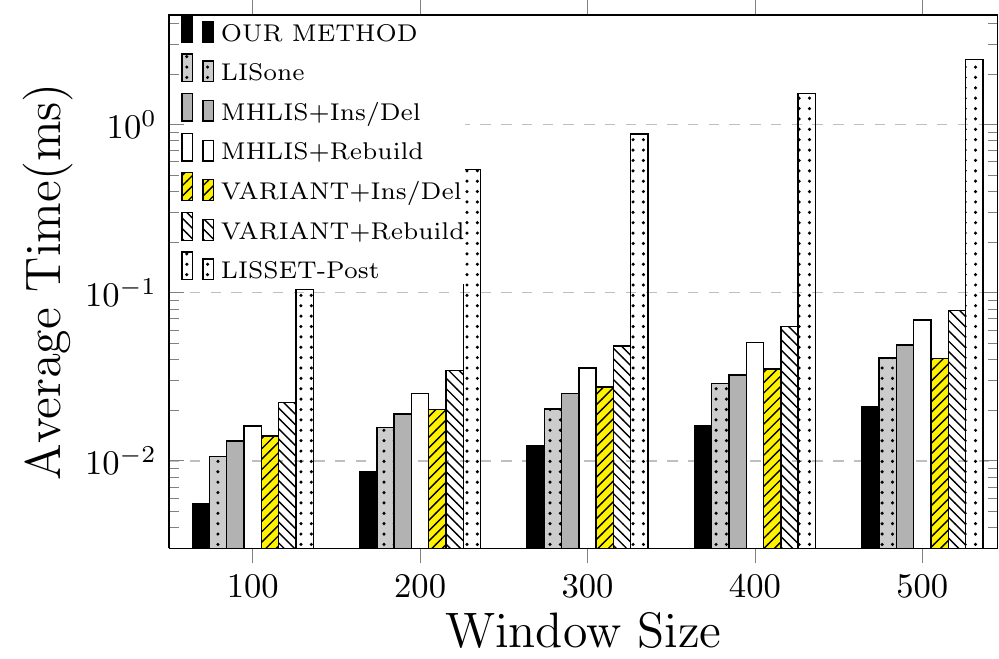}
        }
        \caption{One single LIS}
        \label{fig:lisone}
    \end{subfigure}
    \hspace{0.5in}
    \begin{subfigure}[t]{0.45\textwidth}
		\centering
		\resizebox{\linewidth}{!}
		{
			\includegraphics{\figf 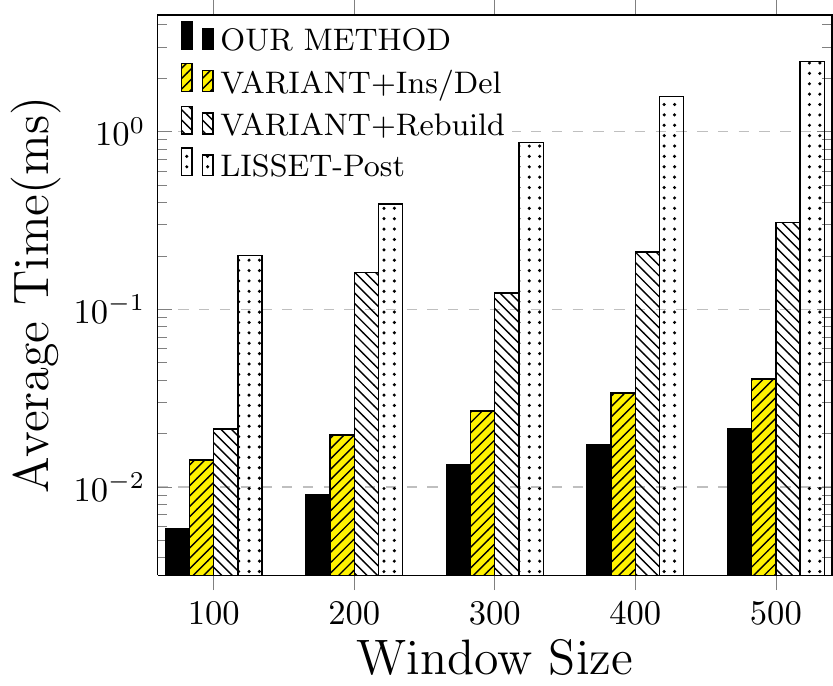}
		}
		\vspace{-0.1in}
		\caption{LIS with maximum weight}
		\label{fig:maxwtimegene}
    \end{subfigure}
	\hspace{0.1in} 
    \begin{subfigure}[t]{0.45\textwidth}
	    \centering
	    \centering
	    \resizebox{\linewidth}{!}
	    {
	        \includegraphics{\figf 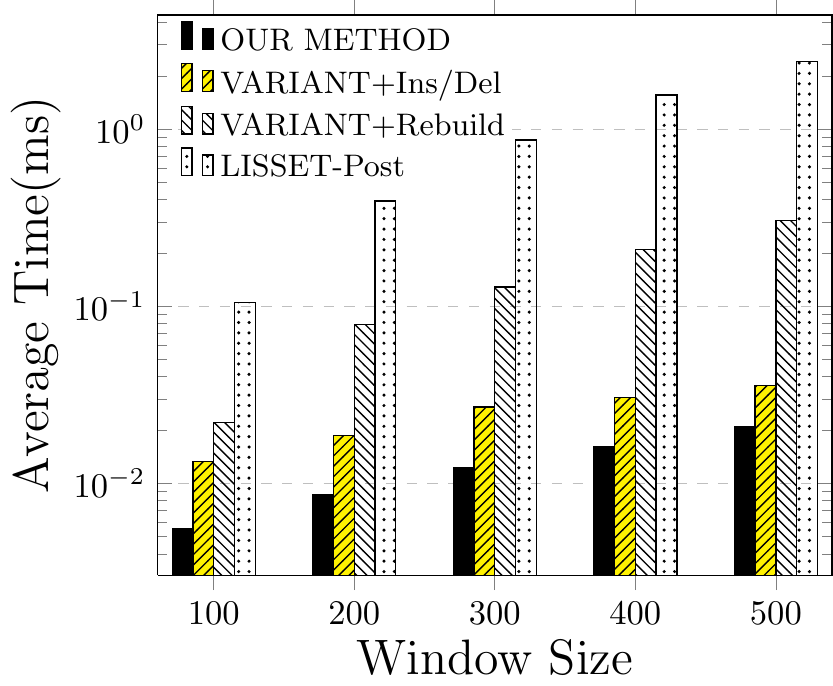}
	    }
	    \vspace{-0.1in}
	    \caption{LIS with minimum weight}
	    \label{fig:minwtimegene}
    \end{subfigure}
    \begin{subfigure}[t]{0.45\textwidth}
	    \centering
	    \resizebox{\linewidth}{!}
	    {
	        \includegraphics{\figf 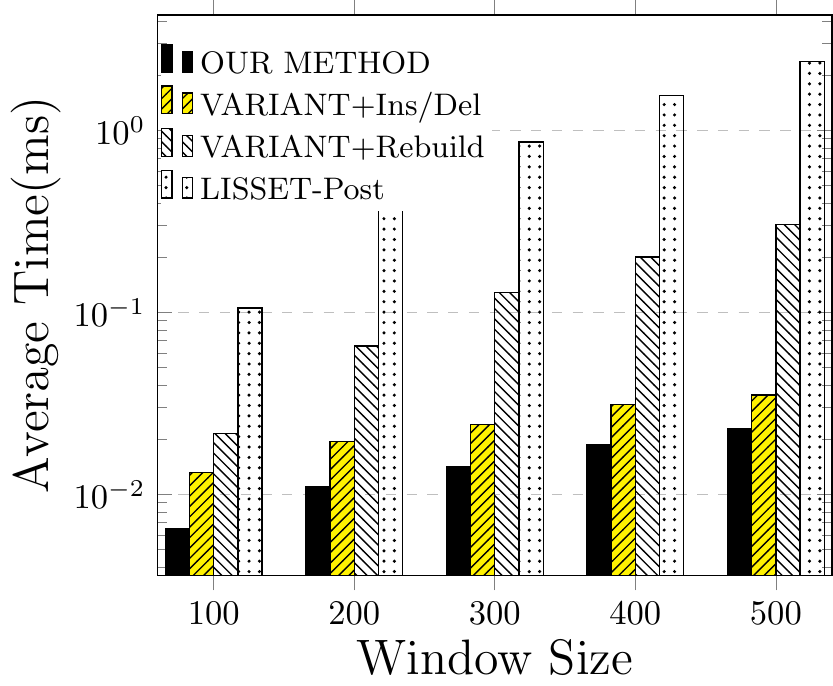}
	    }
	    \caption{LIS with maximum gap}
	    \label{fig:maxhtimegene}
    \end{subfigure}
	\hspace{0.1in}
    \begin{subfigure}[t]{0.45\textwidth}
	    \centering
	    \resizebox{\linewidth}{!}
	    {
	        \includegraphics{\figf 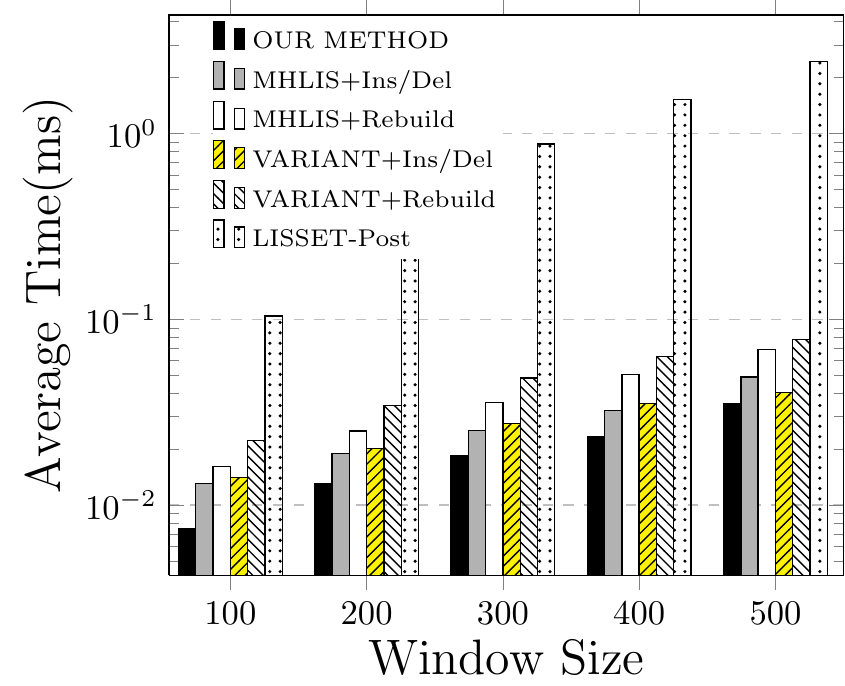}
	    }
	    \caption{LIS with minimum gap}
	    \label{fig:minhtimegene}
    \end{subfigure}
    \vspace{-0.0in}
\caption{Evaluation on gene data}
    \vspace{-0.1in}
\label{fig:exp:gene}
\end{figure}  

\vspace{-0.05in}
\begin{figure}[t!]
    \centering
    \newcommand{\figf}{\powerfolder}
    \begin{subfigure}[t]{0.45\textwidth}
        \centering
        \resizebox{\linewidth}{!}
        {
            \includegraphics{\figf 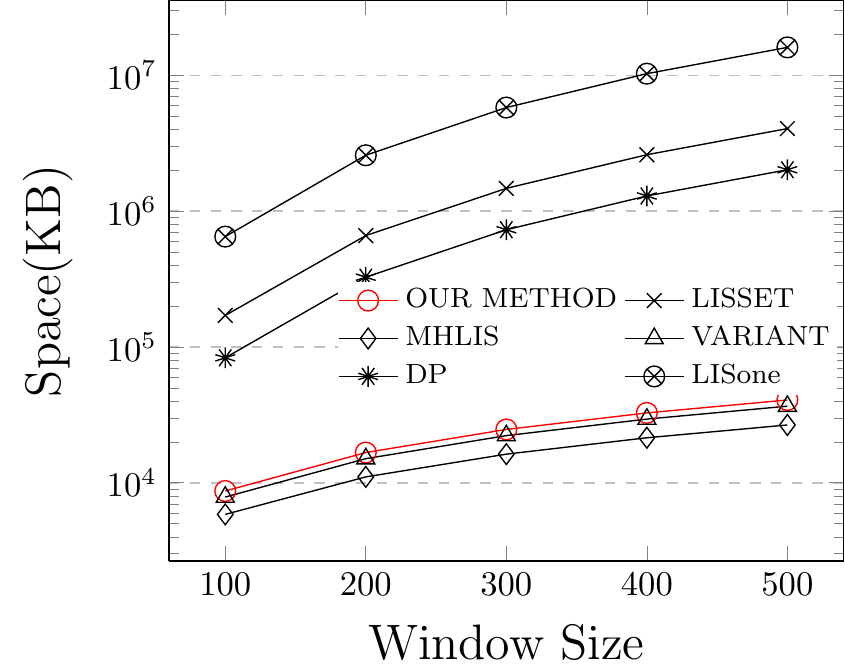}
        }
        \caption{Space cost}
        \vspace{-0.1in}
        \label{fig:index:spacepow}
    \end{subfigure}
      	\hspace{0.1in} 
    \begin{subfigure}[t]{0.45\textwidth}
        \centering
        \resizebox{\linewidth}{!}
        {
            \includegraphics{\figf 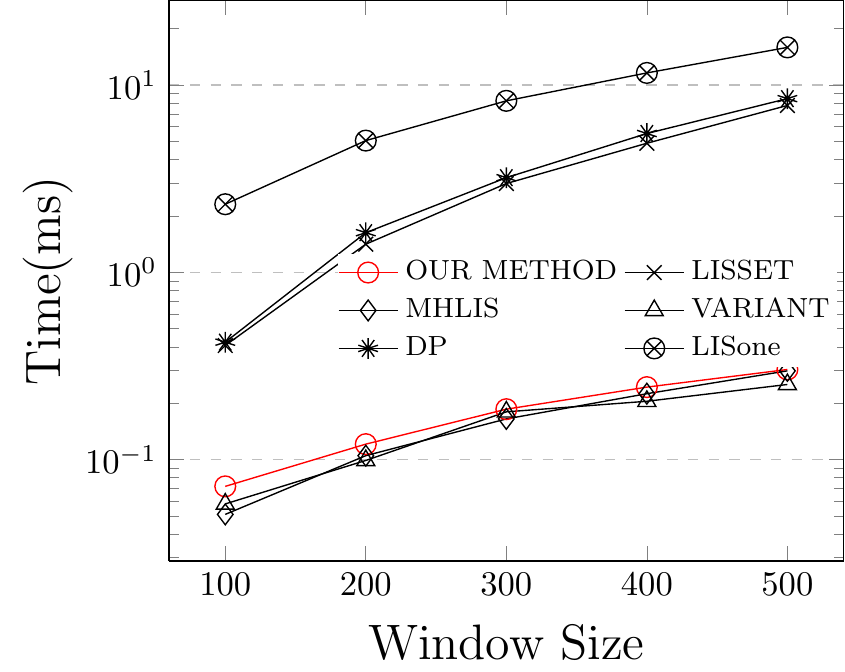}
        }
        \caption{Construction time}
        \label{fig:index:constructpow}
    \end{subfigure}
    \begin{subfigure}[t]{0.45\textwidth}
        \centering
        \resizebox{\linewidth}{!}
        {
            \includegraphics{\figf 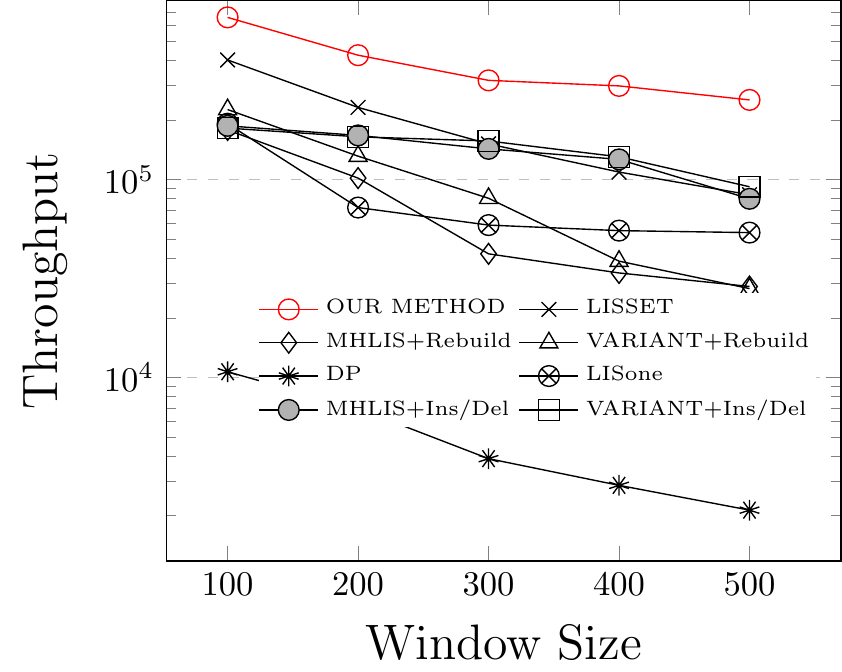}
        }
        \caption{Maintenance }
        \label{fig:index:updatepow}
    \end{subfigure}
   	\hspace{0.1in} 
    \begin{subfigure}[t]{0.45\textwidth}
        \centering
        \resizebox{\linewidth}{!}
        {
            \includegraphics{\figf 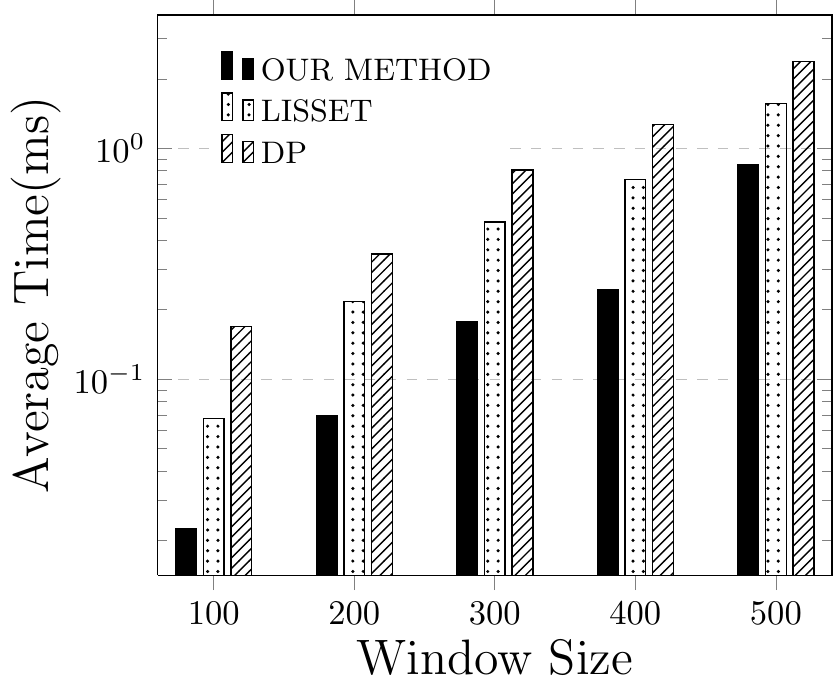}
        }
        \caption{LIS Enumeration}
        \vspace{-0.1in}
        \label{fig:enumtimepow}
    \end{subfigure}
    \begin{subfigure}[t]{0.65\textwidth}
        \centering
        \resizebox{\linewidth}{!}
        {
            \includegraphics{\figf 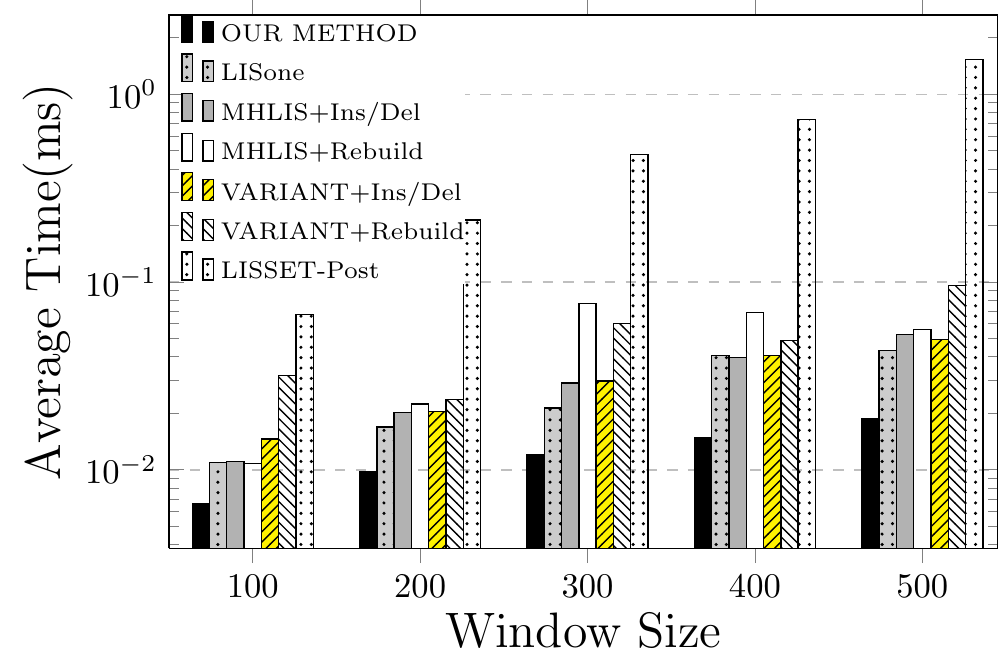}
        }
        \caption{One single LIS}
        \label{fig:lisone}
    \end{subfigure}
    \hspace{0.5in}
    \begin{subfigure}[t]{0.45\textwidth}
		\centering
		\resizebox{\linewidth}{!}
		{
			\includegraphics{\figf 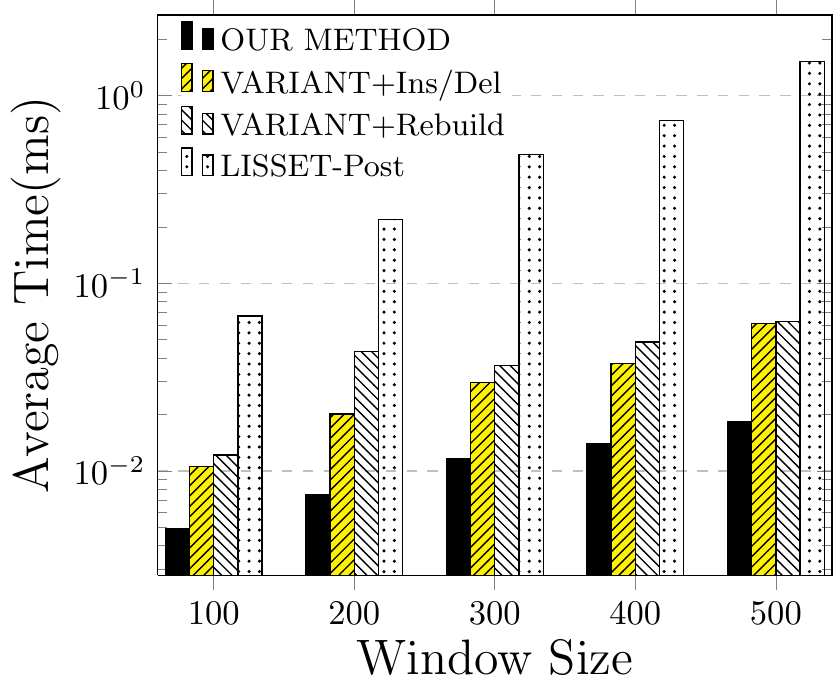}
		}
		\caption{LIS with maximum weight}
		\label{fig:maxwtimepow}
    \end{subfigure}
	\hspace{0.1in} 
    \begin{subfigure}[t]{0.45\textwidth}
	    \centering
	    \centering
	    \resizebox{\linewidth}{!}
	    {
	        \includegraphics{\figf 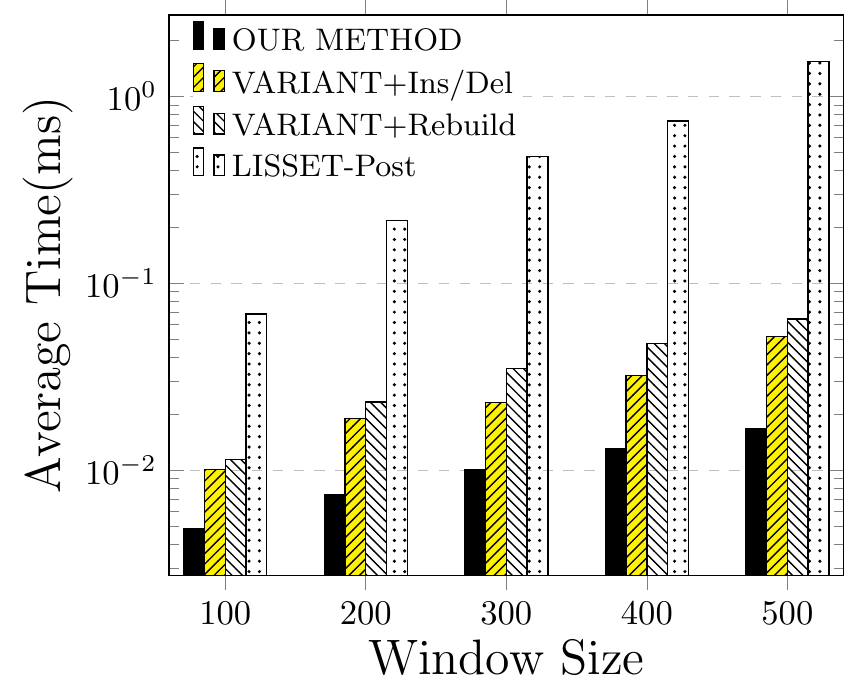}
	    }
	    \caption{LIS with minimum weight}
	    \label{fig:minwtimepow}
    \end{subfigure}
    \begin{subfigure}[t]{0.45\textwidth}
	    \centering
	    \resizebox{\linewidth}{!}
	    {
	        \includegraphics{\figf 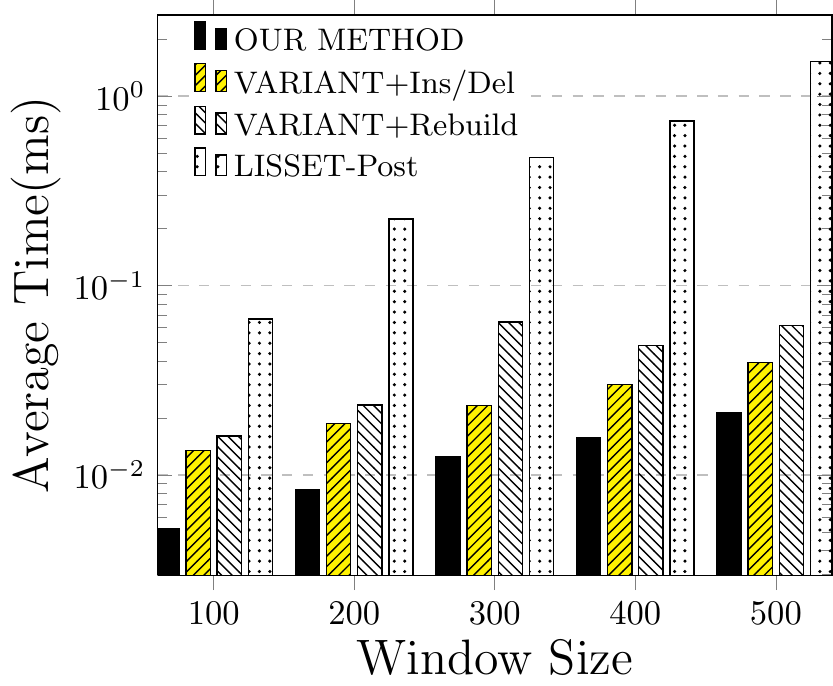}
	    }
	    \caption{LIS with maximum gap}
	    \label{fig:maxhtimepow}
    \end{subfigure}
	\hspace{0.1in}
    \begin{subfigure}[t]{0.45\textwidth}
	    \centering
	    \resizebox{\linewidth}{!}
	    {
	        \includegraphics{\figf 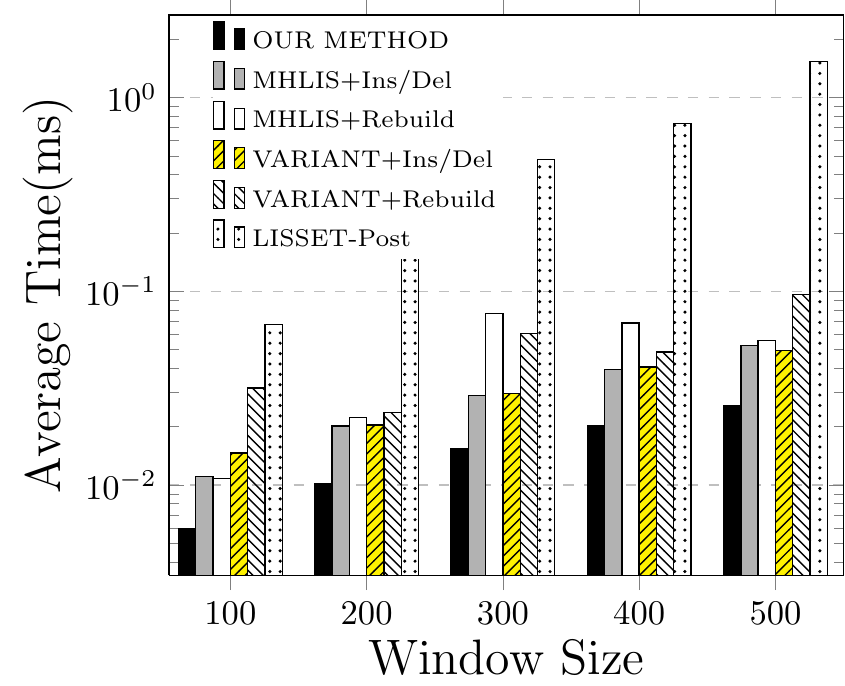}
	    }
	    \caption{LIS with minimum gap}
	    \label{fig:minhtimepow}
    \end{subfigure}
    \vspace{-0.0in}
\caption{Evaluation on power usage data}
    \vspace{-0.1in}
\label{fig:exp:power}
\end{figure}  

\begin{figure}[h!]
    \centering
    \newcommand{\figf}{\syntheticfolder}
    \begin{subfigure}[t]{0.45\textwidth}
        \centering
        \resizebox{\linewidth}{!}
        {
            \includegraphics{\figf 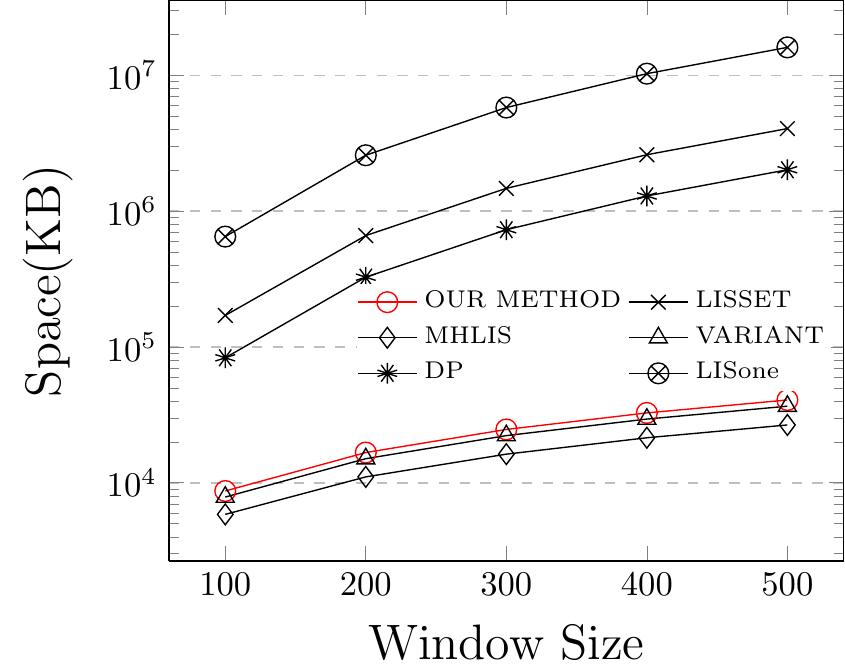}
        }
        \vspace{-0.1in}
        \caption{Space cost}
        \vspace{-0.1in}
        \label{fig:index:spacesyn}
    \end{subfigure}
      	\hspace{0.1in} 
    \begin{subfigure}[t]{0.45\textwidth}
        \centering
        \resizebox{\linewidth}{!}
        {
            \includegraphics{\figf 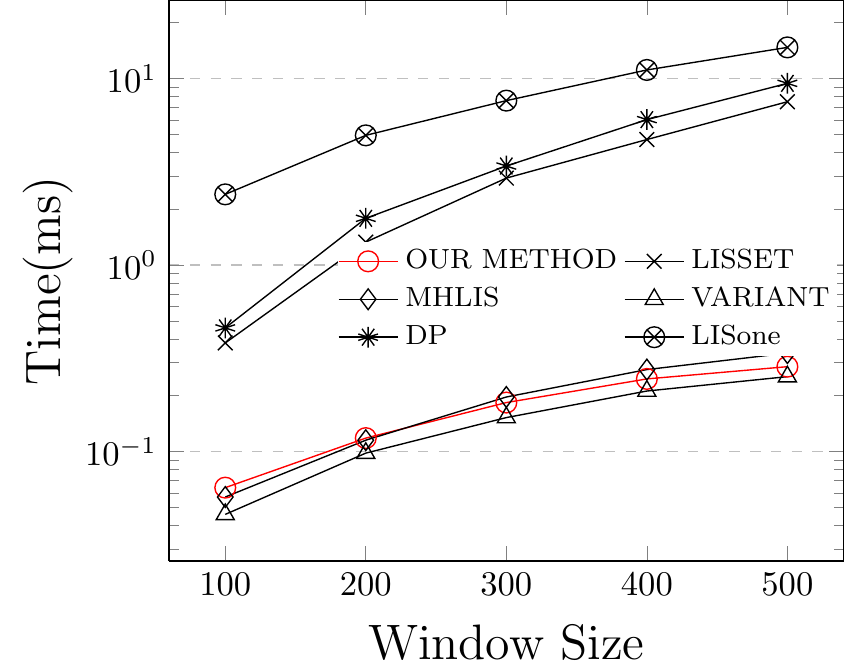}
        }
        \vspace{-0.1in}
        \caption{Construction time}
        \label{fig:index:constructsyn}
    \end{subfigure}
    \begin{subfigure}[t]{0.45\textwidth}
        \centering
        \resizebox{\linewidth}{!}
        {
            \includegraphics{\figf 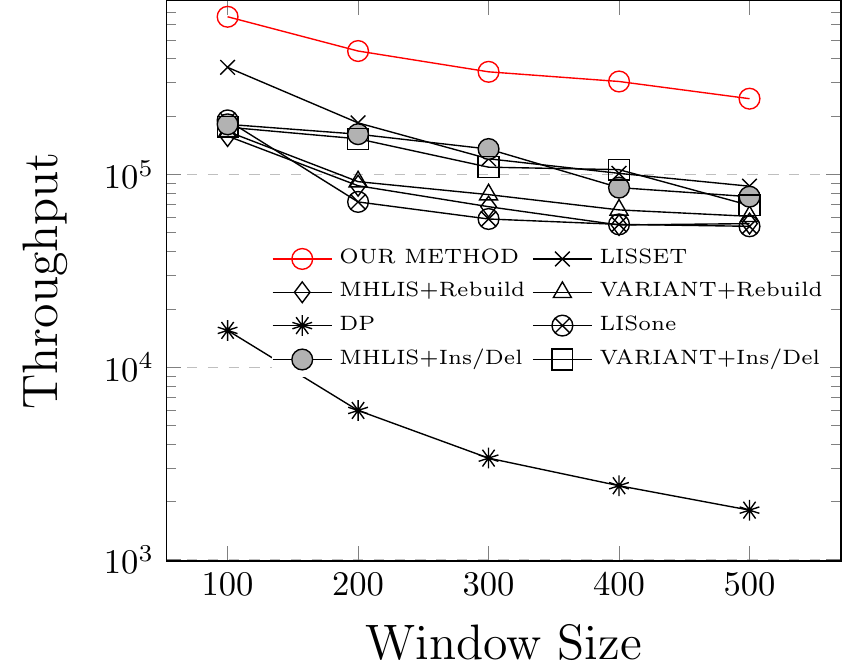}
        }
        \vspace{-0.1in}
        \caption{Maintenance }
        \label{fig:index:updatesyn}
    \end{subfigure}
   	\hspace{0.1in} 
    \begin{subfigure}[t]{0.45\textwidth}
        \centering
        \resizebox{\linewidth}{!}
        {
            \includegraphics{\figf 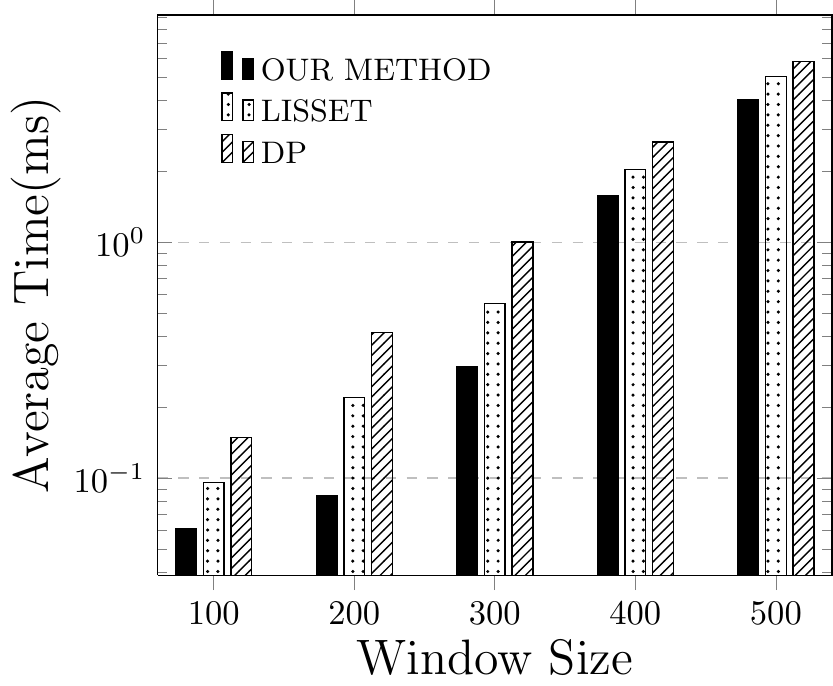}
        }
        \vspace{-0.1in}
        \caption{LIS Enumeration}
        \vspace{-0.1in}
        \label{fig:enumtimesyn}
    \end{subfigure}
    \begin{subfigure}[t]{0.65\textwidth}
        \centering
        \resizebox{\linewidth}{!}
        {
            \includegraphics{\figf 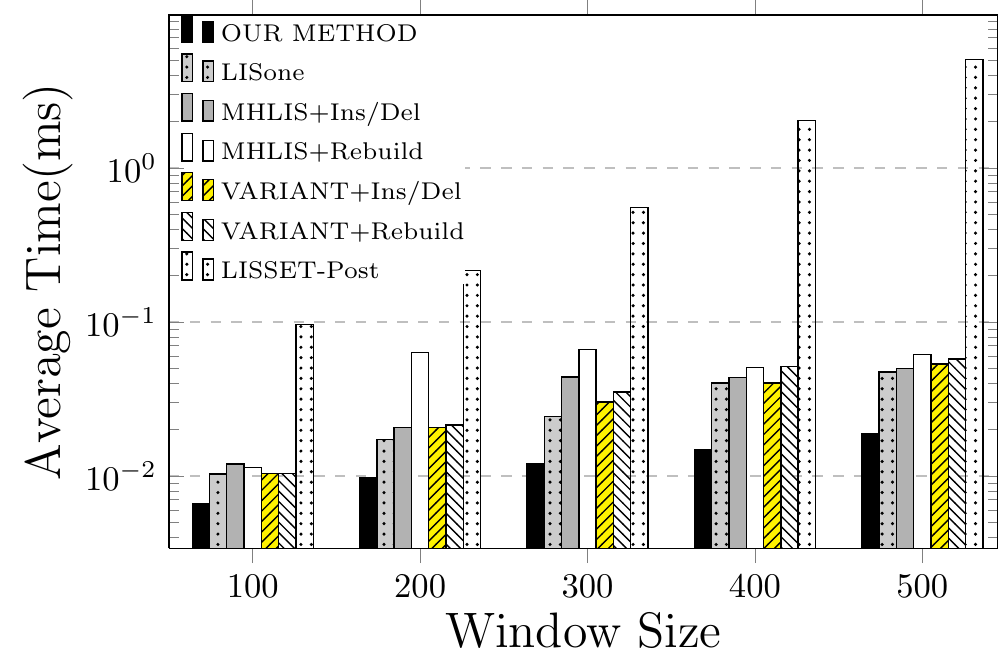}
        }
        \caption{One single LIS}
        \label{fig:lisone}
    \end{subfigure}
    \hspace{0.5in}
    \begin{subfigure}[t]{0.45\textwidth}
		\centering
		\resizebox{\linewidth}{!}
		{
			\includegraphics{\figf 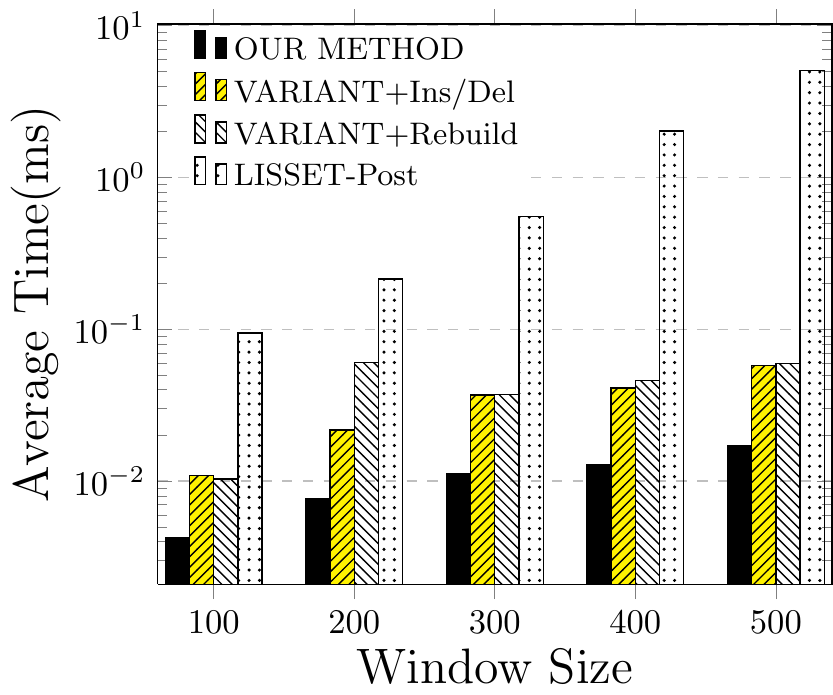}
		}
		\vspace{-0.1in}
		\caption{LIS with maximum weight}
		\label{fig:maxwtimesyn}
    \end{subfigure}
	\hspace{0.1in} 
    \begin{subfigure}[t]{0.45\textwidth}
	    \centering
	    \centering
	    \resizebox{\linewidth}{!}
	    {
	        \includegraphics{\figf 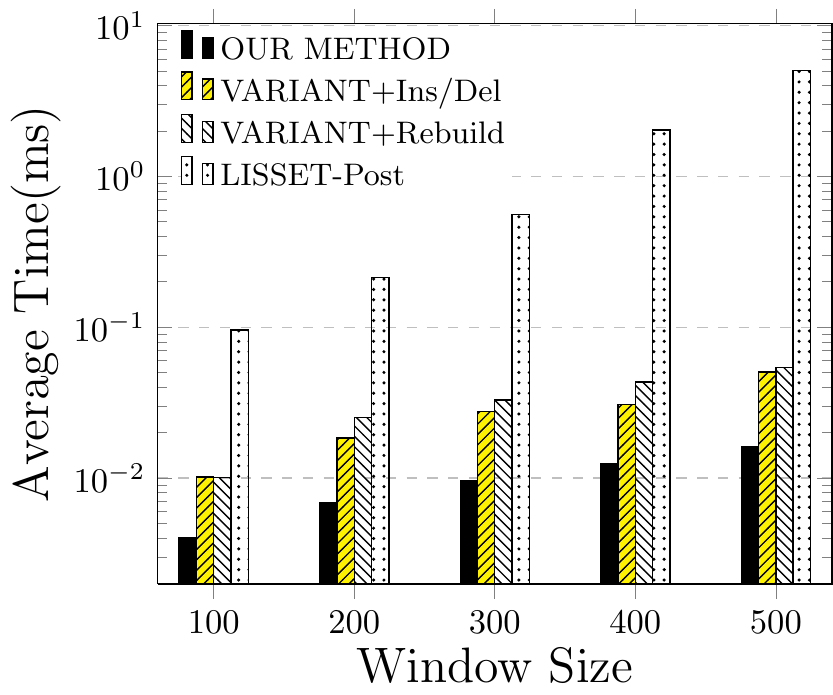}
	    }
	    \vspace{-0.1in}
	    \caption{LIS with minimum weight}
	    \label{fig:minwtimesyn}
    \end{subfigure}
    \begin{subfigure}[t]{0.45\textwidth}
	    \centering
	    \resizebox{\linewidth}{!}
	    {
	        \includegraphics{\figf 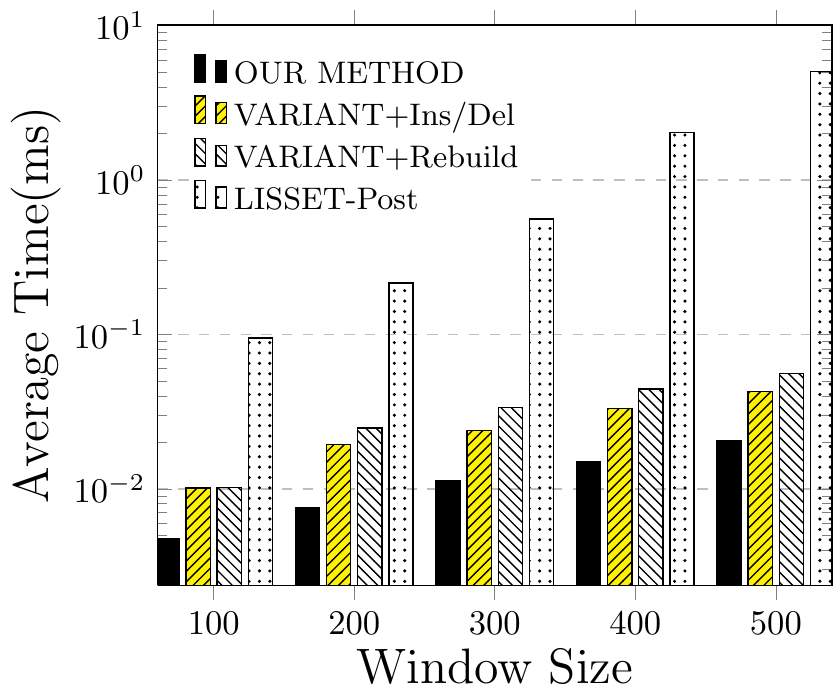}
	    }
	    \caption{LIS with maximum gap}
	    \label{fig:maxhtimesyn}
    \end{subfigure}
	\hspace{0.1in}
    \begin{subfigure}[t]{0.45\textwidth}
	    \centering
	    \resizebox{\linewidth}{!}
	    {
	        \includegraphics{\figf 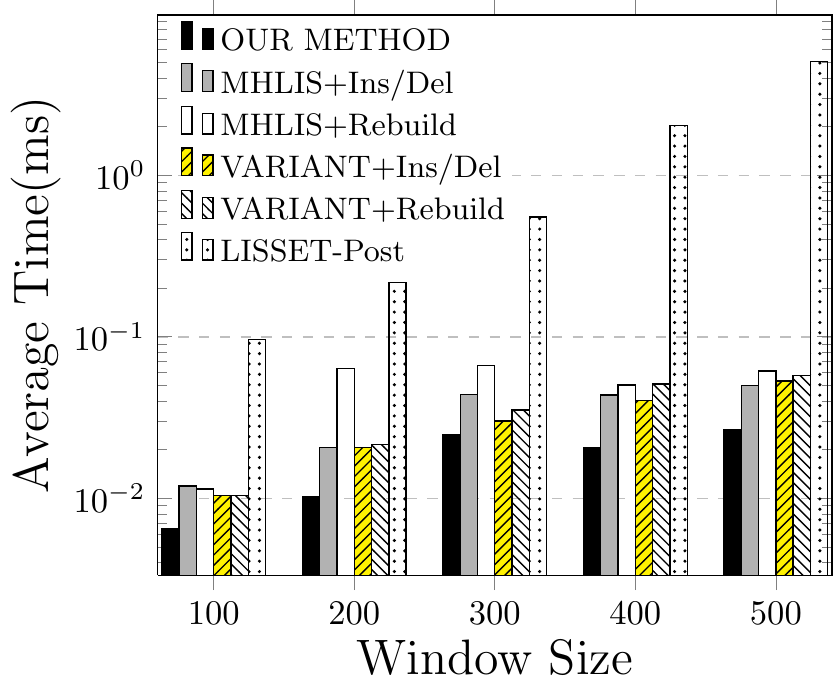}
	    }
	    \caption{LIS with minimum gap}
	    \label{fig:minhtimesyn}
    \end{subfigure}
    \vspace{-0.0in}
\caption{Evaluation on synthetic data}
    \vspace{-0.1in}
\label{fig:exp:synthetic}
\end{figure}  

\section{Proofs of Lemmas \& Theorems}
\label{sec:appendix:lemmaproofs}

\newcommand{\ilemma}[1]{L\scriptsize EMMA\normalsize\ \ref{#1}}

\newcommand{\mylemma}[3]{
\textbf{Proof of Lemma \ref{#1}}


\vspace{-0.06in}
\begin{proof} 
#3
\end{proof}
\vspace{-0.05in}
}

\nop{
\mylemma 
{lem:pred}
{ 
Given a sequence $\alpha$ and $a_i $ $\in \alpha$ where $RL_{\alpha}(a_i)$ $= k$, then $\forall$ $s = $ \{$a_{i_1}$, $a_{i_2}$,...,$a_{i_k}$\} where $s$ is an MIS of $a_i$ ($a_{i_k}$), then $RL_{\alpha}(a_{i_{k-1}})$ $= (k-1)$ and $a_{i_{k-1}}$ $\seqpar{\alpha} a_i$. Recursively, we have
\[RL_{\alpha}(a_{i_{k'}}) = k' \land a_{i_k'} \seqpar{\alpha} a_{i_{k'+1}} \ (1 \leq k' < k)\]
}
{ 
\vspace{-0.05in}
$a_{i_{k-1}} \seqpar{\alpha}$ $a_{i_k}$ since $s$ is an increasing subsequence, and hence, $RL_{\alpha}(a_{i_{k-1}})$ $< RL_{\alpha}(a_{i_k})$ $= k$. 
Besides the first $k-1$ items form an increasing subsequence ending with $a_{i_{k-1}}$ then $RL_{\alpha}(a_{i_{k-1}})$ $\geq k-1$. 
Thus, $RL_{\alpha}(a_{i_{k-1}})$ $= k-1$. Recursively, $RL_{\alpha}(a_{i_{k'}})$ $ = k'$ $\land$ $a_{i_{k'}} \seqpar{\alpha} a_{i_{k'+1}}$ $(1 \leq k' < k)$.
}
} 

\mylemma 
{lem:consecutive}
{ 
 Let $\alpha=$ $\{a_1$ $,a_2,...,a_w\}$ be a sequence. Consider two items $a_i$ and $a_j$ in a horizontal list $\mathbb{L}_{\alpha}^t$ (see Definition \ref{def:hlist}).  

\begin{enumerate}
\item 
	If $t = 1$ then $a_i$ has no predecessor and if $t > 1$ then $a_i$ has at least one predecessor and all predecessors of $a_i$ are located in $\mathbb{L}_{\alpha}^{t-1}$.
\item
	If $rn_{\alpha}(a_j) = a_i$, then $i > j$ and $a_i < a_j$. If $ln_{\alpha}(a_j) = a_i$, then $ i < j$ and $a_i > a_j$. Items in a horizontal list $\mathbb{L}_{\alpha}^t$ ($t=1,\cdots,m$) are monotonically decreasing while their subscripts (i.e., their original position in $\alpha$) are monotonically increasing from the left to the right. And no item is compatible with any other item in the same list.
\item
	$\forall a_i$ $\alpha$, all predecessors of $a_i$ form a nonempty consecutive block in $\mathbb{L}_{\alpha}^{t-1}$ ($t > 1$).
\item
	$un_\alpha(a_i)$ is the rightmost predecessor of $a_i$ in $\mathbb{L}_{\alpha}^{t-1}$ ($t > 1$).
\end{enumerate}
}
{ 

	1. It holds according to the definition of predecessor. 

	2. Since $a_i$ $= rn_{\alpha}(a_j)$, $a_i$ is after $a_j$ and $i > j$. Besides, $a_i$ $< a_j$, otherwise if $a_i$ $> a_j$, then $a_j \seqpar{\alpha}$ $a_i$ and the rising length of $a_i$  and $a_j$ could not be the same, which contradicts the definition of right neighbor.

	3. For $t>1$, predecessors of $a_i$ locates in $\mathbb{L}_{\alpha}^{t-1}$. Assuming that $a_{k_1}$ and $a_{k_2}$ are two predecessor of $a_i$ in $\mathbb{L}_{\alpha}^{t-1}$ and $a_{k'}$ is an item between $a_{k_1}$ and $a_{k_2}$ in $\mathbb{L}_{\alpha}^{t-1}$. We know that items in $\mathbb{L}_{\alpha}^{t-1}$ are decreasing from the left to the right while their subscripts are increasing and consequently, $a_{k'}$ $< a_{k_1}$ $< a_i$ and $k'$ $< k_2$ $< i$. Hence, $a_{k'}$ $\seqpar{\alpha} a_i$. Besides, $a_{k'}$ $\in$ $\mathbb{L}_{\alpha}^{t-1}$ and $a_{k'}$ must be a predecessor of $a_i$. Thus, items between predecessors of $a_i$ in $\mathbb{L}_{\alpha}^{t-1}$ are also predecessors of $a_i$ and all predecessors of $a_i$ form a consecutive block.
	
	4. (By contradiction) According to the definition of up neighbor, $un_{\alpha}(a_i)$ is before $a_i$ and $RL_{\alpha}(un_{\alpha}(a_i))$ $= RL_{\alpha}(a_i)-1$.  Assuming that $un_{\alpha}(a_i)$ is not a predecessor of $a_i$, then $un_{\alpha}$ $> a_i > a_j$ where $a_j$ is a predecessor of $a_i$. Since $un_{\alpha}(a_i)$ is nearer to $a_i$ than $a_j$, $un_{\alpha}(a_i)$ is at the right of $a_j$. Thus, $a_j \seqpar{\alpha}$ $un_{\alpha}(a_i)$ and $RL_{\alpha}(un_{\alpha}(a_i))$ $\geq RL_{\alpha}(a_i)+1$ $= RL_{\alpha}(a_i)$ which contracts the definition of up neighbor. Thus, $un_{\alpha}(a_i)$ is the nearest predecessor of $a_i$. Besides, according to Statement 2, $un_{\alpha}(a_i)$ is the rightmost predecessor of $a_i$ in $\mathbb{L}_{\alpha}^{t-1}$.
}

\mylemma 
{lem:property}
{ 
Given a sequence $\alpha$ and its corresponding index $\mathbb{L}_{\alpha}$, for $a_i \in \mathbb{L}_{\alpha}^{t}$, for any $1\leq t \leq m$.
\begin{enumerate}
	\item 
		$RL_{\alpha}(a_i)$ = $t$  if and only if $a_i$ $\in$ $\mathbb{L}_{\alpha}^{t}$, namely:
		\[\forall i \in [1, w] ,RL_{\alpha}(a_i) = t \Leftrightarrow a_i  \in  \mathbb{L}_{\alpha}^{t}\]
	\item 
		$un_\alpha(a_i)$(if exists) is the rightmost item in $\mathbb{L}_{\alpha}^{t-1}$ which is before $a_i$ in sequence $\alpha$.
	\item 
		$dn_\alpha(a_i)$(if exists) is the rightmost item in $\mathbb{L}_{\alpha}^{t+1}$ which is before $a_i$ in sequence $\alpha$. Besides, $dn_{\alpha}(a_i)$ $> a_i$.
\end{enumerate}	
}
{ 
We just prove that $dn_{\alpha}(a_i)$ $> a_i$ since the other claims hold obviously according to the definitions of horizontal list, up neighbor and down neighbor respectively.  Assuming that $a_j$ $\in \mathbb{L}_{\alpha}^{t}$ is a predecessor of $dn_{\alpha}(a_i)$, then $a_j$ is before $dn_{\alpha}(a_i)$. Hence, $a_j$ is before $a_i$ since $dn_{\alpha}(a_i)$ is before $a_i$. Thus, $a_j$ is at the left of $a_i$ in $\mathbb{L}_{\alpha}^{t}$ and $a_j$ $> a_i$(Lemma \ref{lem:consecutive}(\ref{item:decreasing})). Besides, $dn_{\alpha}(a_i)$ $> a_j$, thus $dn_{\alpha}(a_i)$ $> a_i$.
}

\mylemma 
{lem:tailsorted}
{ 
Given a sequence $\alpha$ and its corresponding index $\mathbb{L}_{\alpha}$, the following equation holds
\[Tail(\mathbb{L}^i_{\alpha}) \leq Tail(\mathbb{L}^j_{\alpha}) \leftrightarrow i < j\], where $1\leq i < j \leq m$ and $m$ is the number of horizontal lists in $\mathbb{L}_{\alpha}$.
}
{ 
Consider $Tail(\mathbb{L}_{\alpha}^{i})$ and $Tail(\mathbb{L}_{\alpha}^{i+1})$ where $1 \leq i < m$, assuming that $a_p$ is a predecessor of $Tail(\mathbb{L}_{\alpha}^{i+1})$ then $a_p$ $< Tail(\mathbb{L}_{\alpha}^{i+1})$. Besides, $Tail(\mathbb{L}_{\alpha}^{i})$ $\leq a_p$ since items in $\mathbb{L}_{\alpha}^{i+1}$ is decreasing from the left to the right, thus, $Tail(\mathbb{L}_{\alpha}^{i})$ $\leq a_p$ $< Tail(\mathbb{L}_{\alpha}^{i+1})$. Apparently, this claim holds. 
}

\mylemma 
{theorem:nocross}
{ 
Let $\alpha=\{a_1,...,a_w\}$ be a sequence and $\mathbb{L}_{\alpha}$ be its corresponding index. Let $m$ be the number of horizontal lists in $\mathbb{L}_{\alpha}$. Let $a_i$ and $a_j$ be two items in $\mathbb{L}_{\alpha}^t, t \geq 1$. If $a_i$ is on the left of $a_j$, $un^{k}_{\alpha}(a_i)=un^{k}_{\alpha}(a_j)$ or $un^{k}_{\alpha}(a_i)$ is on the left of $un^{k}_{\alpha}(a_j)$, for every $0\leq k< t$.

}
{ 
If $t = 1$, $un_{\alpha}^{0}(a_i) = a_i$ is certainly on the left of $un_{\alpha}^{0}(a_j)$. If $t > 1$, $un_{\alpha}(a_i)$ is before $a_i$ in $\alpha$. Since $a_i$ is on the left of $a_j$, $a_i$ is certainly before $a_j$ in $\alpha$ (Lemma \ref{lem:consecutive}(\ref{item:decreasing})), hence, $un_{\alpha}(a_i)$ is also before $a_j$ in $\alpha$. While, $un_{\alpha}(a_j)$ is the rightmost item in $\mathbb{L}_{\alpha}^{t-1}$ who is before $a_j$ (Lemma \ref{lem:property}(\ref{item:unrightmost})). Thus, $un_{\alpha}(a_i)$ is either $un_{\alpha}(a_j)$ or an item on the left of $un_{\alpha}(a_j)$. Recursively,  for every $0\leq k < t$, $un_{\alpha}^{k}$ is either $un_{\alpha}^{k}(a_j)$ or an item on the left of $un_{\alpha}^{k}(a_j)$
}

\mylemma 
{lem:dnleftunright}
{ 
  Given a sequence $\alpha$ and corresponding $\mathbb{L}_{\alpha}$,
  \begin{enumerate}
  \item
  $\forall a_i \in$ $Left(\mathbb{L}_{\alpha}^{t})$, $dn_{\alpha}(a_i)$ (if exists) $\in$ $Left(\mathbb{L}_{\alpha}^{t+1})$.
  \item $\forall a_i \in$ $Right(\mathbb{L}_{\alpha}^{t+1})$, $un_{\alpha}(a_i)$ (if exists) $\in$ $Right(\mathbb{L}_{\alpha}^{t})$.
  \end{enumerate}
}
{ 

	  1. Let $a_h$ $ = Head(Right(\mathbb{L}_{\alpha}^{t}))$ and $a_t$ $=Tail(Left(\mathbb{L}_{\alpha}^{t+1}))$. $dn_{\alpha}(a_i)$ is before $a_i$ in $\alpha$ (Lemma \ref{lem:property}(\ref{item:dnrightmost})), while $a_i$ is before $a_h$ in $\alpha$ according to the horizontal adjustment, thus, $dn_{\alpha}(a_i)$ is before $a_h$ in $\alpha$. Also according to the horizontal adjustment, $dn_{\alpha}(a_h) = a_t$, $a_t$ is the rightmost item in $\mathbb{L}_{\alpha}^{t+1}$ who is before $a_h$ in $\alpha$. Thus, $dn_{\alpha}(a_i)$ is either $a_t$ or some item on the left of $a_t$. Since $a_t$ is the tail item of $Left(\mathbb{L}_{\alpha}^{t+1})$, $dn_{\alpha}(a_i) \in Left(\mathbb{L}_{\alpha}^{t+1})$.

     2. Assuming that $a_h$ and $a_t$ are $Head(Right(\mathbb{L}_{\alpha}^{t}))$  and $Tail(Left(\mathbb{L}_{\alpha}^{t+1}))$ respectively. $a_t$ is the rightmost item in $\mathbb{L}_{\alpha}^{t+1}$ that is before $a_h$ in $\alpha$ since $dn_{\alpha}(a_h) = a_t$ (Lemma \ref{lem:property}(\ref{item:dnrightmost})). while $a_i$ is on the right of $a_t$ in $\mathbb{L}_{\alpha}^{t+1}$, $a_h$ must be before $a_i$ in $\alpha$. Since $un_{\alpha}(a_i)$ is the rightmost item in $\mathbb{L}_{\alpha}^{t}$ that is before $a_i$ in $\alpha$, $un_{\alpha}(a_i)$ can not be on the left of $a_h$. Hence,  $un_{\alpha}(a_i) \in$ $Right(\mathbb{L}_{\alpha}^{t})$.
}

\mylemma 
{lem:vertical}
{ 
Let $\alpha= \{a_1,a_2,\cdots,a_w)$ be a sequence. Let $\mathbb{L}_{\alpha}$ be its corresponding index and $m$ be the total number of horizontal lists in $\mathbb{L}_{\alpha}$. Let $\alpha^- = \{a_2,\cdots, a_w\}$ be obtained from $\alpha$ by deleting $a_1$. Consider an item $a_i \in \mathbb{L}_{\alpha^{-}}^{t}$, where $1\leq t \leq m$. According to the horizontal list adjustment, there are two cases for $a_i$: $a_i$ is from $Left(\mathbb{L}_{\alpha}^{t+1})$ or $a_i$ is from $Right(\mathbb{L}_{\alpha}^{t})$. Then, the following claims hold:

	\begin{enumerate}
	\item Assuming $a_i$ is from $Left(\mathbb{L}_{\alpha}^{t+1})$
		\begin{enumerate}
		\item \label{item:dnremain}
			$dn_{\alpha^-}(a_i) = dn_{\alpha}(a_i)$ (i.e., the down neighbor do not change).
		\item \label{item:unnotx}
			Let $x$ be the rightmost item of $Left(\mathbb{L}_{\alpha}^{t})$.

		if $un_{\alpha}(a_i) \neq x$, $un_{\alpha^-}(a_i) = un_{\alpha}(a_i)$ (i.e., the up neighbor do not change).
		
		
		\end{enumerate}
		
		\item
			Assuming $a_i$ is from $Right(\mathbb{L}_{\alpha}^{t})$
		
		\begin{enumerate}
		\item \label{item:unremain}
			$un_{\alpha^-}(a_i) = un_{\alpha}(a_i)$ (i.e., the up neighbor do not change).
		\item \label{item:dnnoty}
			Let $y$ be the rightmost item of $Left(\mathbb{L}_{\alpha}^{t+1})$.
		
		if $dn_{\alpha}(a_i) \neq y$, $dn_{\alpha^-}(a_i) = dn_{\alpha}(a_i)$ (i.e., the up neighbor do not change)
		\end{enumerate}
	
	\end{enumerate}

}
{ 
	
	1. If $a_i$ is from $Left(\mathbb{L}_{\alpha}^{t+1})$

		   (a) Let $a_d^- = dn_{\alpha^-}(a_i)$, if $a_d^-$ exists, the following three claims holds:

			  \ \ \ \ i. $a_d^-$ is before $a_i$ in $\alpha$(Also $\alpha^-$):\ \   this holds according to the definition of down neighbor.
		
		      \ \ \ \ ii. $a_d^-$ $\in Left(\mathbb{L}_{\alpha}^{t+2})$:\ \  since $a_d^-$ $\in \mathbb{L}_{\alpha^-}^{t+1}$, 
	          $a_d^-$ comes from either $Left(\mathbb{L}_{\alpha}^{t+2})$ or $Right(\mathbb{L}_{\alpha}^{t+1})$ (according to the horizontal update method), however, all items in $Right(\mathbb{L}_{\alpha}^{t+1})$ are after $a_i$ in $\alpha$ since $a_i$ comes from $Left(\mathbb{L}_{\alpha}^{t+1})$, hence, $a_d^-$ can only come from $Left(\mathbb{L}_{\alpha}^{t+2})$.
		
		      \ \ \ \ iii. $a_d^-$ is exactly $dn_{\alpha}(a)$:\ \  since $a_d^-$ $\in Left(\mathbb{L}_{\alpha}^{t+2})$, if $a_d^-$ is not $Tail(Left(\mathbb{L}_{\alpha}^{t+2}))$ , then $rn_{\alpha}(a_d^-)$ $\in Left(\mathbb{L}_{\alpha}^{t+2})$ is exactly $rn_{\alpha^-}(a_d^-)$ (According to the horizontal update) and $a_d^-$ is the rightmost item in $\mathbb{L}_{\alpha}^{t+2}$ who is before $a$ in $\alpha$, thus, $a_d^-$ is $dn_{\alpha}(a)$ (Lemma \ref{lem:property}(\ref{item:dnrightmost})); if $a_d^-$ is the tail item of $Left(\mathbb{L}_{\alpha}^{t+2})$, then $a_d^-$ is the rightmost item in $Left(\mathbb{L}_{\alpha}^{t+2})$ who is before $a$ in $\alpha$, and $dn_{\alpha}(a_i)$ can only be $a_d^-$ because we know that $dn_{\alpha}(a_i)$ $\in Left(\mathbb{L}_{\alpha}^{t+2})$(Lemma \ref{lem:dnleftunright}).
		    
		        Besides, if $a_d^-$ does not exist, there is no item in $\mathbb{L}_{\alpha^-}^{t+1}$ who is before $a_i$ in $\alpha^-$, which means there is no item in $Left(\mathbb{L}_{\alpha}^{t+2})$(Also $\mathbb{L}_{\alpha}^{t+2}$) who is before $a_i$ in $\alpha$, namely, $dn_{\alpha}(a)$ does not exist, either. Above all, $dn_{\alpha^-}(a_i) = dn_{\alpha}(a_i)$.

		(b)	if $un_{\alpha}(a_i)$ is not $x$, then $un_{\alpha}(a_i)$ can only be an item on the left of $x$ in $Left(\mathbb{L}_{\alpha}^{t-1})$ (Lemma \ref{lem:dnleftunright}). Then $rn_{\alpha}(un_{\alpha}(a_i))$ must be the same as $rn_{\alpha^-}(un_{\alpha}(a_i))$ according to our horizontal adjustment. Thus, $un_{\alpha}(a_i)$ is still the rightmost item in $\mathbb{L}_{\alpha^-}^{t-2}$ who is before $a_i$ in $\alpha$(Also $\alpha^-$), namely, $un_{\alpha}(a_i)$ is exactly $un_{\alpha^-}(a_i)$.
		
		
		2. If $a_i$ is from $Right(\mathbb{L}_{\alpha}^{t})$

		 (a) If $t = 1$,  $un_{\alpha}(a_i)$ = $un_{\alpha^-}(a_i) = NULL$ according to our horizontal adjustment. If $t > 1$, $un_{\alpha}(a_i) \in$ $Right(\mathbb{L}_{\alpha}^{t-1})$ (Lemma \ref{lem:dnleftunright}), thus, $un_{\alpha}(a_i) \in$ $\mathbb{L}_{\alpha^-}^{t-1}$. $rn_{\alpha}(un_{\alpha}(a_i))$(if exist) is after $a_i$ in $\alpha$, hence, $rn_{\alpha^-}(un_{\alpha}(a_i))$ is after $a_i$ in $\alpha^-$ because $rn_{\alpha}(un_{\alpha}(a_i))$ and $rn_{\alpha^-}(un_{\alpha}(a_i))$ is the same item(or both of them don't exist) according to the horizontal adjustment. Thus, $un_{\alpha}(a_i)$ is the rightmost item in $\mathbb{L}_{\alpha^-}^{t-1}$ whose position is before $a_i$ in $\alpha$, namely, $un_{\alpha}(a_i)$ is exactly $un_{\alpha^-}(a_i)$.

		(b) Since $y$ is before $Head(Right(\mathbb{L}_{\alpha}^{t}))$ in $\alpha$, then $y$ is also before $a_i$ in $\alpha$. Besides, $dn_{\alpha}(a_i)$ is the rightmost item in $\mathbb{L}_{\alpha}^{t+1}$ who is before $a_i$, then $dn_{\alpha}(a_i)$ is either $y$ or an item on the right of $y$. If $dn_{\alpha}(a_i)$ is not $y$, $dn_{\alpha}(a_i)$ must be  in $Right(\mathbb{L}_{\alpha}^{t+1})$. Hence, $rn_{\alpha^-}(dn_{\alpha}(a_i))$ will be the same as $rn_{\alpha}(dn_{\alpha}(a_i))$, thus, $dn_{\alpha}(a_i)$ is still the rightmost item in $\mathbb{L}_{\alpha^-}^{t+1}$ who is before $a_i$ in $\alpha$(Also $\alpha^-$), namely, $dn_{\alpha^-}(a_i) = dn_{\alpha}(a_i)$.
}




\newcommand{\mytheorem}[3]{
\textbf{Proof of Theorem \ref{#1}}


\vspace{-0.05in}
\begin{proof} 
#3
\end{proof}
\vspace{-0.06in}
}

\nop{
\mytheorem 
{theorem:pred}
{ 
Given a sequence $\alpha$ and two items $a_i$, $a_j$ $\in \alpha$,
\begin{enumerate}
\item 
	If $a_j$ is a predecessor of $a_i$, $\forall s^{\prime}$ $\in MIS_{\alpha}(a_j)$, the sequence $s$ $= s^{\prime}\oplus a_i$\footnote{$\oplus$ means appending item $a_i$ to the end of $s^{\prime}$} is an increasing subsequence ending with $a_i$ whose length equals to $RL_{\alpha}(a_i)$, namely, $s$ $\in MIS_{\alpha}(a_i)$.
\item
	$\forall s$ $\in MIS_{\alpha}(a_i)$, assuming that $a_j$ is the item ahead of $a_i$ in $s$, then $a_j$ is a predecessor of $a_i$ and the sequence $s^{\prime}= s \backslash a_i$ \footnote{$\backslash$ means deleting the last item $a_i$ from $s$} is an increasing subsequence ending with $a_j$ whose length equals to $RL_{\alpha}(a_j)$, namely,  $s^{\prime}$ $\in MIS_{\alpha}(a_j)$.
\end{enumerate}
}
{ 
	1. Apparently, sequence $s = $ $s^{\prime}\oplus a_i$ is an increasing subsequence ending with $a_i$ since $a_j$ $\seqpar{\alpha} a_i$. Besides, $|s|$ $= |s^{\prime}|+1$ $= RL_{\alpha}(a_j)+1$ $=RL_{\alpha}(a_i)$, and thus, $s$ $\in MIS_{\alpha}(a_i)$.

	2. It can be easily proved with the definition of predecessor and increasing subsequence.
}
} 

\nop{
\mytheorem 
{theorem:spacecostDAG}
{ 
Given a sequence $\alpha$ of length $w$, the graph $G(\alpha)$ defined in Definition \ref{def:daggraph} has $O(w^2)$ space.
}
{ 
For each item $a_i$ in $\alpha$, the number of predecessors of $a_i$ is $O(w)$ since any item before $a_i$ may be a predecessor of $a_i$ and space complexity of $G(\alpha)$ is $O(w^2)$. In a most extreme case where the first $w/2$ items and the last $w/2$ items are decreasing, respectively, while the first item of $\alpha$ is less than the last item, then for each item $a_i$ where $w/2 < i$, all of the first $w/2$ items of $\alpha$ are predecessors of $a_i$ and the space cost is $w^2/4$.
} 
} 


\mytheorem 
{theorem:indexspace}
{ 
 Given a sequence $\alpha=\{a_1,...,a_w\}$, the index $\mathbb{L}_{\alpha}$ defined in Definition \ref{def:orthogonal} has $O(w)$ space.
}
{ 
Each item in sequence $\alpha$ has at most four neighbors in $\mathbb{L}_{\alpha}$ (Some neighbors of an item can be NULL). So the space cost is $O(w)$.
}

\mytheorem 
{theorem:timeconstruct}
{ 
Let $\alpha = \{a_1,a_2,...,a_w\}$ be a sequence with $w$ items. Then we have the following:

\begin{enumerate}
\item   The time complexity of  Algorithm \ref{alg:insertelement} for inserting one item is $O(log w)$.

\item The time complexity of Algorithm \ref{alg:buildingindex} for building the whole corresponding index is $O(wlogw)$.
\end{enumerate}
}
{ 
1. Consider Algorithm \ref{alg:insertelement}. Binary search costs $O(log m)$ time, where $m$ denotes the number of horizontal lists in $\mathbb{L}_{\alpha}$. All other operations cost $O(1)$ time. Since $m \leq w$, so the time complexity of Algorithm \ref{alg:insertelement} is $O(logw)$.

2. Consider Algorithm \ref{alg:buildingindex}. Algorithm \ref{alg:buildingindex} loops on Algorithm \ref{alg:insertelement} for $O(w)$ times and Algorithm \ref{alg:insertelement} costs $O(logw)$ time. Thus the time complexity of Algorithm \ref{alg:buildingindex} for building the corresponding index is $O(wlogw)$.
}

\mytheorem 
{theorem:allLIS}
{ 
The time complexity of Algorithm \ref{alg:lisenum} to enumerate all LIS in $\alpha$ is $O($\emph{OUTPUT}$)$, where \emph{OUTPUT} is the sum of all LIS lengths.
}
{ 
The correctness of the theorem is based on the following simple facts: (1) Every item pushed into the stack and popped out from the stack is printed into a LIS at least once. Hence, associated cost is at most $3$ times of the output size. (2) Items scanned but not pushed into the stack (i.e., items that are on the left of all predecessors) occur at most once at each level $\mathbb{L}_{\alpha}^{k}, 1 \leq k \leq m$. Hence, associated cost is at most one time of the output size. So the total cost is at most $4$ times the output size.
}

\mytheorem 
{lem:deletion}
{ 
Let $\alpha=$ \{$a_1,a_2,\cdots,a_w$\} be a sequence. Let $\mathbb{L}_{\alpha}$ be its corresponding index and $m$ be the total number of horizontal lists in $\mathbb{L}_{\alpha}$. Let $\alpha^- = \{a_2,\cdots, a_w\}$ be obtained from $\alpha$ by deleting $a_1$. Then for any $a_i, 2\leq i\leq m \in \mathbb{L}_{\alpha}^ t, 1\leq t \leq m$, we have the following:
\begin{enumerate}
\item If $un^{t-1}_{\alpha}(a_i)$ is $a_1$, then $RL_{\alpha^-}(a_i) = RL_{\alpha}(a_i)-1$.
\item If $un^{t-1}_{\alpha}(a_i)$ is not $a_1$, then $RL_{\alpha^-}(a_i) = RL_{\alpha}(a_i)$.
\end{enumerate}
where $un^{t-1}_{\alpha}(a_i)$ is defined in Definition \ref{def:khopup}.
}
{ 
First note that, any increasing subsequence of $\alpha^-$ that ends with $a_i$ is also an increasing subsequence of $\alpha$ that ends with $a_i$. Therefore, $RL_{\alpha^-}(a_i) \leq RL_{\alpha}(a_i)$. On the other hand, $a_i$ can only be head item of any increasing subsequence since $a_1$ is the first item of $\alpha$, thus, once $a_1$ is removed, the length of increasing subsequence ending with $a_i$ in $\alpha$ can at most decrease by $1$. Therefore, $RL_{\alpha^-}(a_i) \geq RL_{\alpha}(a_i)-1$.

1. Consider the case $un^{t-1}_{\alpha}(a_i)$ is $a_1$. $a_i$ is in $\mathbb{L}_{\alpha}^t$.
Assuming that $s$ $\in MIS_{\alpha}(a_i)$ where $s = $ \{$a_{i_{t-1}}$, $\cdots$,$a_{i_{1}}$, $a_{i_{0}} = a_i$\}, $a_{i_{1}}$ is a predecessor of $a_i$ in $\mathbb{L}_{\alpha}^{t-1}$. Consider another sequence $s'$ where $s' = $ $(un^{t-1}_{\alpha}(a_i)$, $\cdots$,$un^{1}_{\alpha}(a_i), un^{0}_{\alpha}(a_i))$. Obviously, $s'$ $\in MIS_{\alpha}(a_i)$ and the item $un^{1}_{\alpha}(a_i)$ is also in $\mathbb{L}_{\alpha}^ {t-1}$. According to Lemma \ref{lem:property}(\ref{item:unrightmost}), $a_{i_{1}}$ is on the left of $un^{1}_{\alpha}(a_i)$ (could be $un^{1}_{\alpha}(a_i)$ itself). Therefore, according to Lemma \ref{theorem:nocross},  $un_{\alpha}(a_{i_1})$ is on the left of $un_\alpha( un^{1}_{\alpha}(a_i))$, which is $un^{2}_{\alpha}(a_i)$. Note that, $a_{i_{2}}$ is a predecessor of $a_{i_{1}}$.
Hence, according to Lemma \ref{lem:property}(\ref{item:unrightmost}), $a_{i_{2}}$ is on the left of $un_{\alpha}(a_{i_1})$. So $a_{i_{2}}$ is on the left of $un^{2}_{\alpha}(a_i)$.
This argument continues and we have every $a_{i_{t-j}}$ is on the left of $un^{t-1}_{\alpha}(a_i)$ (could be the same item) for every $1\leq j<t$.
Thus, if $un^{t-1}_{\alpha}(a_i)$ is $a_1$, then each sequence in $MIS_{\alpha}(a_i)$ begins with $a_1$ and the rising length of $a_i$ must decrease by 1 after deleting $a_1$. Therefore $RL_{\alpha^-}(a_i) = RL_{\alpha}(a_i)-1$.

2. Consider the case $un^{t-1}_{\alpha}(a_i)$ is not $a_1$. $\beta= \{un^{t-1}_{\alpha}(a_i),\cdots,un^0_{\alpha}(a_i)\}$ is an increasing subsequence of ending with $a_i$ in $\alpha$. Since $un^{t-1}_{\alpha}(a_i) \neq a_1$, so $\beta$ is also an increasing subsequence of $\alpha^-$. Besides, $|\beta|$ is $RL_{\alpha}(a_i)$, therefore, we have $RL_{\alpha^-}(a_i) = RL_{\alpha}(a_i)$.
}

\mytheorem 
{theorem:updatedecreasing}
{ 
The list formed by appending $Right(\mathbb{L}_{\alpha}^{t})$ to $Left(\mathbb{L}_{\alpha}^{t+1})$ are monotonic decreasing from the left to the right.
}
{ 
Since $Left(\mathbb{L}_{\alpha}^{t+1})$ is sublist of $\mathbb{L}_{\alpha}^{t+1}$ which is monotonic decreasing, $Left(\mathbb{L}_{\alpha}^{t+1})$ is monotonic decreasing too. Similar, $Right(\mathbb{L}_{\alpha}^{t})$ is also monotonic decreasing. If $Left(\mathbb{L}_{\alpha}^{t+1})$ or $Right(\mathbb{L}_{\alpha}^{t})$ is $NULL$, this theorem holds certainly. Otherwise, let $a_j$ be the last item in $Left(\mathbb{L}_{\alpha}^{t+1})$ and $a_k$ be the first item $Right(\mathbb{L}_{\alpha}^{t})$. According to the way we divide horizontal lists of $\mathbb{L}_{\alpha}$, $a_j$ is the down neighbour of $a_k$. Thus, $a_k < dn_\alpha(a_k) = a_j$(Lemma \ref{lem:property}(\ref{item:dnrightmost})). Therefore, the list formed by appending $Right(\mathbb{L}_{\alpha}^{t})$ to $Left(\mathbb{L}_{\alpha}^{t+1})$ is monotonic decreasing from the left to the right.
}
 
 
\mytheorem 
{theo:deletetime}
{ 
The time complexity of Algorithm \ref{alg:deletion} is $O(w)$, where $w$ denotes the time window size.
}
{ 
We can see that the time complexity of Algorithm \ref{alg:division} is $O(|LIS|)$ since division of each horizontal list costs $O(1)$ and there are $|LIS|$ horizontal lists in total. Besides, during the up neighbors update(Lines \ref{code:unupdateBegin}-\ref{code:unupdateEnd}), each horizontal list is scanned at most twice and each item in $\alpha$ will be scanned at most twice. Similarly, during the down neighbors update(Lines \ref{code:dnupdateBegin}-\ref{code:dnupdateEnd}), each item in $\alpha$ is also scanned at most twice. Since $|\alpha| = w$, the time complexity of Algorithm \ref{alg:deletion} is $O(|LIS|+w)$, namely, $O(w)$ since $|LIS|$ $\leq w$.
}


\mytheorem 
{theorem:constrainedmining}
{ 
Let $\alpha=\{a_1,...,a_w\}$ be a sequence and $\mathbb{L}_{\alpha}$ be its corresponding index. Let $m$ be the number of horizontal lists in $\mathbb{L}_{\alpha}$ and DAG $G_{\alpha}$ be the corresponding DAG created  from $\mathbb{L}_{\alpha}$.
\begin{enumerate}
\item Let $a_i$ and $a_j$ be two items in $\mathbb{L}_{\alpha}^t, t \geq 1$. then $lm^{k}_{\alpha}(a_i)$, $lm^{k}_{\alpha}(a_j)$, $un^{k}_{\alpha}(a_i)$ and $un^{k}_{\alpha}(a_j)$ are all in $\mathbb{L}_{\alpha}^{t-k}$, for every $1\leq k< t$. Furthermore, if $a_i < a_j$, then $lm^{k}_{\alpha}(a_i) \leq lm^{k}_{\alpha}(a_j)$ and $un^{k}_{\alpha}(a_i) \leq un^{k}_{\alpha}(a_j)$.

\item Let $a_i$ be an item in $\mathbb{L}_{\alpha}^t, t\geq 1$. Let $\beta =\{a_{i_{t-1}}, \cdots,a_{i_{k}}, \cdots, a_{i_{0}}=a_i \}$ be any longest increasing subsequence (LIS) that ends at $a_i$ (note that $a_{i_{0}}$ and $a_i$ are used to denote the same item for presentation simplicity). Then $lm^{k}_{\alpha}(a_i)$, $a_{i_{k}}$, and $un^{k}_{\alpha}(a_i)$ are all in $\mathbb{L}_{\alpha}^{t-k}$, for every $0\leq k \leq t-1$. Furthermore, $lm^{k}_{\alpha}(a_i) \leq a_{i_{k}} \leq un^{k}_{\alpha}(a_i)$.

\item Let $a_i$ be an item in $\mathbb{L}_{\alpha}^{m}$ (i.e, the last list). Then $\{lm^{m-1}_{\alpha}(a_i)$, $\cdots$,$lm^{0}_{\alpha}(a_i))$ has maximum weight and minimum gap among all LIS that end at $a_i$,  $\{un^{m-1}_{\alpha}(a_i), \cdots,un^{0}_{\alpha}(a_i))$ has the minimum weight and maximum gap among the all LIS that end at $a_i$.

\item Let $a^{m}_h$ and $a^{m}_t$ be the head and tail of $\mathbb{L}_{\alpha}^{m}$ respectively. Then
the LIS  $\{lm^{m-1}_{\alpha}(a^{m}_h), \cdots,lm^{0}_{\alpha}(a^{m}_h))$ has the maximum weight. The LIS $\{un^{m-1}_{\alpha}(a^{m}_t), \cdots,un^{0}_{\alpha}(a^{m}_t))$ has the smallest weight.

\end{enumerate}
}
{ 
1. $a_i$ and $a_j$ are in $\mathbb{L}_{\alpha}^t$. By the definition of leftmost child and up neighbor, both the leftmost child and the up neighbor of an item $a_p$ are placed at the horizontal list above the horizontal list $a_p$ is in. Therefore, $lm^{k}_{\alpha}(a_i)$, $lm^{k}_{\alpha}(a_j)$, $un^{k}_{\alpha}(a_i)$ and $un^{k}_{\alpha}(a_j)$ are all in $\mathbb{L}_{\alpha}^{t-k}$ for every $1\leq k< t$. Both $a_i$ and $a_j$ are in $\mathbb{L}_{\alpha}^t$, which is a monotonic decreasing subsequence from the left to the right according to Lemma \ref{lem:consecutive}(\ref{item:decreasing}) . Therefore, if $a_i < a_j$, then $j<i$. Denote the indexes (i.e., their positions in $\alpha$) of $lm_{\alpha}(a_i)$ and $lm_{\alpha}(a_j)$ by $i'$ and $j'$ respectively. Then $j'<j$, $a_{j'}< a_j$ and $i'<i$, $a_{i'}< a_i$. Next we want to show that $j' \leq i'$. Assume, for the sake of contradiction, that $i'< j'$. Combined with $j'<j$, we have $i'<j$. In addition, because $a_{i'}<a_i$ and $a_i<a_j$, so $a_{i'}<a_j$. Thus, $a_{i'}$ is compatible with $a_j$, i.e., $lm_{\alpha}(a_i)$ (which is $a_{i'}$) is compatible with $a_j$. We know that $lm_{\alpha}(a_j)$ (i.e., $a_{j'}$) is the leftmost item in  $\mathbb{L}_{\alpha}^{t-1}$ which is compatible with $a_j$. Hence $j'$ is the smallest index of all items in $\mathbb{L}_{\alpha}^{t-1}$ which are compatible with $a_j$. So $j' \leq i'$. Contradiction. Thus, if $a_i < a_j$, then $lm^{1}_{\alpha}(a_i) \leq lm^{1}_{\alpha}(a_j)$. Applying the same argument to $lm^{1}_{\alpha}(a_i)$ and $lm^{1}_{\alpha}(a_j)$ recursively, we have $lm^{k}_{\alpha}(a_i) \leq lm^{k}_{\alpha}(a_j)$ for each $1\leq k< t$. The statement $un^{k}_{\alpha}(a_i) \leq un^{k}_{\alpha}(a_j)$ can be proved symmetrically.

2. According to Statement (1), both $lm^{k}_{\alpha}(a_i)$ and $un^{k}_{\alpha}(a_i)$ are in $\mathbb{L}_{\alpha}^{t-k}$, for every $1\leq k< t$. 
According to Lemma \ref{lem:property}(\ref{item:rlen}), $a_{i_{k}}$ is also in $\mathbb{L}_{\alpha}^{t-k}$ because $RL_{\alpha}(a_{i_k})$ $= t-k$ . 
Consider $a_{i_{1}}$ as a predecessor of $a_{i_0}$ (i.e., $a_i$),  $lm^{1}_{\alpha}(a_i) \geq a_{i_{1}}$ according to Statement (1). 
Hence, $lm^{2}_{\alpha}(a_i) \geq lm^{1}_{\alpha}(a_{i_{1}})$. Besides, $a_{i_{2}}$ is a predecessor of $a_{i_{1}}$. Thus $lm^{2}_{\alpha}(a_i) \geq a_{i_{2}}$. 
Therefore, $lm^{2}_{\alpha}(a_i) \geq a_{i_{2}}$. 
Repeating the above argument $t-1$ times, we have $lm^{k}_{\alpha}(a_i) \geq a_{i_{k}}$ for every $1\leq k<t$. 
The statement $a_{i_{k}} \geq un^{k}_{\alpha}(a_i)$ can be proved symmetrically.

3. According to Statement (2), for any subsequence in $MIS_{\alpha}(a_i)$, its item at $\mathbb{L}_{\alpha}^{m-k}$ is less than or equal to $lm^{k}_{\alpha}(a_i)$ for every $ 0 \leq k < m$. Therefore, \{$lm^{m-1}_{\alpha}(a_i)$, $\cdots$,$lm^{0}_{\alpha}(a_i)$\} has largest weight among all subsequences in $MIS_{\alpha}(a_i)$. Note that, among all subsequences in $MIS_{\alpha}(a)$, \{$lm^{m-1}_{\alpha}(a_i)$, $\cdots$,$lm^{0}_{\alpha}(a_i)$\} has the largest head $lm^{m-1}_{\alpha}(a_i)$. Therefore, it also has smallest gap among all subsequences in $MIS_{\alpha}(a_i)$. Symmetrically, we can prove that \{$un^{m-1}_{\alpha}(a_i)$, $\cdots$,$un^{0}_{\alpha}(a_i)$\} has smallest weight and largest gap among all the subsequence in $MIS_{\alpha}(a_i)$.
 4. It holds obviously according to Statement (3) and the fact that $\mathbb{L}_{\alpha}^{m}$ is monotonically decreasing.
}


\mytheorem 
{timecom:maxh}
{ 
The time complexity of Algorithm \ref{alg:lismaxheightnew} is $O(w+$\emph{OUTUT}$)$, where $w$ denotes the time window size and \emph{OUTPUT} is the total lengths of all LIS with maximal gap.
}
{ 
The sweeping steps from $\mathbb{L}_{\alpha}^{2}$ to $\mathbb{L}_{\alpha}^{m}$ need to access each item at most twice. It takes $O(w)$ time. The output cost is at most one time of the output size. Therefore, the total cost is  $O(w+$\emph{OUTUT}$)$.
}




\nop{
\section{Running Example of Insertion}\label{sec:appendixinsertioinexample}
Figure \ref{fig:ris} illustrates an example of building $\mathbb{L}_{\alpha}$ for the sequence $\alpha=\{a_1=3,a_2=9,a_3=6,a_4=2,a_5=8,a_6=5,a_7=7\}$ . The straight arrows in Figure \ref{fig:ris} denote up or down neighbors while the curved ones denote left or right neighbours. At step 1, we create the first horizontal list $\mathbb{L}_{\alpha}^{1}$ and add $a_1=3$ into $\mathbb{L}_{\alpha}^{1}$(see Figure \ref{fig:ipica}). At step 2, we can see that $a_2=9$ cannot be added into $\mathbb{L}_{\alpha}^{1}$ because $a_1=3 < a_2=9$, then we create $\mathbb{L}_{\alpha}^{2}$ and add $9$ into $\mathbb{L}_{\alpha}^{2}$ and set $un_{\alpha}(a_2)$ = $a_1$ (See Figure \ref{fig:ipicb}). At step 3, since $Tail(\mathbb{L}_{\alpha}^{2})$ $> a_3=6$, we append $a_3=6$ to $\mathbb{L}_{\alpha}^{2}$ and set $ln_{\alpha}(a_3)=a_2$ and $rn_{\alpha}(a_2)=a_3$. We also set the up neighbor of $a_3$ to be the current tail of $\mathbb{L}_{\alpha}^{1}$, i.e, $un_{\alpha}(a_3)=a_1$ (See Figure \ref{fig:ipicc}). We omit the following steps which are explicitly presented in Figure \ref{fig:ipicd}-\ref{fig:ipicg}.
}


\section{LIS Enumeration}
Pseudo codes for for LIS enumertion are presented in Algorithm \ref{alg:lisenum}.
\begin{algorithm}[h!]
\small
\caption{Enumerate $LIS$ in $\alpha$}
 \label{alg:lisenum}
\KwIn{$\alpha$ and $\mathbb{L}_{\alpha}$}
\KwOut{All LIS in $LIS(\alpha)$, i.e., $LIS(\alpha)$}

\For{Each item $a_i$ in $\mathbb{L}_{\alpha}^{m}$}
{
	$stack.clear()$ \\
	$stack.push(a_i)$ \\
	\While{stack is not empty}
	{
		\If{$stack.top() \in \mathbb{L}_{\alpha}^{1}$}{
			OUTPUT($stack$) \\
			$a = stack.top()$ \\
			$stack.pop()$ \\
			\If{$ln_{\alpha}(a) \seqpar{\alpha} stack.top()$}{
				$stack.push(ln_{\alpha}(a))$ \\
				CONTINUE
			}
		}
		
		\If{the last operation of stack is PUSH}{
			$stack.push(un_{\alpha}(stack.top()))$ \\
		}
		\Else{
			$a = ln_{\alpha}(stack.top())$ \\
			$stack.pop()$ \\
			\If{$a \seqpar{\alpha} stack.top()$}{
				$stack.push(a)$ \\
			}	
		}	
	}
}
RETURN
\end{algorithm}

\section{Deletion}
Pseudo codes for maintenance after deletion happens are presented Algorithm \ref{alg:deletion}.

\begin{algorithm}[!h]
\small
\caption{Data structure maintenance after deletion}
\label{alg:deletion}
\KwIn{$\alpha$, $\mathbb{L}_{\alpha}$}
\KwOut{$\mathbb{L}_{\alpha^-}$}
/* Horizontal Update */ \\ \label{code:horizontalBegin}
Call Algorithm \ref{alg:division} to divide each $\mathbb{L}_{\alpha}^{t}$ in $\mathbb{L}_{\alpha}$ into $Left(\mathbb{L}_{\alpha}^{t})$ and $Right(\mathbb{L}_{\alpha}^{t})$ \\
\For{$t \gets 1$ to $m-1$}{
	Building $\mathbb{L}_{\alpha^-}^{t}$ by appending $Right(\mathbb{L}_{\alpha}^{t})$ to $Left(\mathbb{L}_{\alpha}^{t+1})$ \\
	\label{code:horizontalEnd}
}

\label{code:unupdateBegin}
/* Up neighbors update */ \\
\For{$t \gets 1$ to $m-1$}{
	Call Algorithm \ref{alg:updateun} to update the up neighbors of items in $Left(\mathbb{L}_{\alpha}^{t+1})$ \label{code:unupdateEnd}
}
\label{code:dnupdateBegin}
/* Down neighbors update */ \\
\For{$t \gets 1$ to $m-1$}{
	Call Algorithm \ref{alg:updatedn} to update the down neighbors of items in $Right(\mathbb{L}_{\alpha}^{t})$ \label{code:dnupdateEnd}
}

RETURN
\end{algorithm}

\vspace{-0.15in}

\section{LIS with Extreme Gap}
\label{sec:appendix:extremegap}
Pseudo codes of LIS with maximum and minimum gap are presented in Algorithm  \ref{alg:lismaxheightnew} and \ref{alg:mingap}, respectively.

\begin{algorithm}[h!]
\small
\caption{Find LIS with minimum gap}
 \label{alg:mingap}
\KwIn{A sequence $\alpha$, $\mathbb{L}_{\alpha}$}
\KwOut{All LIS of $\alpha$ with minimum gap}
$m = |\mathbb{L}_{\alpha}|$

/* Compute $lm_{\alpha}^{t-1}(a_i)$ for each $a_i$ $\in \mathbb{L}_{\alpha}^{t}$ */\\
For each $a_i$ $\in \mathbb{L}_{\alpha}^{1}$, Set $lm_{\alpha}^{0}(a_i) = a_i$\\
\label{code:lmBegin}
\For{$t\gets 2$ to $m$}
{
    $a_{i} = Head(\mathbb{L}_{\alpha}^{t})$ \\
    $a_{k} = Head(\mathbb{L}_{\alpha}^{t-1})$ \\
    \While{$a_{i} \neq NULL$}
    {
        \While{$a_k > a_i$}
        {
            $a_{k} = rn_{\alpha}(a_k)$
        }
        $lm_{\alpha}^{t-1}(a_i) = lm_{\alpha}^{t-2}(a_k)$ \\
        $a_i = rn_{\alpha}(a_i)$ \label{code:lmEnd}
    }
}

$GAP_{min} = \min{\{ a_i - lm_{\alpha}^{m-1}(a_i)\  |\  a_i \in \mathbb{L}_{\alpha}^{m} \}}$ \label{code:minimum}
\\
\label{code:minheightBegin}
\For{each item $a_i$ in $\mathbb{L}_{\alpha}^{m}$}
{
    \If{$a_i - lm_{\alpha}^{m-1}(a_i) \neq GAP_{min}$}
    {
        CONTINUE
    }
	$stack.clear()$ \label{code:mingap:enum:begin}\\
    $stack.push(a_i)$ \\
	\While{stack is not empty}
	{
	    \If{$stack.top() = lm_{\alpha}^{m-1}(a_i)$}{
		    OUTPUT($stack$) \\
		    $stack.pop()$ \\
	    }
	    \If{the last operation of stack is PUSH}{
		    $stack.push(lm_{\alpha}(stack.top()))$
	    }
	    \Else{
		    $a = rn_{\alpha}(stack.top())$ \\
		    $stack.pop()$ \\
		    /* Assuming that $a \in \mathbb{L}_{\alpha}^{k}$ */ \\
		    \label{code:mingapenum}
		    \If{$a \seqpar{\alpha} stack.top()$ AND $lm_{\alpha}^{k-1}(a) = lm_{\alpha}^{m-1}(a_i)$}
		    {
				$stack.push(a)$ \label{code:minheightEnd} \label{code:mingap:enum:end}
		    }
	    }
	}
}\label{code:minheightfirst}
\end{algorithm}

\vspace{-0.1in}
\begin{algorithm}[!h]
\small
\caption{Find LIS with maximum gap}
 \label{alg:lismaxheightnew}
\KwIn{A sequence $\alpha$ and $\mathbb{L}_{\alpha}$}
\KwOut{All LIS of $\alpha$ with maximum gap}
$m = |\mathbb{L}_{\alpha}|$ \\
For each $a_i \in \mathbb{L}_{\alpha}^{1}$, set $un_{\alpha}^{0}(a_i) = a_i$ \\
\label{code:rmBegin}

\For{$t\gets 2$ to $m$}
{
    \For{each item $a_i$ in $\mathbb{L}_{\alpha}^{t}$}
    {
        $un_{\alpha}^{t-1}(a_i) = un_{\alpha}^{t-2}(un_{\alpha}(a_i))$;\label{code:rmEnd}
    }
}\label{code:rmdynamic}

$GAP_{max} = \max{\{ a_i - un_{\alpha}^{m-1}(a_i)\  |\  a_i \in \mathbb{L}_{\alpha}^{m} \}}$ \label{code:maximum}
\\
\label{code:maxheightBegin}
\For{each item $a_i$ in $\mathbb{L}_{\alpha}^{m}$}
{
    \If{$a_i - un_{\alpha}^{m-1}(a_i) \neq GAP_{max}$}
    {
        CONTINUE
    }    
	$stack.clear()$  \label{code:maxgap:enum:begin} \\
    $stack.push(a_i)$ \\
	\While{stack is not empty}
	{
	    \If{$stack.top() = un_{\alpha}^{m-1}(a_i)$}{
		    OUTPUT($stack$) \\
		    $stack.pop()$ \\
	    }
	    \If{the last operation of stack is PUSH}{
		    $stack.push(un_{\alpha}(stack.top()))$
	    }
	    \Else{
		    $a = ln_{\alpha}(stack.top())$ \\
		    $stack.pop()$ \\
		    /* Assuming that $a \in \mathbb{L}_{\alpha}^{k}$ */ \\
		    \If{$a \seqpar{\alpha} stack.top()$ AND $un_{\alpha}^{k-1}(a) = un_{\alpha}^{m-1}(a_i)$}
		    {
				$stack.push(a)$ \label{code:maxheightEnd} \label{code:maxgap:enum:end}
		    }\label{code:gapenum}
	    }
	}
}\label{code:maxheightfirst}
RETURN \\
\end{algorithm}

\section{LIS with Extreme Weight}
\label{sec:appendix:extremeweight}

The pseudo codes of LIS with extreme weight are presented in Algorithm \ref{alg:lisminweightnew} and Algorithm \ref{alg:lismaxweightnew}.
\begin{algorithm}[!h]
\small
\caption{Find LIS with minimum weight}
 \label{alg:lisminweightnew}
\KwIn{A sequence $\alpha$ and $\mathbb{L}_{\alpha}$}
\KwOut{LIS of $\alpha$ with minimum weight}
$m = |\mathbb{L}_{\alpha}|$ \\
Create an array $S$ with size $m$;\\
$S[0] = Tail(\mathbb{L}_{\alpha}^{m})$ \\
$k = 0$ \\
\While{$un_{\alpha}(S[k]) \neq NULL$}
{
    $S[k+1] = un_{\alpha}(S[k])$;\\
    $k =k+1$;
}
RETURN $S$\\
\end{algorithm}

\begin{algorithm}[!h]
\small
\caption{Find LIS with maximum weight}
 \label{alg:lismaxweightnew}
\KwIn{A sequence $\alpha$ and $\mathbb{L}_{\alpha}$}
\KwOut{LIS of $\alpha$ with maximum weight}
$m = |\mathbb{L}_{\alpha}|$ \\
Create an array $S$ with size $m$;\\
$S[0] = Head(\mathbb{L}_{\alpha}^{m})$  \\
$k = 1$  \\
\While{$k < m$}
{
    $a_i = un_{\alpha}(S[k-1])$ \\
    \While{$ln_{\alpha}(a_i) \seqpar{\alpha} S[k-1]$}
    {
        $a_i = ln_{\alpha}(a_i)$
    }
    $S[k] = a_i$ \\
    $k=k+1$;
}
RETURN $S$ \\
\end{algorithm}

\section{Complexity for sorted sequence}
\label{sec:appendix:sorted}

\subsection{Sequence in Descending Order}
When items are sorted in descending order, the rising length of any item is $1$ and there is only one horizontal list in QN-list consisting all items:
\begin{enumerate}
\item 
	 \underline{$O(1)$ for each insertion}. According to our insertion algorithm, insertion require a binary search over the sequence formed by the tail items of all horizontal lists. However, since there is only one horizontal list, one binary search costs only $O(1)$ time.
\item
	\underline{$O(1)$ for each deletion}. According to our deletion algorithm, when we delete $a_1$, we find that there is no down neighbor of $rn_{\alpha}(a_1) = a_2$ and then deletion is finished. Thus, each deletion costs only $O(1)$ time.
\end{enumerate}

\subsection{Sequence in Ascending Order}

When items are sorted in ascending order, each horizontal list contains only one item and the number of horizontal lists in QN-list is exactly the length of the sequence.
\begin{enumerate}
\item 
	\underline{$O(\log w)$ for each insertion}. The sequence formed by the tail items of all horizontal lists is exactly the number of horizontal lists. Thus, the binary search conducted over the sequence is $O(\log w)$ and each insertion costs $O(\log w)$.
\item
	\underline{$O(1)$ for each deletion}. According to our deletion algorithm, when we delete $a_1$, we find that there is no right neighbor of $a_1$ and then deletion is finished. Thus, each deletion costs only $O(1)$ time.
\end{enumerate}

\section{LIS with Other Constriants}
\label{sec:appendix:other:constraints}
We now discuss how to efficiently support LIS with other existing constraints over our data structure. For each type of constraint, we will first introduce the definition of the corresponding problem and then present the solution over our data structure. In Section \ref{sec:appendix:other:constraints:comparison}, we compare our method with previous work not only theoretically but also experimentally.

\subsection{LIS with Extreme Width}

\subsubsection{Definition}
\begin{definition}\textbf{(Width)}\cite{variant2009} \label{def:width}
Let $\alpha$ be a sequence, $s$ be an LIS in $LIS(\alpha)$ where $s =$ \{$a_{i_1}$, $a_{i_2}$,...,$a_{i_k}$\} ($k=|s|$). The \emph{width} of $s$ is defined as $width(s)=i_k-i_1$, i.e., the positional distance between the tail item($a_{i_k}$) and the head item($a_{i_1}$) of $s$.
\end{definition}
\begin{definition}\textbf{(LIS with extreme width)} \label{def:extreme:width}
Given a sequence $\alpha$, for $s \in LIS(\alpha)$:

$s$ is an \textbf{LIS with Maximum Width} if 
\[ \forall s^{\prime} \in LIS(\alpha), width(s) \geq width(s^{\prime}) \]

$s$ is an \textbf{LIS with Minimum Width} if
\[ \forall s^{\prime} \in LIS(\alpha), width(s) \leq width(s^{\prime}) \]

\end{definition}

\subsubsection{Solution over our data structure}
Given a sequence $\alpha$ and the QN-list $\mathbb{L}_{\alpha}$. Consider an item $a_i \in \mathbb{L}_{\alpha}^m$ where $m = |\mathbb{L}|$. 
Assuming that $\beta =\{a_{i_{m-1}}$, $\cdots,a_{i_{m-2}}$, $\cdots, a_{i_{0}}=a_i \}$ is an LIS ending with $a_i$, then we know that $lm^{m-1}_{\alpha}(a_i) \geq a_{i_{m-1}} \geq un^{m-1}_{\alpha}(a_i)$ (Theorem \ref{item:pmost:compare}). 
However, with Lemma \ref{item:decreasing}, we can conclude that $POS(lm^{m-1}_{\alpha}(a_i)) \leq POS(a_{i_{m-1}}) \leq POS(un^{m-1}_{\alpha}(a_i))$ where $POS(a_{i})$ denotes the position of $a_{i}$ in the sequence. 
Thus, we can see that among all LIS ending with $a_i$, the one with maximum(minimum) width must starts with $lm^{m-1}_{\alpha}(a_i)$ ($un^{m-1}_{\alpha}(a_i)$).

In Section \ref{sec:computation} for computing LIS with extreme gap, we have designed two sweeping algorithm to compute $lm^{t-1}_{\alpha}(a_i)$ (Line \ref{code:lmBegin}-\ref{code:lmEnd} in Algorithm \ref{alg:mingap}) and $un^{t-1}_{\alpha}(a_i)$ (Line \ref{code:rmBegin}-\ref{code:rmEnd} in Algorithm \ref{alg:lismaxheightnew}), respectively, for each item $a_i \in \mathbb{L}_{\alpha}^t$, $1\leq t \leq m$. After finding out some $a_i \in \mathbb{L}_{\alpha}^m$ where $POS(a_i)-POS(un_{\alpha}^{m-1}(a_i))$ is the minimum width, we can enumerate LIS starting with $un_{\alpha}^{m-1}(a_i)$ and ending with $a_i$, of which the pseudo codes are exactly presented at Line \ref{code:maxgap:enum:begin}-\ref{code:maxgap:enum:end} in Algorithm \ref{alg:lismaxheightnew}. Analogously, after finding out some $a_j$ where $POS(a_j)-POS(lm_{\alpha}^{m-1}(a_j))$ is the maximum width, we can enumerate LIS starting with $lm_{\alpha}^{m-1}(a_j)$ and ending with $a_j$ (See Line \ref{code:mingap:enum:begin}-\ref{code:mingap:enum:end} in Algorithm \ref{alg:mingap}).

\begin{algorithm}[t!]
\small
\caption{Find LIS with minimum width}
 \label{alg:min:width}
\KwIn{A sequence $\alpha$ and $\mathbb{L}_{\alpha}$}
\KwOut{All LIS of $\alpha$ with minimum width}
$m = |\mathbb{L}_{\alpha}|$ \\
For each $a_i \in \mathbb{L}_{\alpha}^{1}$, set $un_{\alpha}^{0}(a_i) = a_i$ \\

\For{$t\gets 2$ to $m$}
{
    \For{each item $a_i$ in $\mathbb{L}_{\alpha}^{t}$}
    {
        $un_{\alpha}^{t-1}(a_i) = un_{\alpha}^{t-2}(un_{\alpha}(a_i))$; 
    }
} 

$WIDTH_{min} = \min{\{ POS(a_i) - POS(un_{\alpha}^{m-1}(a_i))\  |\  a_i \in \mathbb{L}_{\alpha}^{m} \}}$ 
\\
\For{each item $a_i$ in $\mathbb{L}_{\alpha}^{m}$}
{
    \If{$a_i - un_{\alpha}^{m-1}(a_i) \neq WIDTH_{min}$}
    {
        CONTINUE
    }    
	$stack.clear()$  
	\\
    $stack.push(a_i)$ \\
	\While{stack is not empty}
	{
	    \If{$stack.top() = un_{\alpha}^{m-1}(a_i)$}{
		    OUTPUT($stack$) \\
		    $stack.pop()$ \\
	    }
	    \If{the last operation of stack is PUSH}{
		    $stack.push(un_{\alpha}(stack.top()))$
	    }
	    \Else{
		    $a = ln_{\alpha}(stack.top())$ \\
		    $stack.pop()$ \\
		    /* Assuming that $a \in \mathbb{L}_{\alpha}^{k}$ */ \\
		    \If{$a \seqpar{\alpha} stack.top()$ AND $un_{\alpha}^{k-1}(a) = un_{\alpha}^{m-1}(a_i)$}
		    {
				$stack.push(a)$ 
		    } 
	    }
	}
} 
RETURN \\
\end{algorithm}


\begin{algorithm}[h!]
\small
\caption{Find LIS with maximum width}
 \label{alg:max:width}
\KwIn{A sequence $\alpha$, $\mathbb{L}_{\alpha}$}
\KwOut{All LIS of $\alpha$ with maximum width}
$m = |\mathbb{L}_{\alpha}|$

/* Compute $lm_{\alpha}^{t-1}(a_i)$ for each $a_i$ $\in \mathbb{L}_{\alpha}^{t}$ */\\
For each $a_i$ $\in \mathbb{L}_{\alpha}^{1}$, Set $lm_{\alpha}^{0}(a_i) = a_i$\\
\For{$t\gets 2$ to $m$}
{
    $a_{i} = Head(\mathbb{L}_{\alpha}^{t})$ \\
    $a_{k} = Head(\mathbb{L}_{\alpha}^{t-1})$ \\
    \While{$a_{i} \neq NULL$}
    {
        \While{$a_k > a_i$}
        {
            $a_{k} = rn_{\alpha}(a_k)$
        }
        $lm_{\alpha}^{t-1}(a_i) = lm_{\alpha}^{t-2}(a_k)$ \\
        $a_i = rn_{\alpha}(a_i)$ 
    }
}

$WIDTH_{max} = \max{\{ POS(a_i) - POS(lm_{\alpha}^{m-1}(a_i))\  |\  a_i \in \mathbb{L}_{\alpha}^{m} \}}$ 
\\
\For{each item $a_i$ in $\mathbb{L}_{\alpha}^{m}$}
{
    \If{$a_i - lm_{\alpha}^{m-1}(a_i) \neq WIDTH_{max}$}
    {
        CONTINUE
    }
	$stack.clear()$ 
	\\
    $stack.push(a_i)$ \\
	\While{stack is not empty}
	{
	    \If{$stack.top() = lm_{\alpha}^{m-1}(a_i)$}{
		    OUTPUT($stack$) \\
		    $stack.pop()$ \\
	    }
	    \If{the last operation of stack is PUSH}{
		    $stack.push(lm_{\alpha}(stack.top()))$
	    }
	    \Else{
		    $a = rn_{\alpha}(stack.top())$ \\
		    $stack.pop()$ \\
		    /* Assuming that $a \in \mathbb{L}_{\alpha}^{k}$ */ \\
		    \If{$a \seqpar{\alpha} stack.top()$ AND $lm_{\alpha}^{k-1}(a) = lm_{\alpha}^{m-1}(a_i)$}
		    {
				$stack.push(a)$ 
		    }
	    }
	}
} 
\end{algorithm}

Apparently, our method for outputting LIS with minimum/maximum width over our data structure cost the same time as our algorithm for LIS with minimum/maximum gap, namely, $O(n+OUTPUT)$. Pseudo codes for LIS with minimum/maximum width are presented in Algorithm \ref{alg:min:width} and Algorithm \ref{alg:max:width}, respectively.


\subsection{Slope-constrained LIS(SLIS)}
\subsubsection{Definition}

\begin{definition}\textbf{(Slope-constrained LIS)}\cite{fastyang2008} \label{def:slope:lis}
Given a sequence $\alpha = $ \{$a_1$, $a_2$,...,$a_n$\} and a nonnegative slope boundary $m$. Computing \emph{slope-constrained LIS} (SLIS) is to output an LIS of $\alpha$: \{$a_{i_1}$,$a_{i_2}$,...,$a_{i_m}$\} such that the slope between two consecutive points is not less than $m$, i.e., $\frac{a_{i_{k+1}}-a_{i_k}}{i_{k+1}-i_k} \geq m$ for all $1\leq k<m$.
\end{definition}

\subsubsection{Solution over our data structure}
With Definition \ref{def:slope:lis}, we can find that the slope only constrains each two consecutive items in an LIS. Thus, the slope are in essence constraints over the predecessors of an item in the sequence. Solution for RLIS computation over our data structure contain two main phrase. In the first phrase, we filter some items that will not exist in an RLIS by coloring them as \emph{black}. In the second phrase, we efficiently output an SLIS over the labeled data structure.

\noindent\underline{\Paragraph{Coloration}}

Items who have no predecessor satisfying the slope constraints will never exist in an RLIS and we can filter those items. Besides, for an non-black item $a_i$, if predecessors of $a_i$ that satisfy the slope constraints are all black, then we can also color $a_i$ as black since there will be no proper predecessor for $a_i$ in an SLIS.  Therefore, black items should be figured out in a recursive way. Since items in $\iLa^1$ have no predecessor, they are all non-black. For convenient, for item $a_i$, we call the non-black predecessor who satisfy the slope constraints with $a_i$ as the \emph{slope-proper} predecessor of $a_i$(Or the predecessor is slope-proper to $a_i$). Let's consider how to color items in $\iLa^{t+1}$ when coloration over items in $\iLa^{t}$ has been done. 

\begin{theorem} \label{theorem:slope:dynamic}
Given a sequence $\alpha$ and $\iLa$. Consider $a_i$, $a_j$ $\in \iLa^{t+1}$ and $a_k \in$ $\iLa^{t}$ where $k<i<j$. If $a_k$ is a leftmost slope-proper predecessor of $a_i$, then the leftmost slope-proper predecessor  of $a_j$ is either $a_k$ or an item at the right of $a_k$.
\end{theorem}
\begin{proof}
With Lemma \ref{lem:consecutive}(\ref{item:decreasing}) we can conclude that:
\[\frac{a_i-a_k}{i-k} > \frac{a_j-a_k}{j-k}\]
Thus, for an non-black item $a_{k'}$ at the left of $a_k$ in $\iLa^{t}$, $a_{k'}$ is either larger than $a_i$ or $(a_i-a_{k'})/(i-k') < m$, thus, $a_{k'}$ will be also either larger than $a_j$ or $(a_j-a_{k'})/(j-k') < m$ which means non-black item at the left side of $a_k$ will not satisfy the slope constraints with $a_j$.
\end{proof}

We know that finding a leftmost slope-proper predecessor for $a_i \in \iLa^{t+1}$ is enough to confirm that $a_i$ is a non-black item. With Theorem \ref{theorem:slope:dynamic}, we can know that after determining the leftmost slope-proper predecessor $a_j$ of $a_i$, the leftmost slope-proper predecessor of $rn_{\alpha}(a_i)$ can be searched from $a_j$ to the right of $\iLa^{t}$. Thus, when coloration over items in $\iLa^{t}$ has been done, we can color items in $\iLa^{t+1}$ by scanning $\iLa^{t}$ and $\iLa^{t+1}$ only once (Line \ref{code:scolor:begin}-\ref{code:scolor:end} in Algorithm \ref{alg:slope}). 

\noindent\underline{\Paragraph{Outputting an SLIS}}

It's easy to know that after the coloration, for any item $a_i$ who is still non-black, there must exist an increasing subsequence $s$ ending with $a_i$ where every item in $s$ is non-black. Thus, outputting an SLIS can be done as following: (1) we firstly find out an item $a_i \in \mathbb{L}_{\alpha}^{m}$ $(m=|\mathbb{L}_{\alpha}|)$ who is non-black. Then we can always find out a slope-proper predecessor $a_j$ of $a_i$.  Recursively, we can find a slope-proper predecessor $a_k$ of $a_j$. Thus, we can easily find out an LIS satisfying slope constraints, namely, SLIS (See Line \ref{code:slis:begin}-\ref{code:slis:end} in Algorithm \ref{alg:slope}). Note that if items in $\mathbb{L}_{\alpha}^{m}$ are all black, there is no SLIS.

Pseudo codes for RLIS over our data structure are presented in Algorithm \ref{alg:slope}


\begin{algorithm}[t!]
\small
\caption{Find an SLIS}
 \label{alg:slope}
\KwIn{A sequence $\alpha$, $\mathbb{L}_{\alpha}$}
\KwIn{User-defined slope $m$}
\KwOut{An SLIS of $\alpha$ satisfying slope $m$}
$m = |\mathbb{L}_{\alpha}|$ \\
/* Coloration */\\
Initial all items as non-black \\
\label{code:scolor:begin}
\For{$t\gets 2$ to $m$}
{
    Let $a_{i} = Head(\mathbb{L}_{\alpha}^{t})$ AND $a_{k} = Head(\mathbb{L}_{\alpha}^{t-1})$ \\
    \While{$a_{i} \neq NULL$}
    {
	    \If{$a_k$ is before $a_i$ AND $a_k$ is not slope-proper to $a_i$}{
		    $a_k = rn_{\alpha}(a_k)$ \\
		    CONTINUE
	    }
	    \If{$a_k$ is after $a_i$ OR $a_k = NULL$}{
		    Color $a_i$ as black 
	    }
		$a_i = rn_{\alpha}(a_i)$  \label{code:scolor:end}
    }
}

Search $\iLa^{m}$ from left to right to find an non-black item $a_i$ \label{code:slis:begin}\\
\If{($a_i$ exists)}
{
	Initial $stack$ with $a_i$ \\
	\While{$stack.top()$ has predecessors}
	{
		$a_j = un_{\alpha}(stack.top())$ \\
		\While{$a_j$ is not slope-proper to $a_i$}{
			$a_j = ln_{\alpha}(a_j)$
		}
		$stack.push(a_j)$ 
	}
	$OUTPUT(stack)$ \label{code:slis:end}\\
}
RETURN
\end{algorithm}

\begin{algorithm}[t!]
\small
\caption{Find an RLIS}
 \label{alg:range}
\KwIn{A sequence $\alpha$, $\mathbb{L}_{\alpha}$}
\KwIn{Two user-defined ranges: $[L_V, U_V]$ and $[L_I, U_I]$}
\KwOut{An RLIS of $\alpha$ satisfying the two ranges}
$m = |\mathbb{L}_{\alpha}|$ \\
/* Coloration */\\
Initial all items as non-black \\
\label{code:rcolor:begin}
\For{$t\gets 2$ to $m$}
{
    Let $a_{i} = Head(\mathbb{L}_{\alpha}^{t})$ AND $a_{k} = Head(\mathbb{L}_{\alpha}^{t-1})$ \\
    \While{$a_{i} \neq NULL$}
    {
	    \If{$\neg Black(a_k)$ AND $L_V\leq$ $a_i-a_k$ AND $i-k\leq U_I$}{
		    \If{$a_i-a_k> U_V$ OR $i-k$ $< L_I$ }{
			    Color $a_i$ as black			     
		    }
		    $a_i = rn_{\alpha}(a_i)$ 
	    }
	    \ElseIf{$a_k = NULL$}{
		    Color $a_i$ as black \\
		    $a_i = rn_{\alpha}(a_i)$ \\
	    }
	    \Else{
		    $a_k = rn_{\alpha}(a_k)$  \label{code:rcolor:end}
	    }
	    
    }
}

Search $\iLa^{m}$ from left to right to find an non-black item $a_i$ \label{code:rlis:begin}\\
\If{($a_i$ exists)}
{
	Initial $stack$ with $a_i$ \\
	\While{$stack.top()$ has predecessors}
	{
		$a_j = un_{\alpha}(stack.top())$ \\
		\While{$a_j$ is not slope-proper to $a_i$}{
			$a_j = ln_{\alpha}(a_j)$
		}
		$stack.push(a_j)$ 
	}
	$OUTPUT(stack)$ \label{code:rlis:end}\\
}
RETURN
\end{algorithm}

\subsection{Range-constrained LIS(RLIS)}
\subsubsection{Definition}

\begin{definition}\textbf{(Range-constrained LIS)}\cite{fastyang2008} \label{def:range:lis}
Given a sequence $\alpha = $ \{$a_1$, $a_2$,...,$a_n$\} and $0<L_I\leq U_I<n$, $0\leq L_V\leq U_V$. Computing \emph{range-constrained LIS} (RLIS) is to output an LIS of $\alpha$: \{$a_{i_1}$,$a_{i_2}$,...,$a_{i_m}$\} satisfying $L_I\leq i_{k+1}-i_{k}\leq U_I$ and $L_V\leq a_{i_{k+1}}-a_{i_k}\leq U_V$.
\end{definition}

\subsubsection{Solution over our data structure}
With Definition \ref{def:range:lis}, we can see that, just like the slope constraints, the range also only constrains each two consecutive items in an LIS. Thus, similar to the solution to SLIS, the solution to RLIS also contains two main phrase, namely, the coloration phrase and output phrase. However, we can easily see that what is different from computing SLIS is the coloration phrase while the outputting RLIS phrase will be exactly the same as that of outputting SLIS.

\noindent\underline{\Paragraph{Coloration}}

The range constraint is different from the slope constraint since the gap and the positional distance between two consecutive items in an LIS should neither be too large nor too small. However, the slope between two items can be arbitrarily large. Similarly, for a non-black predecessor $a_j$ of item $a_i$, if $a_i$, $a_j$ satisfy the range constraint, namely, $(a_i-a_j)\in [L_v, U_v]$ and $(i-j)\in [L_I, U_I]$, we call $a_j$ as a \emph{range-proper} predecessor of $a_i$(Or $a_j$ is range-proper to $a_i$).

\begin{theorem}\label{theorem:range:leftmost}
Given a sequence $\alpha$ and $\iLa$. Consider $a_i$ $\in \iLa^{t+1}$ and $a_k \in$ $\iLa^{t}$. Assuming that $a_k$ is the leftmost non-blacks item in $\iLa^{t}$ that satisfy $L_V\leq$ $a_i-a_k$ and $i-k\leq U_I$, then $a_i$ has range-proper predecessor($a_i$ should be non-black) if and only if $a_k$ is range-proper to $a_i$. 
\end{theorem}
\begin{proof}
If $a_k$ is range-proper to $a_i$, then $a_i$ has range-proper predecessor. However, if $a_i$ has range-proper predecessor, assumed as $a_j$, namely, $a_i-a_j\in [L_V, U_V]$ and $i-j\in [L_I, U_I]$. Since $a_k$ is the leftmost non-black item in $\iLa^t$ that satisfy $L_V\leq$ $a_i-a_k$ and $i-k\leq U_I$, $a_k$ is either $a_j$ or an item at the left of $a_j$, namely $a_k \geq a_j$. With Lemma \ref{lem:consecutive}(\ref{item:decreasing}), we know that $j \geq k$. Then, $i-k \geq$ $i-j\geq$ $L_I$ and $a_i-a_k$ $\leq a_i-a_j$ $\leq U_V$. Thus, $L_V\leq$ $a_i-a_k$ $\leq U_V$ and $L_I\leq$ $i-k$ $\leq U_I$, which means $a_k$ is range-proper to $a_i$.
\end{proof}

With Theorem \ref{theorem:range:leftmost}, we can see that for an item $a_i$ in $\iLa^{t+1}$, if we find out the leftmost item $a_k$ in $\iLa^{t}$ that satisfy $L_V\leq$ $a_i-a_k$ and $i-k\leq U_I$, we can easily determine whether color $a_i$ as black or not. For brevity, for item $a_i \in \iLa^{t+1}$, the leftmost non-black item $a_k$ in $\iLa^t$ where $L_V\leq$ $a_i-a_k$ and $i-k\leq U_I$ as \emph{leftmost partially-proper} item of $a_i$.

\begin{theorem}\label{theorem:range:dynamic}
Given a sequence $\alpha$ and $\iLa$. Consider $a_i$,$a_j$ $\in \iLa^{t+1}$. Assume that $a_{i'}$, $a_{j'}$ $\in \iLa^{t}$ are the leftmost partial-proper items of $a_i$ and $a_j$, respectively. Then if $a_i$ is at the left side of $a_j$, $a_{i'}$ is either $a_{j'}$ or at the left of $a_{j'}$.
\end{theorem}
\begin{proof}
Consider a non-black item $a_{k'}$ at the left of $a_{i'}$ in $\iLa^t$. With Lemma \ref{lem:consecutive}(\ref{item:decreasing}), we know that $k' < i'$, $a_i > a_j$ and $i < j$. Since $a_{i'}$ is the leftmost partial-proper items of $a_i$, namely, the leftmost non-black item satisfying $L_V\leq$ $a_i-a_{i'}$ and $i-{i'}\leq U_I$, we can know that $a_{k'}$ either $L_V > a_i-a_{k'}$ or $i-{k'} > U_I$. If $L_V > a_i-a_{k'}$, then $L_V > a_i-a_{k'}$ $> a_j-a_{k'}$. Otherwise, if $i-{k'} > U_I$, since $i < j$, then $j-{k'} > U_I$. Thus, $a_{k'}$ can not be the leftmost partial-proper item of $a_j$.
\end{proof}

With Theorem \ref{theorem:range:dynamic}, we can see that after determining the leftmost partial-proper item $a_k \in \iLa^t$ of $a_i \in$ $\iLa^{t+1}$, the leftmost partial-proper item of $rn_{\alpha}(a_i)$ can be searched from $a_k$ to the left in $\iLa^t$. Each time when we figure out the leftmost partial-proper item $a_j$ of an item $a_i$, we further check whether $a_j$ is rnage-proper to $a_i$ since $a_i$ has range-proper predecessor if and only if $a_j$ is range-proper to $a_i$. It is quite similar to the process in the coloration for RLIS(See Line \ref{code:rcolor:begin}-\ref{code:rcolor:end} in Algorithm \ref{alg:range}).

\noindent\underline{\Paragraph{Outputting an SLIS}}

This phrase is just the same as that of outputting SLIS(See Line \ref{code:rlis:begin}-\ref{code:rlis:end}).

Pseudo codes for RLIS is presented in Algorithm \ref{alg:range}.


\subsection{Comparison} \label{sec:appendix:other:constraints:comparison}
We compare our solution to these problems with previous work. We compare our method with those previous work on theoretical complexity in Section \ref{sec:appendix:theory}. Besides, we experimentally evaluate our solution to these problem with regarding to these previous work in Section \ref{sec:appendix:experiment}.

\subsubsection{Theoretical Comparison} \label{sec:appendix:theory}
Table \ref{appendix:tab:index:comp} present the theoretical comparison between our solution and previous state-of-the-art. We can see that our method is the only one that is able to support efficient update. Besides, our method is not worse than any previous work on both space complexity and query time complexity.

\begin{table}[!h]
\centering
\small
    \caption{Theoretical Comparison on Data Structure} 
    \label{appendix:tab:index:comp}

    \begin{small}
    \resizebox{1\textwidth}{!}
    {
    \begin{tabular}{|l|c|c|c|c|c|c|c|}
     \hline \multirow{2}{*}{\textbf{Methods} } &\multirow{2}{*}{\textbf{Space} }&\multirow{2}{*}{\textbf{Update} }&\multirow{2}{*}{\textbf{Construction} }&\multicolumn{4}{c|}{\textbf{Query}} \\ \cline{5-8}
     & & & & \bfseries{Max-Width} & \bfseries{Min-Width} & \bfseries{RLIS} & \bfseries{SLIS}  \\ \hline
      Our Method & $O(w)$ & $O(w)$ & $O(w\log w)$ & $O(w)$ & $O(w)$ & $O(w)$ & $O(w)$ \\ \hline
      VARIANT\cite{variant2009} & $O(w)$ & -- & $O(w\log w)$ & $O(w)$ & $O(w)$ & -- & --  \\ \hline
      RLIS\cite{fastyang2008} & $O(w)$ & -- & $O(w\log w)$ & -- & -- & $O(w)$ & --   \\ \hline
      SLIS\cite{fastyang2008} & $O(w)$ & -- & $O(w\log w)$ & -- & -- & -- & $O(w)$  \\ \hline
    \end{tabular}
    }
    \end{small}
    \vspace{-0.1in}
\end{table}

\subsubsection{Experimental Comparison} \label{sec:appendix:experiment}
The set up and data sets are exactly the same as those in Section \ref{sec:experimenteva}.Note that we set three different ranges($R1=$ \{$L_I=1$,$U_I=20$, $L_v=0$, $U_v=50$\}, $R2=$ \{$L_I=20$,$U_I=40$, $L_v=50$, $U_v=100$\}, $R3=$ \{$L_I=40$,$U_I=60$, $L_v=100$, $U_v=150$\}) and three different slopes($S1=0$, $S2=0.5$, $S3=1.0$). We implement all comparative methods in C++ according to the corresponding paper with the best of our effort. All codes are available in Github \cite{lisgit}. We can see from these experimental results(Figure \ref{fig:exp:stocka}-\ref{fig:exp:synthetica}) that our method outperform all these previous works.

\begin{figure}[h!]
\centering
\resizebox{0.9\linewidth}{!}
{
    \includegraphics{\expfolder 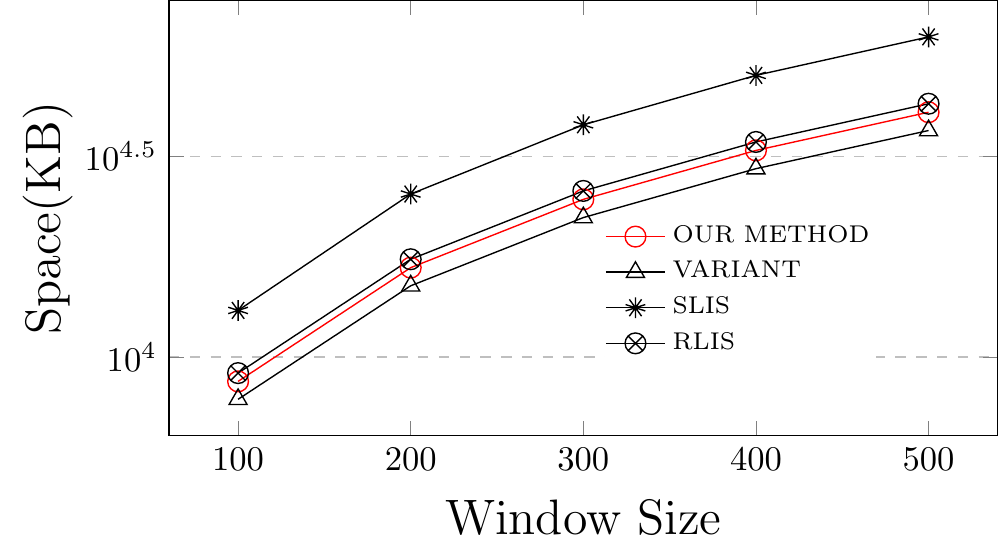}
}
\caption{Space comparison}
\label{fig:genea:construct}
\end{figure}

\begin{figure}[t!]
    \centering
    \newcommand{\figf}{\stockafolder}
    \begin{subfigure}[t]{0.45\textwidth}
        \centering
        \resizebox{\linewidth}{!}
        {
            \includegraphics{\figf 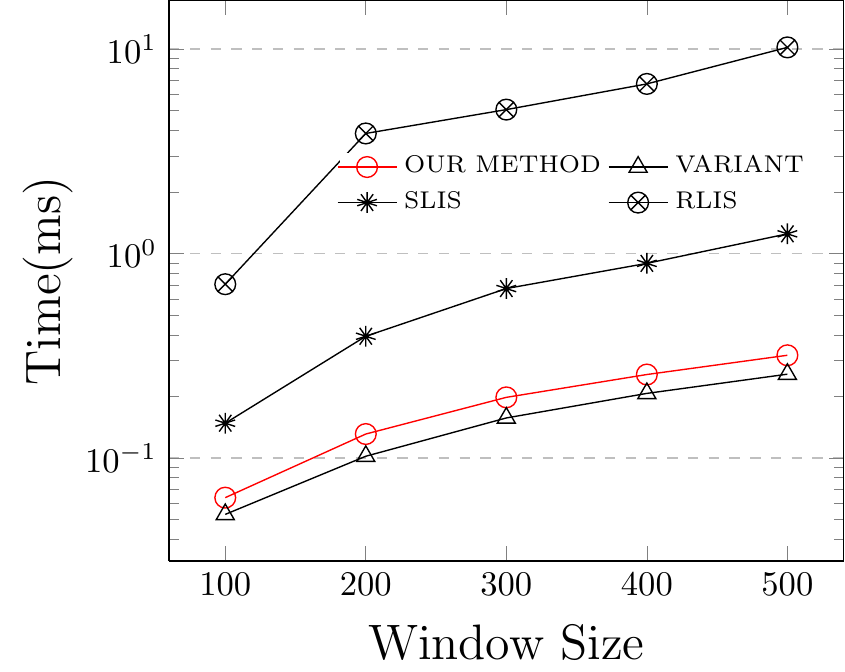}
        }
        \caption{Construction time cost}
        \label{fig:stocka:construct}
    \end{subfigure}
    \hspace{0.1in} 
    \begin{subfigure}[t]{0.45\textwidth}
        \centering
        \resizebox{\linewidth}{!}
        {
            \includegraphics{\figf 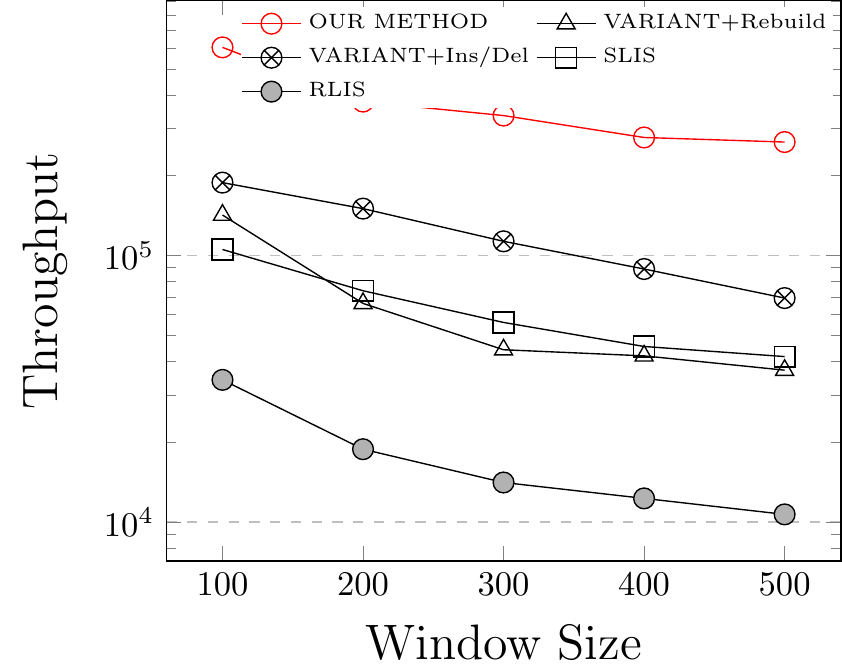}
        }
        \caption{Maintenance efficiency}
        \label{fig:stocka:update}
    \end{subfigure}
   	\hspace{0.1in} 
    \begin{subfigure}[t]{0.45\textwidth}
	    \centering
	    \resizebox{\linewidth}{!}
	    {
	        \includegraphics{\figf 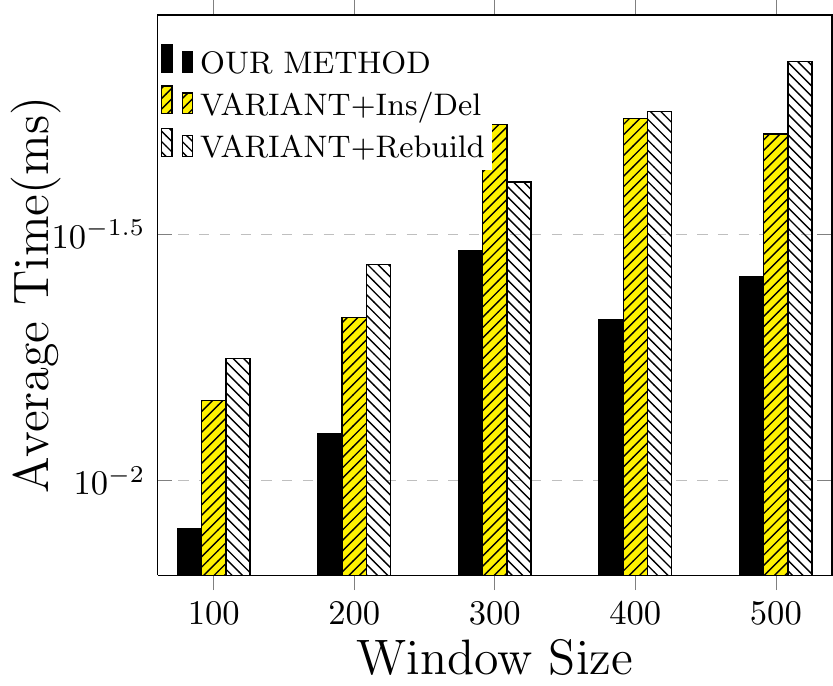}
	    }
	    \vspace{-0.1in}
	    \caption{LIS with maximum width}
	    \label{fig:stocka:maxwidth}
    \end{subfigure}
	\hspace{0.1in} 
    \begin{subfigure}[t]{0.45\textwidth}
	    \centering
	    \resizebox{\linewidth}{!}
	    {
	        \includegraphics{\figf 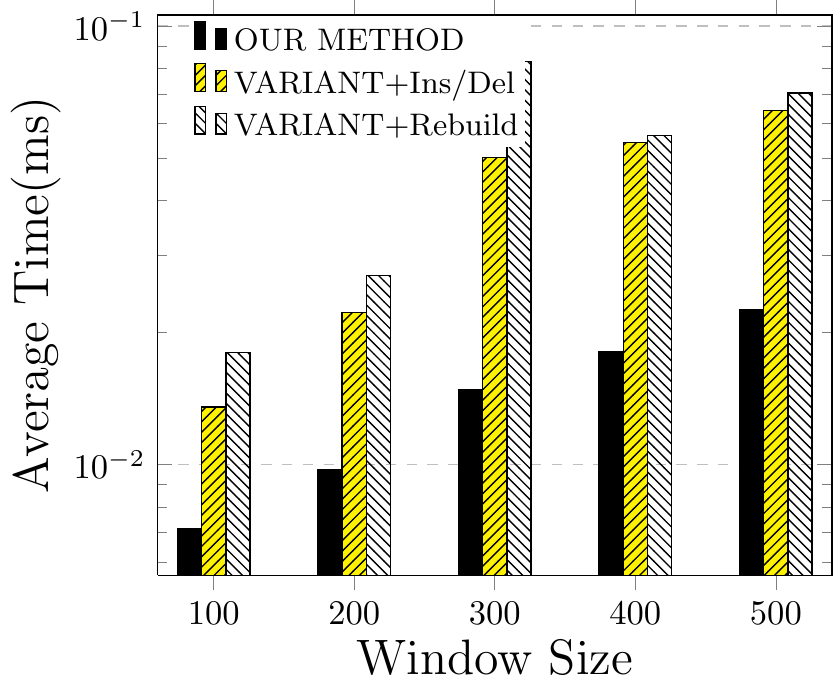}
	    }
	    \vspace{-0.1in}
	    \caption{LIS with minimum width}
	    \label{fig:stocka:minwidth}
    \end{subfigure}
    \begin{subfigure}[t]{0.45\textwidth}
	    \centering
	    \resizebox{\linewidth}{!}
	    {
	        \includegraphics{\figf 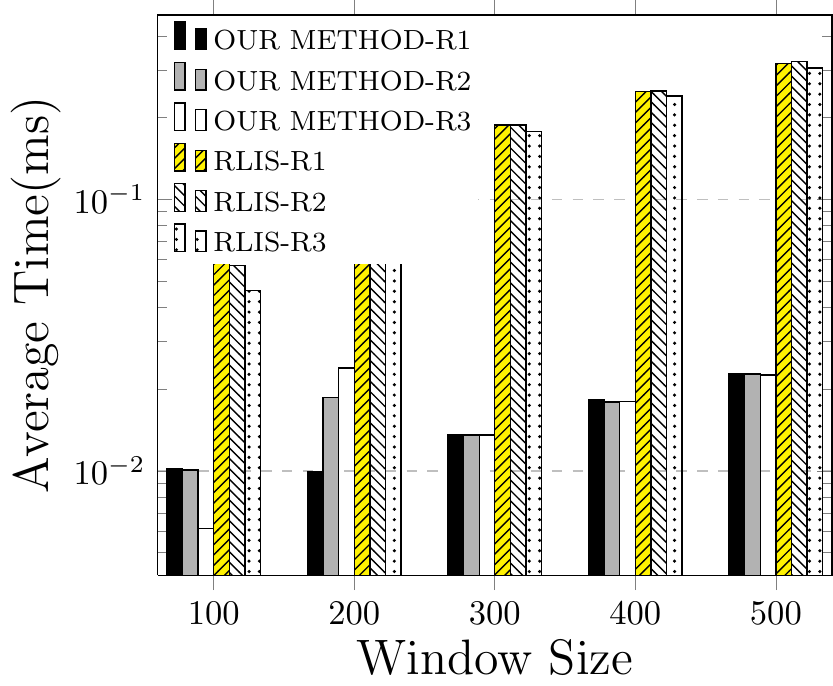}
	    }
	    \caption{Range-constrained LIS}
	    \label{fig:stocka:range}
    \end{subfigure}
	\hspace{0.1in}
    \begin{subfigure}[t]{0.45\textwidth}
	    \centering
	    \resizebox{\linewidth}{!}
	    {
	        \includegraphics{\figf 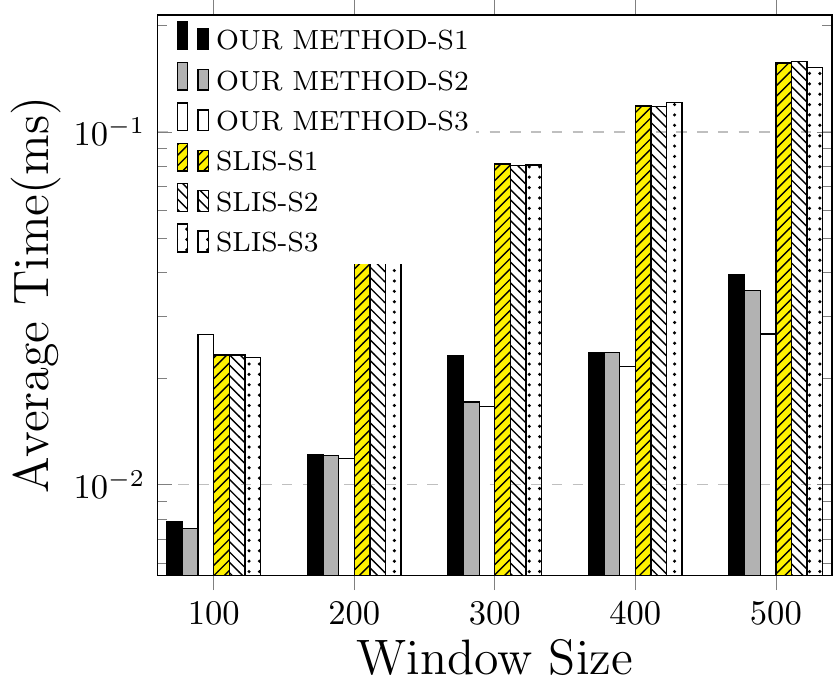}
	    }
	    \caption{Slope-constrained LIS}
	    \label{fig:stocka:slope}
    \end{subfigure}
    \vspace{-0.0in}
\caption{Evaluation on stock data}
    \vspace{-0.1in}
\label{fig:exp:stocka}
\end{figure}

\begin{figure}[t!]
    \centering
    \newcommand{\figf}{\geneafolder}
    \begin{subfigure}[t]{0.45\textwidth}
        \centering
        \resizebox{\linewidth}{!}
        {
            \includegraphics{\figf 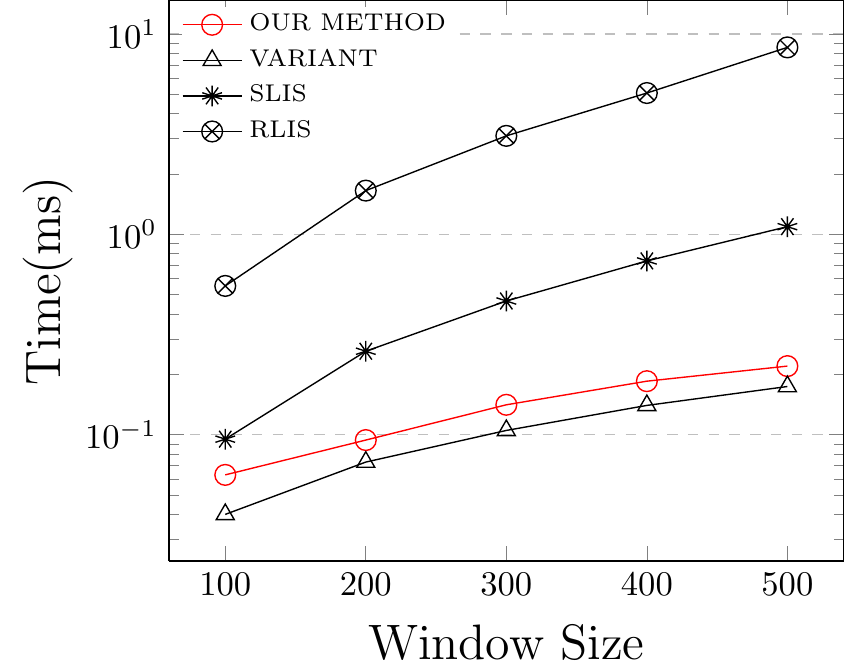}
        }
        \caption{Construction time cost}
        \label{fig:genea:construct}
    \end{subfigure}
    \hspace{0.1in} 
    \begin{subfigure}[t]{0.45\textwidth}
        \centering
        \resizebox{\linewidth}{!}
        {
            \includegraphics{\figf 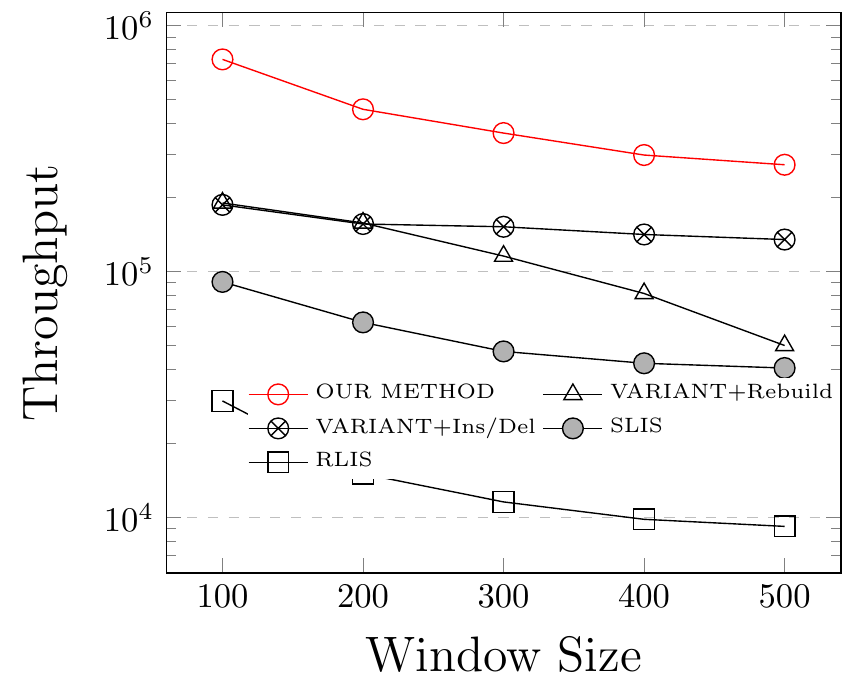}
        }
        \caption{Maintenance efficiency}
        \label{fig:genea:update}
    \end{subfigure}
   	\hspace{0.1in} 
    \begin{subfigure}[t]{0.45\textwidth}
	    \centering
	    \resizebox{\linewidth}{!}
	    {
	        \includegraphics{\figf 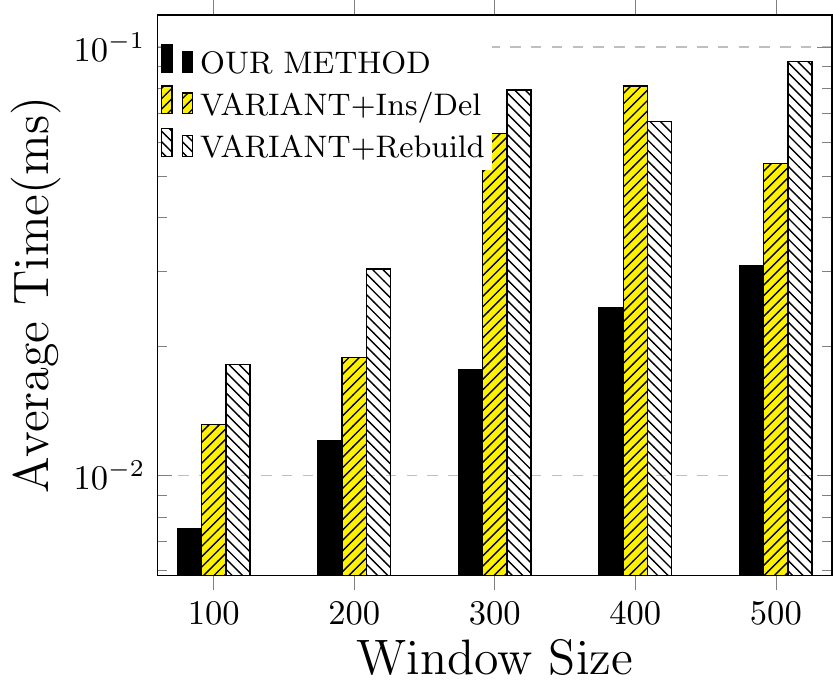}
	    }
	    \vspace{-0.1in}
	    \caption{LIS with maximum width}
	    \label{fig:genea:maxwidth}
    \end{subfigure}
	\hspace{0.1in} 
    \begin{subfigure}[t]{0.45\textwidth}
	    \centering
	    \resizebox{\linewidth}{!}
	    {
	        \includegraphics{\figf 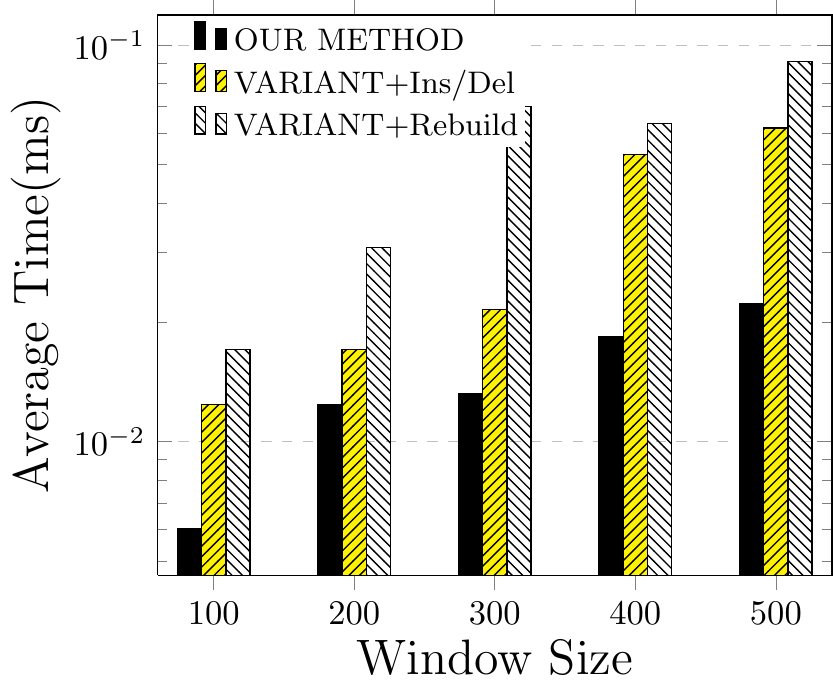}
	    }
	    \vspace{-0.1in}
	    \caption{LIS with minimum width}
	    \label{fig:genea:minwidth}
    \end{subfigure}
    \begin{subfigure}[t]{0.45\textwidth}
	    \centering
	    \resizebox{\linewidth}{!}
	    {
	        \includegraphics{\figf 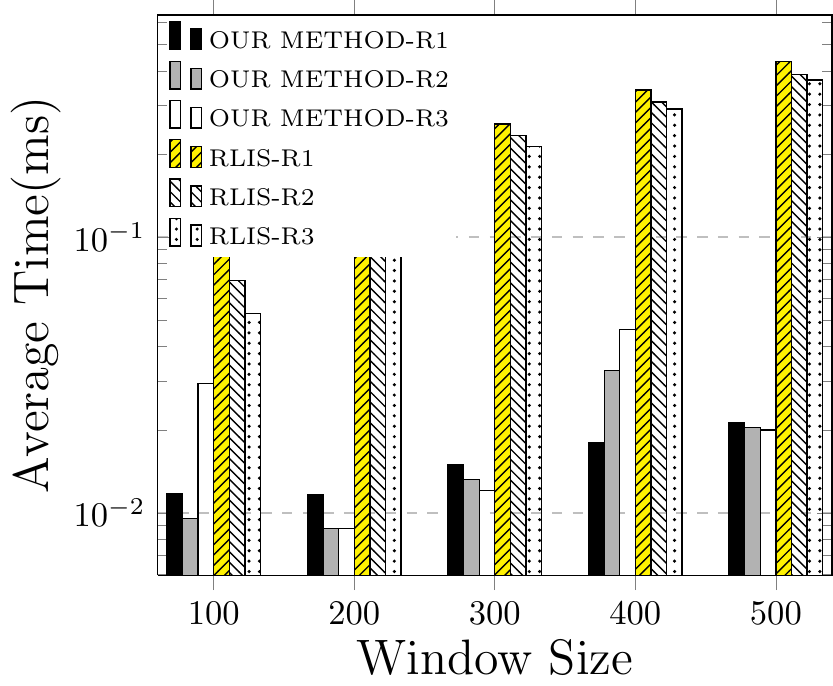}
	    }
	    \caption{Range-constrained LIS}
	    \label{fig:genea:range}
    \end{subfigure}
	\hspace{0.1in}
    \begin{subfigure}[t]{0.45\textwidth}
	    \centering
	    \resizebox{\linewidth}{!}
	    {
	        \includegraphics{\figf 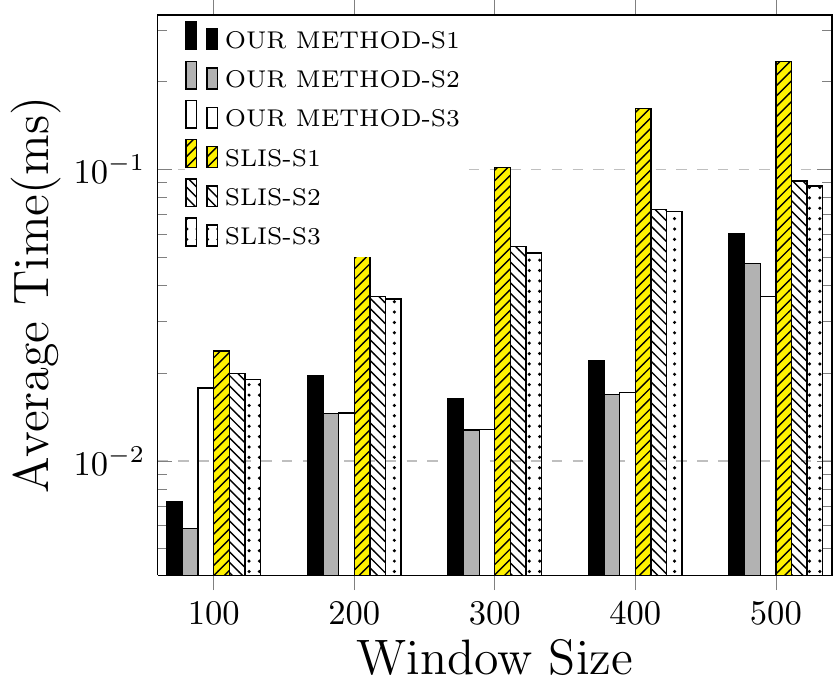}
	    }
	    \caption{Slope-constrained LIS}
	    \label{fig:genea:slope}
    \end{subfigure}
    \vspace{-0.0in}
\caption{Evaluation on gene data}
    \vspace{-0.1in}
\label{fig:exp:genea}
\end{figure}

\begin{figure}[t!]
    \centering
    \newcommand{\figf}{\powerafolder}
    \begin{subfigure}[t]{0.45\textwidth}
        \centering
        \resizebox{\linewidth}{!}
        {
            \includegraphics{\figf 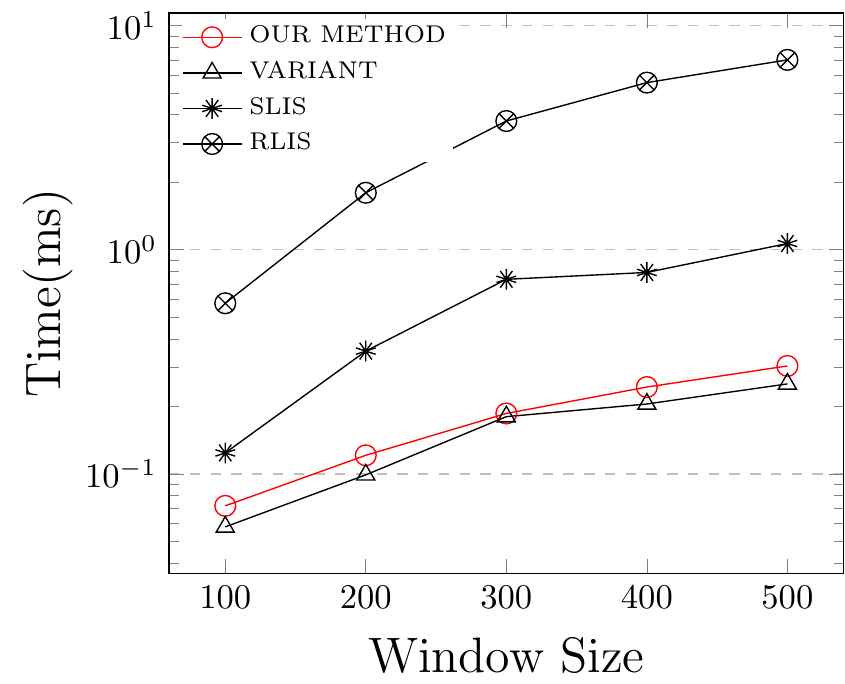}
        }
        \caption{Construction time cost}
        \label{fig:powera:construct}
    \end{subfigure}
    \hspace{0.1in} 
    \begin{subfigure}[t]{0.45\textwidth}
        \centering
        \resizebox{\linewidth}{!}
        {
            \includegraphics{\figf 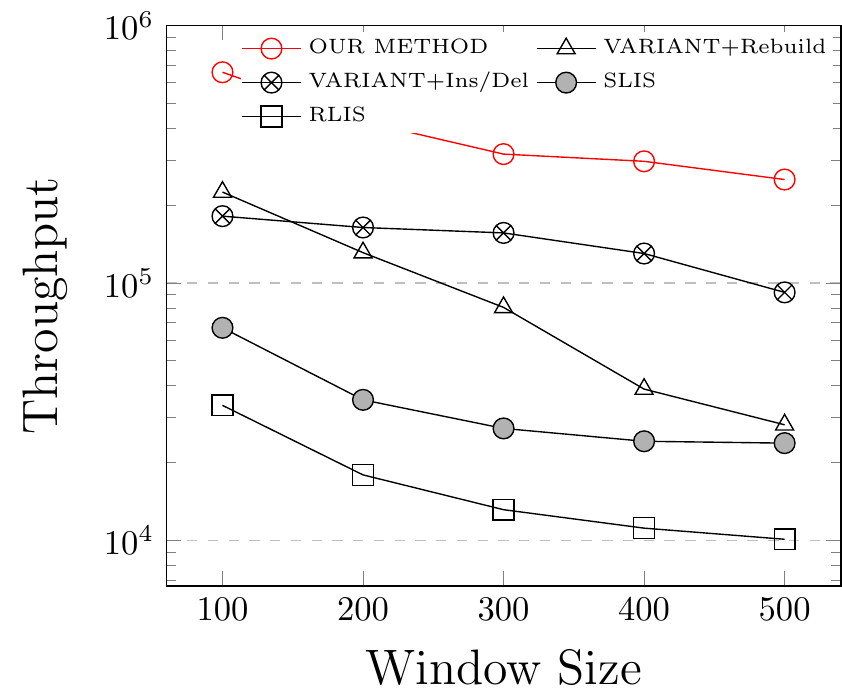}
        }
        \caption{Maintenance efficiency}
        \label{fig:powera:update}
    \end{subfigure}
   	\hspace{0.1in} 
    \begin{subfigure}[t]{0.45\textwidth}
	    \centering
	    \resizebox{\linewidth}{!}
	    {
	        \includegraphics{\figf 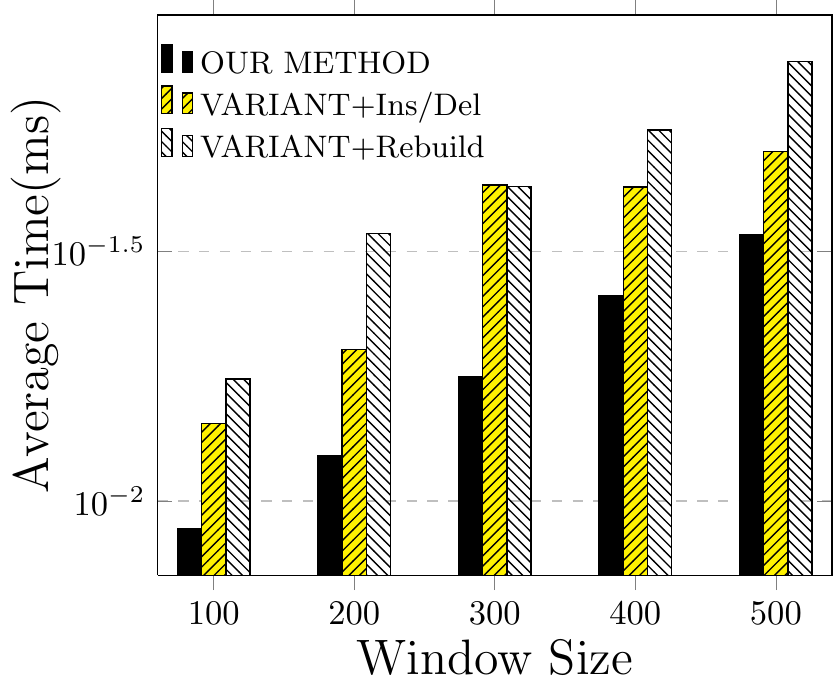}
	    }
	    \vspace{-0.1in}
	    \caption{LIS with maximum width}
	    \label{fig:powera:maxwidth}
    \end{subfigure}
	\hspace{0.1in} 
    \begin{subfigure}[t]{0.45\textwidth}
	    \centering
	    \resizebox{\linewidth}{!}
	    {
	        \includegraphics{\figf 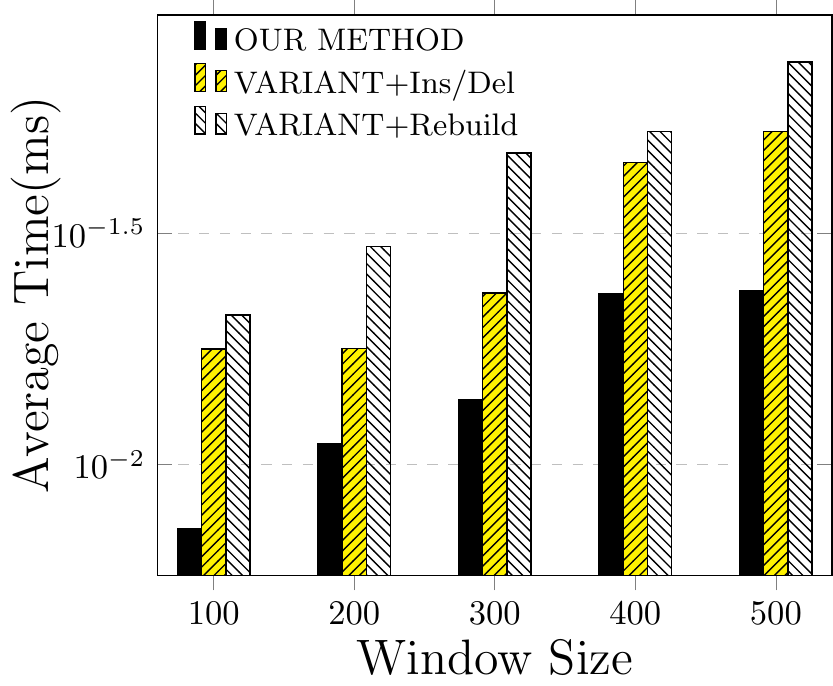}
	    }
	    \vspace{-0.1in}
	    \caption{LIS with minimum width}
	    \label{fig:powera:minwidth}
    \end{subfigure}
    \begin{subfigure}[t]{0.45\textwidth}
	    \centering
	    \resizebox{\linewidth}{!}
	    {
	        \includegraphics{\figf 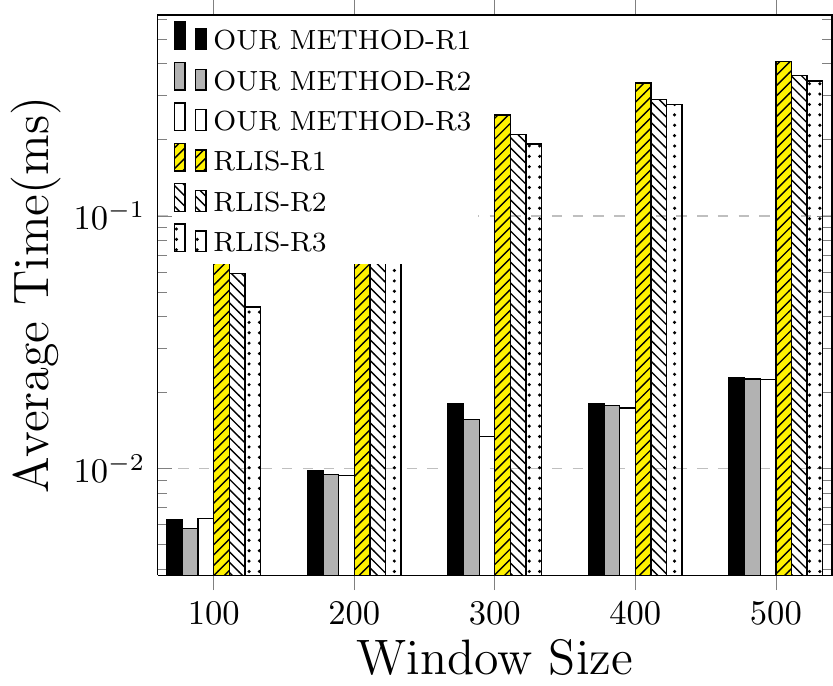}
	    }
	    \caption{Range-constrained LIS}
	    \label{fig:powera:range}
    \end{subfigure}
	\hspace{0.1in}
    \begin{subfigure}[t]{0.45\textwidth}
	    \centering
	    \resizebox{\linewidth}{!}
	    {
	        \includegraphics{\figf 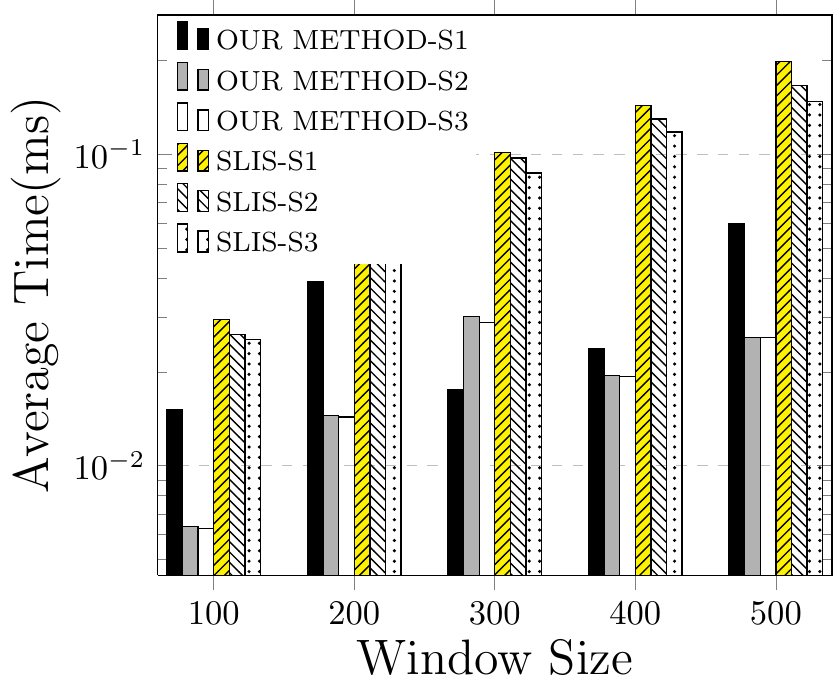}
	    }
	    \caption{Slope-constrained LIS}
	    \label{fig:powera:slope}
    \end{subfigure}
    \vspace{-0.0in}
\caption{Evaluation on power data}
    \vspace{-0.1in}
\label{fig:exp:powera}
\end{figure}

\begin{figure}[t!]
    \centering
    \newcommand{\figf}{\syntheticafolder}
    \begin{subfigure}[t]{0.45\textwidth}
        \centering
        \resizebox{\linewidth}{!}
        {
            \includegraphics{\figf 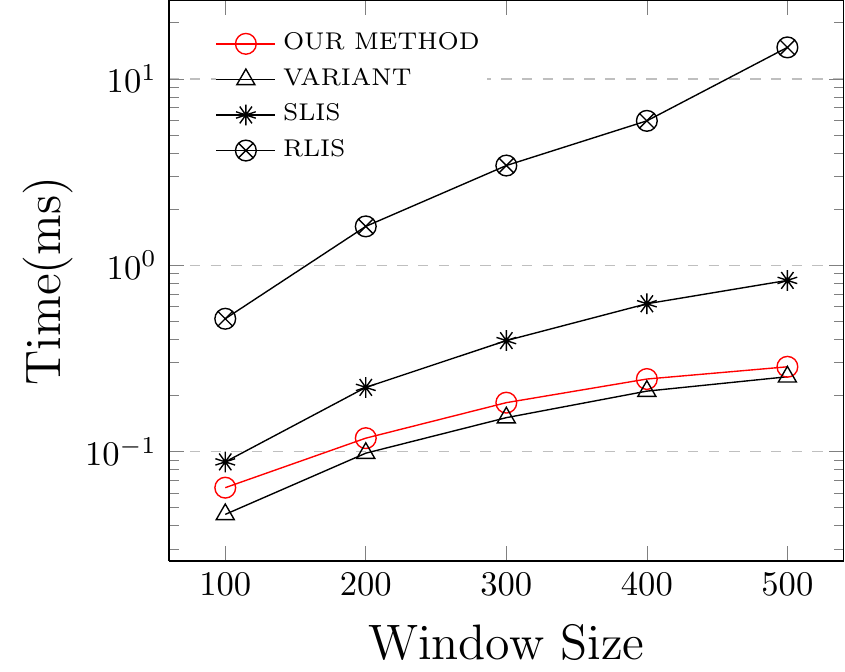}
        }
        \caption{Construction time cost}
        \label{fig:synthetica:construct}
    \end{subfigure}
    \hspace{0.1in} 
    \begin{subfigure}[t]{0.45\textwidth}
        \centering
        \resizebox{\linewidth}{!}
        {
            \includegraphics{\figf 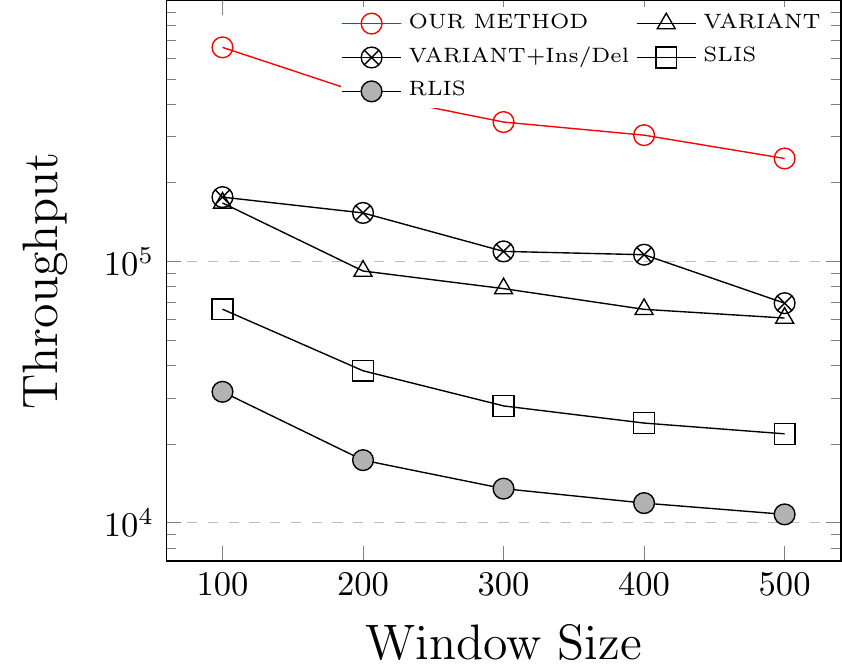}
        }
        \caption{Maintenance efficiency}
        \label{fig:synthetica:update}
    \end{subfigure}
   	\hspace{0.1in} 
    \begin{subfigure}[t]{0.45\textwidth}
	    \centering
	    \resizebox{\linewidth}{!}
	    {
	        \includegraphics{\figf 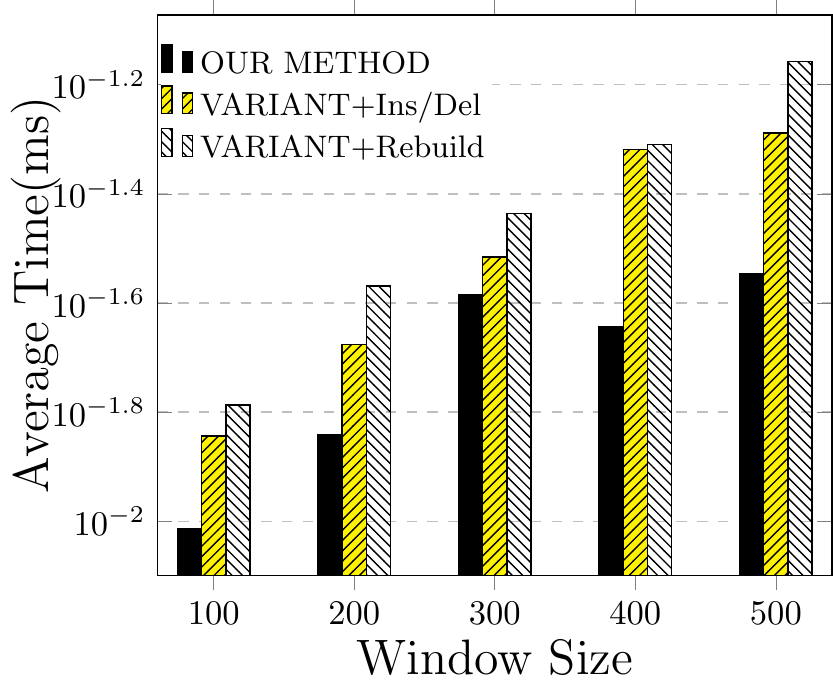}
	    }
	    \vspace{-0.1in}
	    \caption{LIS with maximum width}
	    \label{fig:synthetica:maxwidth}
    \end{subfigure}
	\hspace{0.1in} 
    \begin{subfigure}[t]{0.45\textwidth}
	    \centering
	    \resizebox{\linewidth}{!}
	    {
	        \includegraphics{\figf 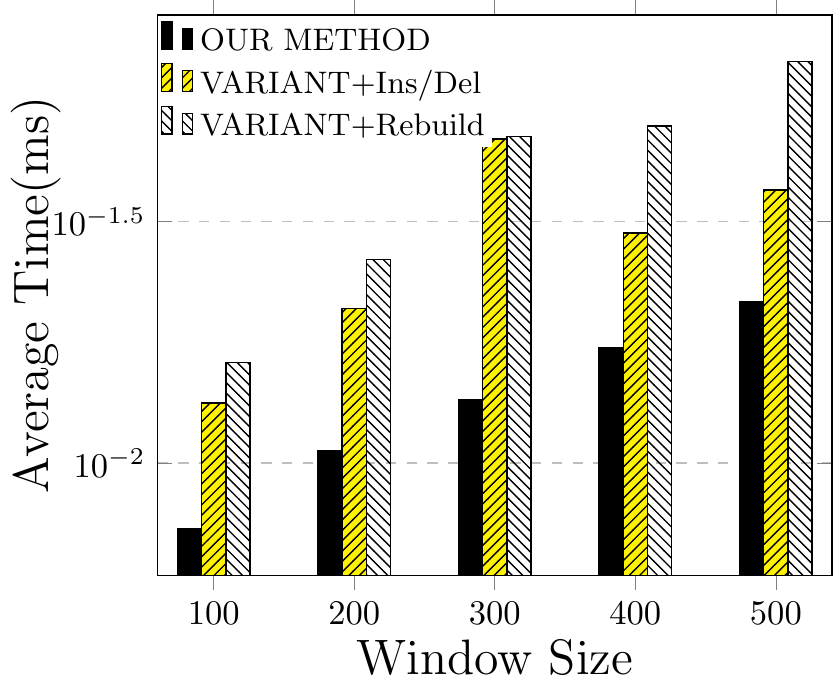}
	    }
	    \vspace{-0.1in}
	    \caption{LIS with minimum width}
	    \label{fig:synthetica:minwidth}
    \end{subfigure}
    \begin{subfigure}[t]{0.45\textwidth}
	    \centering
	    \resizebox{\linewidth}{!}
	    {
	        \includegraphics{\figf 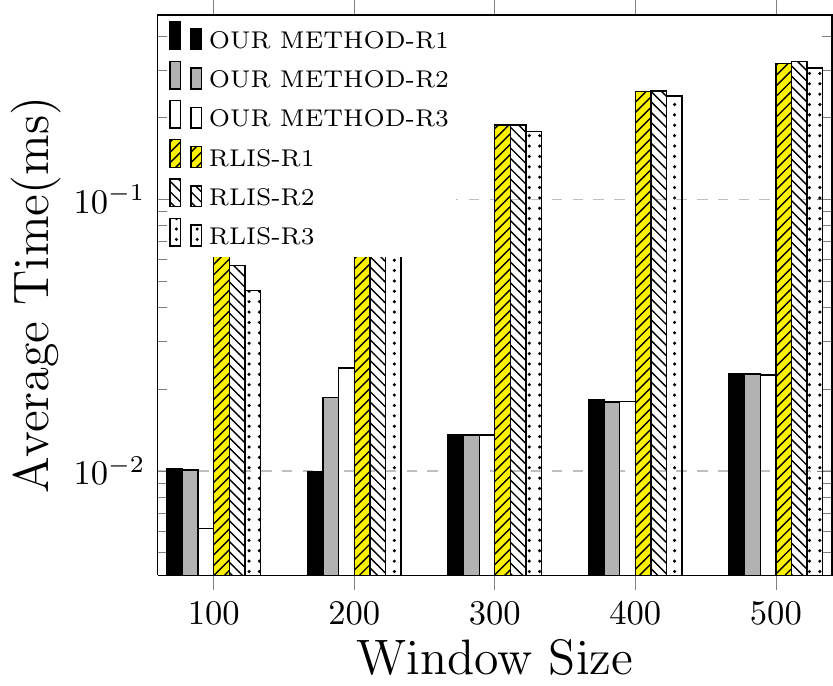}
	    }
	    \caption{Range-constrained LIS}
	    \label{fig:synthetica:range}
    \end{subfigure}
	\hspace{0.1in}
    \begin{subfigure}[t]{0.45\textwidth}
	    \centering
	    \resizebox{\linewidth}{!}
	    {
	        \includegraphics{\figf 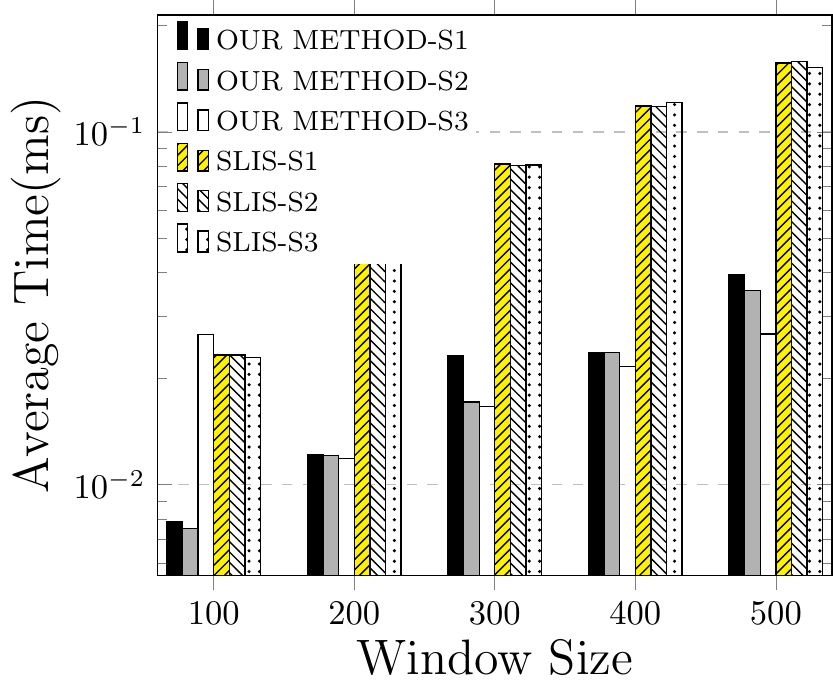}
	    }
	    \caption{Slope-constrained LIS}
	    \label{fig:synthetica:slope}
    \end{subfigure}
    \vspace{-0.0in}
\caption{Evaluation on synthetic data}
    \vspace{-0.1in}
\label{fig:exp:synthetica}
\end{figure}

\end{document}